%% file: ConditioningVandermonde_arxiv.tex
\documentclass[preprint,12pt]{elsarticle}

\usepackage[top=2.5cm, bottom=3.5cm, left=3.5cm, right=3.5cm]{geometry}
\usepackage{setspace, dsfont, enumerate, xspace}
\usepackage{amssymb, amsthm, amsmath, amsfonts, mathrsfs, stmaryrd, upgreek, mathtools}
\usepackage[scr=boondoxo]{mathalfa}
\usepackage{pgfplots, caption, tikz, subfigure, graphicx}
\usepackage{url, hyperref}
\usepackage{titlesec}
\usepackage{ifthen}

\usepackage{tikz,pgfplots}
\usetikzlibrary{matrix,patterns,calc,decorations.fractals,arrows,external,pgfplots.colormaps}
\usepgfplotslibrary{fillbetween}
\usepgfplotslibrary{patchplots}
\usepgfplotslibrary{external}
\pgfplotsset{compat=newest}

\allowdisplaybreaks

\journal{Applied and Computational Harmonic Analysis}

\input{def}

\newcommand{\versionColor}{true}

\newcommand{\pdfFigs}{true}

\newcommand{\generateFigs}{no}

\ifthenelse{\equal{\generateFigs}{yes}}{
	\renewcommand{\pdfFigs}{false}
	\usepgfplotslibrary{external}
	\tikzexternalize[prefix=figs/]
}{}

\begin{document}

\begin{frontmatter}



\title{Vandermonde Matrices with Nodes in the Unit Disk \\ and the Large Sieve}


\author{C\'eline Aubel and Helmut B\"olcskei}

\address{Dept.~IT \& EE, ETH Zurich, Switzerland \vspace{-0.7cm}}

\begin{abstract}
	\input{abstract}

\end{abstract}

\begin{keyword}
	Vandermonde matrices\sep extremal singular values\sep condition number\sep unit disk \sep large sieve\sep Hilbert's inequality
	
	\vspace{0.1cm}
	
	\MSC[2010] 15A12\sep 65F35
\end{keyword}

%
%
%

\end{frontmatter}


\input{introduction}

\input{notation}

\input{problemStatement}

\input{previousWork}

\input{mainResults}

\input{numericalResults}

\appendix
\input{appendix1.tex}
\input{appendix2.tex}
\input{appendix3.tex}

\input{appendix4.tex}

\section*{Acknowledgments}
The authors would like to thank E.~Riegler for interesting discussions. 




\bibliographystyle{elsarticle-num-names} 
\biboptions{square,sort,comma,numbers}
\bibliography{refs}

\end{document}

%% file: abstract.tex

We derive bounds on the extremal singular values and the condition number of $\nbSamples \times \nbNodes$, with $\nbSamples \geq \nbNodes$, Vandermonde matrices with nodes in the unit disk. 
The mathematical techniques we develop to prove our main results are inspired by a link---first established by \citet{Selberg1991} and later extended by \citet{Moitra2015}---between the extremal singular values of Vandermonde matrices with nodes on the unit circle and large sieve inequalities. Our main conceptual contribution lies in establishing a connection between the extremal singular values of Vandermonde matrices with nodes in the unit disk and a novel large sieve inequality involving polynomials in $z \in \C$ with $\abs{z} \leq 1$. Compared to Baz\'an's upper bound on the condition number \cite{Bazan2000}, which, to the best of our knowledge, constitutes the only analytical result---available in the literature---on the condition number of Vandermonde matrices with nodes in the unit disk, our bound not only takes a much simpler form, but is also sharper for certain node configurations. Moreover, the bound we obtain can be evaluated consistently in a numerically stable fashion, whereas the evaluation of Baz\'an's bound requires the solution of a linear system of equations which has the same condition number as the Vandermonde matrix under consideration and can therefore lead to numerical instability in practice. As a byproduct, our result---when particularized to the case of nodes on the unit circle---slightly improves upon the Selberg--Moitra bound.

%% file: introduction.tex

\section{Introduction}

Vandermonde matrices arise in many fields of applied mathematics and engineering such as interpolation and approximation theory~\citep{Bjorck1973, Heinig1995}, differential equations~\citep{Luther2004}, control theory~\citep{Pantelous2013}, sampling theory~\citep{Groechenig1999, Vetterli2002, Feng1996, Mishali2009}, subspace methods for parameter estimation~\citep{Schmidt1986, Roy1986, Hua1990, Liao2014, Moitra2015, Potts2017}, line spectral estimation~\citep{Tang2013}, and fast evaluation of linear combinations of radial basis functions using the fast Fourier transform for non-equispaced knots~\citep{Potts2006, Kunis2006}. 

It is well known that the condition number of real square Vandermonde matrices grows exponentially in the dimension of the matrix~\citep{Gautschi1988, Beckermann2000}. Complex Vandermonde matrices, on the other hand, can be well-conditioned depending on the locations of the nodes in the complex plane. 
There exists significant literature on the condition number of Vandermonde matrices with nodes on the unit circle. Specifically, it is shown in \citep{Gautschi1990, Berman2007, Ferreira99, Bazan2000, Liao2014, Moitra2014} that $\nbSamples \times \nbNodes$, with $\nbSamples \geq \nbNodes$, Vandermonde matrices with nodes $e^{2\pi i\xi_k}$, where $\xi_k \in [0, 1)$, for $k \in \{1, 2, \ldots, \nbNodes\}$, are well-conditioned provided that the minimum wrap-around distance between the node frequencies $\xi_k$ is large enough. 
On the other hand, the literature on $\nbSamples \times \nbNodes$, with $\nbSamples \geq \nbNodes$, Vandermonde matrices with nodes $\node_k$ in the unit disk, i.e., $\abs{\node_k} \leq 1$, for $k \in \{1, 2, \ldots, \nbNodes\}$,  is very scarce. In fact, the only result along these lines that we are aware of is Baz\'an's upper bound on the spectral condition number \citep{Bazan2000}. This bound is, however, implicit as it depends on a quantity whose computation requires the solution of the linear system of equations generated by the Vandermonde matrix under consideration. As the numerical results in Section \ref{sec: numerical results} demonstrate, the evaluation of this bound can therefore be numerically unstable in practice.

\textit{Contributions.} We derive a lower bound on the minimum singular value and an upper bound on the maximum singular value of $\nbSamples \times \nbNodes$ ($\nbSamples \geq \nbNodes$) Vandermonde matrices with general nodes $\node_k = \abs{\node_k}\!e^{2\pi i\xi_k}$ in the unit disk, i.e., $\abs{\node_k} \leq 1$ and $\xi_k \in [0, 1)$, for $k \in \{1, 2, \ldots, \nbNodes\}$. Based on these bounds we get an upper bound on the spectral condition number.
Our bounds depend on $\nbSamples$, the minimum wrap-around distance between the $\xi_k$, and the moduli $\abs{\node_k}$ of the nodes. In particular, the upper bound on the spectral condition number we report is of much simpler form than Baz\'an's bound, and for certain node configurations also sharper. 
The mathematical techniques we develop to prove our main results are inspired by a link---first established by \citet{Selberg1991} and later extended by \citet{Moitra2015}---between the extremal singular values of Vandermonde matrices with nodes on the unit circle and large sieve inequalities~\citep{Linnik1941, Renyi1951, Roth1964, Roth1965, Bombieri1965, Montgomery1968}. The Selberg-Moitra approach employs Fourier-analytic techniques and the Poisson summation formula and therefore does not seem to be amenable to an extension to the case of nodes in the unit disk. 
Our main conceptual contribution lies in establishing a connection between the extremal singular values of Vandermonde matrices with nodes in the unit disk and a novel large sieve inequality involving polynomials in $z \in \C$ with $\abs{z} \leq 1$. This is accomplished by first recognizing that the Selberg--Moitra connection can alternatively be established based on the Montgomery--Vaughan proof \cite{Montgomery1973} of the large sieve inequality, and then extending this alternative connection from the unit circle to the unit disk. 
We also demonstrate how Cohen's dilatation trick, described in \citep[p.~559]{Montgomery1978} and originally developed for the large sieve inequality on the unit circle, can be applied to refine our bounds valid for nodes in the unit disk. 
As a byproduct, our result---when particularized to the unit circle---slightly improves upon the Selberg--Moitra upper bound. This improved bound also applies to the square case, $\nbSamples = \nbNodes$, not covered by the \mbox{Selberg--Moitra result.} 

The numerical evaluation of Baz\'an's bound requires the solution of a linear system of equations which has the same condition number as the Vandermonde matrix under consideration; this can lead to numerical instability in practice. We provide numerical results demonstrating that our bound can not only be evaluated consistently in a numerically stable fashion, but is, in certain cases, also tighter than Baz\'an's bound.

%% file: notation.tex

\textit{Notation. }
The complex conjugate of $z \in \C$ is denoted by~$\overline{z}$. The hyperbolic sine function is defined as $\sh(z) \triangleq (e^z - e^{-z})/2$, for $z \in \C$. For $x \in \R$, $\lfloor x\rfloor$ is the largest integer not exceeding $x$, $\lceil x \rceil$ stands for the smallest integer larger than $x$, and $[x]$ denotes the integer closest to $x$.
Lowercase boldface letters designate (column) vectors and uppercase boldface letters denote matrices. The superscripts $^T$ and $^H$ refer to transposition and Hermitian transposition, respectively. 
For a vector $\boldsymbol{x} \triangleq \{x_k\}_{k = 1}^K \in \C^K$, we write $\norm{\boldsymbol{x}}_p$ for its $\ell^p$-norm, $p \in [1, \infty]$, that is, $\norm{\boldsymbol{x}}_p \triangleq \big(\sum_{k = 1}^K \abs{x_k}^p\big)^{1/p}$, for $p \in [1, \infty)$, and $\norm{\boldsymbol{x}}_\infty \triangleq \max_{1 \leq k \leq K} \abs{x_k}$.
The Moore-Penrose pseudo-inverse of the full-rank matrix $\mathbf{A} \in \C^{M \times N}$ is $\mathbf{A}^\dagger = \mathbf{A}^H(\mathbf{A}\mathbf{A}^H)^{-1}$, if $M < N$, and $\mathbf{A}^\dagger = (\mathbf{A}^H\mathbf{A})^{-1}\mathbf{A}^H$, if $M \geq N$.
We denote the smallest and largest singular value of $\mathbf{A} \in \C^{M \times N}$ by $\sigma_\mathrm{min}(\mathbf{A})$ and $\sigma_\mathrm{max}(\mathbf{A})$, respectively, and for $p \in [1, \infty]$, we let $\norm{\mathbf{A}}_p \triangleq \max\{\norm{\mathbf{A}\boldsymbol{x}}_p \colon \boldsymbol{x} \in \C^N, \norm{\boldsymbol{x}}_p = 1\}$.
In particular, we have $\norm{\mathbf{A}}_2 = \sigma_\mathrm{max}(\mathbf{A})$ and $\norm{\mathbf{A}}_\infty = \max_{1 \leq m \leq M}\sum_{n = 1}^N \abs{a_{m, n}}$.
For $\mathbf{A} \in \C^{M \times N}$ with columns $\boldsymbol{a}_n$, $n \in \{1, 2, \ldots, N\}$, we let $\vectorize(\mathbf{A}) \triangleq \left(\boldsymbol{a}_1^T\ \boldsymbol{a}_2^T\ \ldots\ \boldsymbol{a}_N^T\right)^T$, and  $\overline{\mathbf{A}}$ denotes the matrix obtained by element-wise complex conjugation of $\mathbf{A}$.

%% file: problemStatement.tex

\section{Problem statement}

We consider Vandermonde matrices of the form
\begin{equation*}
	\vandermondeMat{\nbSamples \times \nbNodes} \triangleq \begin{pmatrix} 1 & 1 & \ldots & 1 & 1\\
			\node_1 & \node_2 & \ldots & \node_{\nbNodes-1} & \node_\nbNodes \\
			\node_1^2 & \node_2^2 & \ldots & \node_{\nbNodes-1}^2 & \node_\nbNodes^2 \\
			\vdots & \vdots & \ddots & \vdots & \vdots \\
			\node_1^{\nbSamples-1} & \node_2^{\nbSamples-1} & \ldots & \node_{\nbNodes-1}^{\nbSamples-1} & \node_\nbNodes^{\nbSamples-1}
		\end{pmatrix} 
	\in \C^{\nbSamples \times \nbNodes},
\end{equation*}
where $\nbSamples \geq \nbNodes$, and $\node_1, \node_2, \ldots, \node_\nbNodes \in \C$ are referred to as the ``nodes'' of $\vandermondeMat{\nbSamples \times \nbNodes}$. Throughout the paper, we take the nodes to be non-zero and pairwise distinct, i.e., $\node_{k_1} \neq \node_{k_2}$, for $k_1 \neq k_2$, which ensures that the matrix $\vandermondeMat{\nbSamples \times \nbNodes}$ has full rank.

We shall be interested in the minimum and maximum singular values and the condition number of $\vandermondeMat{\nbSamples \times \nbNodes}$ with respect to the general matrix norm $\big\|\cdot\big\|$ defined~\citep[Eq.~5.3.7]{Golub1983} as 
\begin{equation*}
	\kappa\left(\vandermondeMat{\nbSamples \times \nbNodes}, \big\|\cdot\big\|\right) \triangleq \big\|\vandermondeMat{\nbSamples \times \nbNodes}\big\|\big\|(\vandermondeMat{\nbSamples \times \nbNodes})^\dagger\big\|.
\end{equation*}
We will mostly be concerned with 
\begin{equation}
	\kappa(\vandermondeMat{\nbSamples \times \nbNodes}) \triangleq \kappa(\vandermondeMat{\nbSamples \times \nbNodes}, \norm{\cdot}_{2,2}) = \frac{\sigma_{\mathrm{max}}(\vandermondeMat{\nbSamples \times \nbNodes})}{\sigma_{\mathrm{min}}(\vandermondeMat{\nbSamples \times \nbNodes})},
	\label{eq: definition spectral condition number}
\end{equation}
often referred to as ``spectral condition number''. 

The goal of this paper is to find lower bounds on the minimum singular value and upper bounds on the maximum singular value of Vandermonde matrices $\vandermondeMat{\nbSamples \times \nbNodes}$ with nodes in the unit disk, that is, $\abs{\node_k} \leq 1$, for $k \in \{1, 2, \ldots, \nbNodes\}$. Based on these bounds, we then establish upper bounds on $\kappa(\vandermondeMat{\nbSamples \times \nbNodes})$.

%% file: previousWork.tex

\section{Previous work}

Before stating our main results, we summarize relevant prior work.

\subsection{Vandermonde matrices with real nodes}
\citet[Thms.~2.2 and 3.1]{Gautschi1988} showed that the condition number $\kappa(\vandermondeMat{\nbNodes \times \nbNodes}, \norm{\cdot}_\infty)$, $\nbNodes \geq 3$, is lower-bounded by $(\nbNodes - 1)2^{\nbNodes}$ when $\node_1, \node_2, \ldots, \node_\nbNodes \in \R_+$ and by $2^{\nbNodes/2}$ when $\nbNodes = 2L$, $L \in \N$, and the nodes $\node_1, \node_2, \ldots, \node_\nbNodes \in \R \!\setminus\! \{0\}$ satisfy the symmetry relationship $\node_{k+L} = -\node_k$, for $k \in \{1, 2, \ldots, L\}$. \citet[Thm.~4.1]{Beckermann2000} found that the spectral condition number of $\vandermondeMat{\nbNodes \times \nbNodes}$ satisfies 
\begin{equation*}
	\frac{\sqrt{2}(1 + \sqrt{2})^{\nbNodes-1}}{\sqrt{\nbNodes+1}}\leq \kappa\!\left(\vandermondeMat{\nbNodes \times \nbNodes}\right) \leq (\nbNodes+1)\sqrt{2}(1 + \sqrt{2})^{\nbNodes-1}, 
\end{equation*}
for $\node_1, \node_2, \ldots, \node_\nbNodes \in \R \!\setminus\! \{0\}$, and 
\begin{equation*}
	\frac{C_K}{2(\nbNodes+1)} \leq \kappa\!\left(\vandermondeMat{\nbNodes \times \nbNodes}\right) \leq \frac{\nbNodes+1}{2}C_K,
\end{equation*}
for $\node_1, \node_2, \ldots, \node_\nbNodes \in \R_+$, where $C_K \triangleq (1 + \sqrt{2})^{2\nbNodes} + (1 + \sqrt{2})^{-2\nbNodes}$. 
These results show that square Vandermonde matrices $\vandermondeMat{\nbNodes \times \nbNodes}$ with real nodes necessarily become ill-conditioned as the matrix dimension grows. Specifically, the condition number grows exponentially in the matrix dimension and, in particular, does so independently of the specific values of the nodes $\node_1, \node_2, \ldots, \node_\nbNodes$.

\subsection{Vandermonde matrices with complex nodes}
For Vandermonde matrices with complex nodes the situation is fundamentally different.
Consider, e.g., the DFT matrix $\mathbf{F}_\nbNodes \triangleq \{e^{2\pi ik\ell/\nbNodes}\}_{0 \leq k, \ell \leq \nbNodes-1}$, which is a Vandermonde matrix with nodes $z_k = e^{2\pi i (k-1)/\nbNodes}$, for $k \in \{1, 2, \ldots, \nbNodes\}$, and, as a consequence of  $\mathbf{F}_\nbNodes^H\mathbf{F}_\nbNodes = \nbNodes\mathbf{I}_K$, has the smallest possible spectral condition number, namely, $\kappa(\mathbf{F}_K) = 1$, and this, irrespectively of the matrix dimension $\nbNodes$. 

For general nodes $\node_1, \node_2, \ldots, \node_\nbNodes \in \C$, \citet[Thms.~1 and 3.1]{Gautschi1962} obtained the following bounds on $\norm{(\vandermondeMat{\nbNodes \times \nbNodes})^{-1}}_\infty$:
\begin{equation}
	\max_{1 \leq k \leq \nbNodes} \prod_{\substack{\ell = 1\\ \ell \neq k}}^\nbNodes \frac{\max\{1, \abs{z_{\ell}}\}}{\abs{z_{k} - z_{\ell}}} \leq \norm{(\vandermonde{\nbNodes \times \nbNodes}{\node_1, \node_2, \ldots, \node_\nbNodes})^{-1}}_\infty \leq \max_{1 \leq k \leq \nbNodes} \prod_{\substack{\ell = 1\\ \ell \neq k}}^\nbNodes \frac{1 + \abs{z_{k}}}{\abs{z_{k} - z_{\ell}}}.
	\label{eq: Gautschi's bound general norm infty inv} 
\end{equation}
This allows us to derive bounds on $\sigma_\mathrm{min}(\vandermonde{\nbNodes \times \nbNodes}{\node_1, \node_2, \ldots, \node_\nbNodes})$ and $\kappa\!\left(\vandermondeMat{\nbNodes \times \nbNodes}, \norm{\cdot}_\infty\right)$ by noting that $\norm{(\vandermonde{\nbNodes \times \nbNodes}{\node_1, \node_2, \ldots, \node_\nbNodes})^{-1}}_2 = \sigma_\mathrm{min}(\vandermonde{\nbNodes \times \nbNodes}{\node_1, \node_2, \ldots, \node_\nbNodes})$ and 
\begin{equation*}
	\frac{\norm{(\vandermonde{\nbNodes \times \nbNodes}{\node_1, \node_2, \ldots, \node_\nbNodes})^{-1}}_\infty}{\sqrt{\nbNodes}} \leq \norm{(\vandermonde{\nbNodes \times \nbNodes}{\node_1, \node_2, \ldots, \node_\nbNodes})^{-1}}_2 \leq \sqrt{\nbNodes}\norm{(\vandermonde{\nbNodes \times \nbNodes}{\node_1, \node_2, \ldots, \node_\nbNodes})^{-1}}_\infty.
\end{equation*}
Specifically, this results in
\begin{equation*}
	\frac{1}{\sqrt{\nbNodes}}\max_{1 \leq k \leq \nbNodes} \prod_{\substack{\ell = 1\\ \ell \neq k}}^\nbNodes \frac{\max\{1, \abs{z_{\ell}}\}}{\abs{z_{k} - z_{\ell}}} \leq \sigma_\mathrm{min}(\vandermonde{\nbNodes \times \nbNodes}{\node_1, \node_2, \ldots, \node_\nbNodes}) \leq \sqrt{\nbNodes}\max_{1 \leq k \leq \nbNodes} \prod_{\substack{\ell = 1\\ \ell \neq k}}^\nbNodes \frac{1 + \abs{z_{k}}}{\abs{z_{k} - z_{\ell}}}.
\end{equation*}
Combining \eqref{eq: Gautschi's bound general norm infty inv} with
\begin{equation*}
	\norm{\vandermondeMat{\nbNodes \times \nbNodes}}_\infty = \max_{0 \leq n \leq \nbNodes-1} \sum_{k = 1}^\nbNodes \abs{\node_k}^n  = \max\left\{\sum_{k = 1}^\nbNodes \abs{\node_k}^{\nbNodes-1}, \nbNodes\right\},
\end{equation*}
we get
\begin{align}
	&\max\!\left\{\sum_{k = 1}^\nbNodes \abs{\node_k}^{\nbNodes-1}\!, \nbNodes\!\right\} \!\!\left(\!\max_{1 \leq k \leq \nbNodes} \prod_{\substack{\ell = 1\\ \ell \neq k}}^\nbNodes \frac{\max\{1, \abs{z_{\ell}}\}}{\abs{z_{k} - z_{\ell}}} \!\right) \!\leq \kappa\!\left(\vandermondeMat{\nbNodes \times \nbNodes}, \norm{\cdot}_\infty\right) \notag \\
		&\hspace{2cm}\leq \max\!\left\{\sum_{k = 1}^\nbNodes \abs{\node_k}^{\nbNodes-1}\!, \nbNodes\!\right\} \!\!\left(\!\max_{1 \leq k \leq \nbNodes} \prod_{\substack{\ell = 1\\ \ell \neq k}}^\nbNodes \frac{1 + \abs{z_{k}}}{\abs{z_{k} - z_{\ell}}}\!\right).
	\label{eq: upper bound gautschi complex nodes} 
\end{align}
It is furthermore shown in~\cite[Thm.~1]{Gautschi1962} that the upper bound in \eqref{eq: Gautschi's bound general norm infty inv}, and therefore also the upper bound in \eqref{eq: upper bound gautschi complex nodes} are met with equality if the nodes $\node_1, \node_2, \ldots, \node_\nbNodes \in \C$ lie on a ray emanating from the origin, that is, if there exists a $\theta \in [0, 2\pi)$ such that $\node_k = \abs{\node_k}\!e^{i\theta}$, for $k \in \{1, 2, \ldots, \nbNodes\}$. As real nodes trivially satisfy this condition, namely with $\theta = 0$, this result confirms the worst-case condition number behavior associated with real nodes.

The remaining literature on the condition number of complex Vandermonde matrices can principally be divided into the case of all nodes lying on the unit circle and the---more general---case of nodes in the unit disk.

\subsubsection{Vandermonde matrices with nodes on the unit circle} 

The DFT matrix having spectral condition number equal to $1$, irrespectively of its dimension, indicates that Vandermonde matrices with nodes that are in some sense uniformly distributed on the unit circle could be well-conditioned in general. Inspired by this intuition, \citet{Gautschi1990} studied the spectral condition number of $\vandermondeMat{\nbNodes \times \nbNodes}$ with nodes $z_k = e^{2\pi ic_k}$, $k \in \{1, 2, \ldots, \nbNodes\}$, where $c_k$ is the Van der Corput sequence defined as
$c_k = \sum_{\ell = 0}^{L-1} v_\ell^{(k)} 2^{-\ell-1}$, $L = \lfloor \log_2k\rfloor + 1$, and $(v_0^{(k)}, v_1^{(k)}, \ldots, v_{L-1}^{(k)})$ is the binary representation of $k$, i.e., $k = \sum_{\ell = 0}^{L-1} v_\ell^{(k)}2^\ell$.
Van der Corput sequences are used, e.g., in quasi-Monte Carlo simulation algorithms~\citep{Faure2015} and are known to have excellent uniform distribution properties.
It is shown in~\citep[Cor.~3]{Gautschi1990} that the spectral condition number of Vandermonde matrices $\vandermondeMat{\nbNodes \times \nbNodes}$ built from Van der Corput sequences as described above is upper-bounded by $\sqrt{2\nbNodes}$.

\citet[Lem.~3.1]{Berman2007} formally confirmed the intuition, expressed in \citep{Gautschi1990}, that nodes distributed uniformly on the unit circle lead to small condition number. Specifically, it is shown in \citep[Lem.~3.1 \& Thm.~3.2]{Berman2007} that the spectral condition number of $\vandermondeMat{\nbNodes \times \nbNodes}$ with $\node_k  = e^{-2\pi ip_k\tau/\nbNodes}$, $p_k \in \{0, 1, \ldots, M-1\}$, for $k \in \{1, 2, \ldots, \nbNodes\}$, $M > \nbNodes$, $\tau \in \R$, is equal to $1$ if and only if the nodes $\node_1, \node_2, \ldots, \node_\nbNodes$ are distributed uniformly on the unit circle in the following sense: There exists a $\tau \in \R$ such that the spectral condition number of $\vandermondeMat{\nbNodes \times \nbNodes}$ is equal to $1$, irrespectively of $\nbNodes$, if and only if $\left\{\!\left\langle \frac{p_k}{Q}\right\rangle\!\right\}_{k = 1}^K$ is a complete residue system modulo $K$ \cite[Chap.~3, \S20]{Nagell1951}, where $Q \triangleq \mathrm{gcd}(\{p_k\}_{k = 1}^K)$ and $\left\langle \frac{p_k}{Q}\right\rangle$ is the remainder after division of $\frac{p_k}{Q}$ by $K$. 

For Vandermonde matrices $\vandermondeMat{\nbSamples \times \nbNodes}$, $\nbSamples \geq \nbNodes$, with nodes of the form $\node_k = e^{2\pi i \xi_k}$, where $\xi_k \in [0, 1)$, for $k \in \{1, 2, \ldots, \nbNodes\}$, \citet{Ferreira99} employed Ger\v{s}gorin's disc theorem~\cite[Thm.~6.1.1]{Horn1985} to derive a lower bound on $\sigma_\mathrm{min}^2(\vandermondeMat{\nbSamples \times \nbNodes})$ and an upper bound on $\sigma_\mathrm{max}^2(\vandermondeMat{\nbSamples \times \nbNodes})$, which when combined give
\begin{equation}
	\kappa\!\left(\vandermondeMat{\nbSamples \times \nbNodes}\right) \leq \left(\frac{\nbSamples + \left(\displaystyle[\beta] + \frac{\beta^2}{[\beta]} - 1\right)}{\nbSamples - \left(\displaystyle[\beta] + \frac{\beta^2}{[\beta]} - 1\right)}\right)^{1/2} \reversetriangleq B(\nbSamples, \beta),
	\label{eq: Ferreira's bound} 
\end{equation}
for $\nbSamples > [\beta] + \beta^2/[\beta] - 1$. Here, 
\begin{equation}
	\beta \triangleq \frac{\pi\Delta^{(w)}}{\sqrt{3}\sin(\pi\Delta^{(w)})\delta^{(w)}}
	\label{eq: defn beta Ferreira's bound} 
\end{equation}
with the minimum wrap-around distance
\begin{equation}
	\delta^{(w)} \triangleq \min_{\substack{1 \leq k, \ell \leq K \\ k \neq \ell}}\min_{n \in \Z}  \abs{\xi_k - \xi_\ell + n} \label{eq: defn minimum wrap arround distance first time} \\ 
\end{equation}
and the maximum wrap-around distance
\begin{equation*}
	\Delta^{(w)} \triangleq \max_{\substack{1 \leq k, \ell \leq K \\ k \neq \ell}}  \min_{n \in \Z} \abs{\xi_k - \xi_\ell + n} \notag
\end{equation*}
between the $\xi_k$. Note that $\delta^{(w)} \leq 1/\nbNodes$ as the maximum is achieved for $\nbNodes$ uniformly spaced nodes.  
\citet{Bazan2000}---also based on Ger\v{s}gorin's disc theorem---derived a lower bound on $\sigma_\mathrm{min}^2(\vandermondeMat{\nbSamples \times \nbNodes})$ and an upper bound on $\sigma_\mathrm{max}^2(\vandermondeMat{\nbSamples \times \nbNodes})$ based on which one gets
\begin{equation}
	\kappa(\vandermondeMat{\nbSamples \times \nbNodes}) \leq \sqrt{\frac{\nbSamples + (2\nbNodes-2)/\sigma}{\nbSamples - (2\nbNodes-2)/\sigma}},
	\label{eq: upper bound condition number for nodes on the unit circle Bazan} 
\end{equation}
for $\nbSamples > 2(\nbNodes-1)/\sigma$, where $\sigma$ is the minimum (Euclidean) distance between the nodes $\node_k$ defined as 
\begin{equation}
	\sigma \triangleq \min_{\substack{1 \leq k, \ell \leq \nbNodes \\ k \neq \ell}} \abs{\node_k - \node_\ell}\!.
	\label{eq: minimum distance complex number} 
\end{equation}
\citet{Negreanu2006} and \citet[Thm.~2]{Liao2014} discovered discrete versions of Ingham's inequalities~\citep{Ingham1936}. Besides the performance analysis of the MUSIC algorithm conducted in \citep{Liao2014}, these discrete Ingham inequalities also find application in the finite-difference discretization of homogeneous 1D wave equations~\cite{Negreanu2006}. In the present context, they provide a lower bound on $\sigma_\mathrm{min}^2(\vandermondeMat{\nbSamples \times \nbNodes})$ and an upper bound on $\sigma_\mathrm{max}^2(\vandermondeMat{\nbSamples \times \nbNodes})$, which, taken together, yield
\begin{equation}
	\kappa(\vandermondeMat{\nbSamples \times \nbNodes}) \leq \left(\frac{\displaystyle\frac{8\sqrt{2}\left\lceil (\nbSamples-1)/2\right\rceil}{\pi} + \frac{\sqrt{2}}{2\pi\left\lceil (\nbSamples-1)/2 \right\rceil(\delta^{(w)})^2} + 3\sqrt{2}}{\displaystyle\frac{2(\nbSamples-1)}{\pi} - \frac{2}{\pi(\nbSamples-1)(\delta^{(w)})^2}-4}\right)^{1/2},
	\label{eq: upper bound condition number liao} 
\end{equation}
for $\nbSamples \geq 7$ and 
\begin{equation*}
	\delta^{(w)} > \frac{1}{\nbSamples}\sqrt{\frac{2}{\pi}}\left(\frac{2}{\pi} - \frac{4}{\nbSamples}\right)^{-1/2}.
\end{equation*}
Another upper bound on $\kappa(\vandermondeMat{\nbSamples \times \nbNodes})$ was recently reported by Moitra in \citep{Moitra2014}. As Moitra's result is closely related to our main result, we review it in detail separately in Section \ref{sec: moitra's results summarized... and selberg}.

\subsubsection{Vandermonde matrices with nodes in the unit disk}
\label{sec: Vandermonde matrices with nodes in the unit disk previous work} 

For nodes $\node_k$ in the unit disk, i.e., $\abs{\node_k} \leq 1$, for all $k \in \{1, 2, \ldots, \nbNodes\}$, Gautschi's upper bound \eqref{eq: upper bound gautschi complex nodes} becomes
\begin{equation*}
	\kappa\!\left(\vandermondeMat{\nbNodes \times \nbNodes}, \norm{\cdot}_\infty\right) \leq \nbNodes (2/\sigma)^{\nbNodes-1}.
\end{equation*}
This result holds, however, for square Vandermonde matrices only. To the best of our knowledge, the only analytical result available on the condition number of rectangular (i.e., $\nbSamples \geq \nbNodes$) Vandermonde matrices with nodes in the unit disk is due to \citet{Bazan2000}. We review Baz\'an's result in Section \ref{sec: discussion bazan's results} in the course of a comparison to our results.

%% file: mainResults.tex

\section{Vandermonde matrices with nodes on the unit circle and the large sieve}

The proof of our main result is inspired by a link---first established by \citet[pp.~213--226]{Selberg1991} and later extended by \citet[Thm.~2.3]{Moitra2015}---between the extremal singular values of Vandermonde matrices with nodes on the unit circle and the ``large sieve''~\citep{Linnik1941, Renyi1951, Roth1964, Roth1965, Bombieri1965, Montgomery1968}, a family of inequalities involving polynomials in $e^{2\pi i \xi}$, $\xi \in [0, 1)$, originally developed in the field of analytic number theory~\citep{Bombiri1974, Montgomery2001}.

\subsection{A brief introduction to the large sieve}

We start with a brief introduction to the large sieve emphasizing the aspects relevant to the problem at hand. Specifically, we shall work with the definition of the large sieve as put forward by~\citet[Thm.~1]{Davenport1966}.

\begin{defn}[Large sieve inequality]
	Let $\boldsymbol{y} \triangleq \{y_n\}_{n = 0}^{\nbSamples-1} \in \C^\nbSamples$. Define the trigonometric polynomial
	\begin{equation}
		\forall \xi \in \R, \qquad S_{\boldsymbol{y}, \nbSamples}(\xi) \triangleq \sum_{n = 0}^{\nbSamples-1} y_n e^{-2\pi in\xi}.
		\label{eq: definition S in the large sieve} 
	\end{equation}
	Let $\xi_1, \xi_2, \ldots, \xi_\nbNodes \in [0, 1)$ be such that the minimum wrap-around distance satisfies
	\begin{equation*}
		\delta^{(w)} \triangleq \min_{\substack{1 \leq k, \ell \leq K \\ k \neq \ell}}\min_{n \in \Z} \abs{\xi_k - \xi_\ell + n} > 0.
	\end{equation*}
	A large sieve inequality is an inequality of the form 
	\begin{equation}
		\sum_{k = 1}^\nbNodes \abs{S_{\boldsymbol{y}, \nbSamples}(\xi_k)}^2 \leq \Delta(\nbSamples, \delta^{(w)}) \sum_{n = 0}^{\nbSamples-1} \abs{y_n}^2,
		\label{eq: def large sieve inequality} 
	\end{equation}
	where $\Delta(\nbSamples, \delta^{(w)})$ depends on $\nbSamples$ and $\delta^{(w)}$ only.	
\end{defn}

The large sieve inequality \eqref{eq: def large sieve inequality} essentially says that the energy contained in the samples $S_{\boldsymbol{y}, \nbSamples}(\xi_k)$, $k \in \{1, 2, \ldots, \nbNodes\}$, of the trigonometric polynomial $S_{\boldsymbol{y}, \nbSamples}$ is bounded by the total energy of $S_{\boldsymbol{y}, \nbSamples}$ (given by $\sum_{n = 0}^{\nbSamples-1} \abs{y_n}^2$) multiplied by a factor that depends on $\nbSamples$ and the minimum wrap-around distance between the $\xi_k$ only.

\citet[Thm.~1]{Davenport1966} established \eqref{eq: def large sieve inequality} with $\Delta(\nbSamples, \delta^{(w)}) = 2.2\times\max\{\nbSamples, 1/\delta^{(w)}\}$, \citet{Gallagher1967} with $\Delta(\nbSamples, \delta^{(w)}) = \pi\nbSamples+1/\delta^{(w)}$, \citet{Liu1969} with $\Delta(\nbSamples, \delta^{(w)}) = 2\max\{\nbSamples, 1/\delta^{(w)}\}$, Bombieri and Davenport with $\Delta(\nbSamples, \delta^{(w)}) = (\sqrt{\nbSamples} + 1/\sqrt{\delta^{(w)}})^2$ in \citep{Bombieri1968} and with $\Delta(\nbSamples, \delta^{(w)}) = \nbSamples+5/\delta^{(w)}$ in \citep{Bombieri1969}. 
\citet[Thm.~1]{Montgomery1973} proved \eqref{eq: def large sieve inequality} with $\Delta(\nbSamples, \delta^{(w)}) = \nbSamples + 1/\delta^{(w)}$, later improved to $\Delta(\nbSamples, \delta^{(w)}) = \nbSamples -1 + 1/\delta^{(w)}$ by Cohen and independently by Selberg \cite[Thm.~3]{Montgomery1978}. 
In particular, Cohen used a ``dilatation trick'' to replace $\nbSamples$ in the Montgomery--Vaughan result \citep[Thm.~1]{Montgomery1973} by $\nbSamples-1$, while Selberg's improvement~\cite[pp.~213--226]{Selberg1991} relies on the construction of an extremal majorant of the characteristic function $\chi_E$ of the interval $E \triangleq [0, (\nbSamples-1)\delta^{(w)}]$. An extremal majorant of a function $\psi \colon \R \rightarrow \R$ is an entire function $M_\psi \colon \C \rightarrow \C$ of exponential type at most $2\pi$~\cite[p.~839]{Boas1942} which majorizes $\psi$ along the real axis, i.e., $\psi(u) \leq M_\psi(u)$, for all $u \in \R$, and at the same time minimizes the integral $\int_{-\infty}^\infty (M_\psi(u) - \psi(u))\mathrm{d}u$. 

\subsection{Extremal singular values of Vandermonde matrices with nodes on the unit circle and the large sieve}
\label{sec: moitra's results summarized... and selberg} 

For $\nbSamples \times \nbNodes$, $\nbSamples \geq \nbNodes$, Vandermonde matrices with nodes $e^{2\pi i\xi_k}$\!, $\xi_k \!\in \![0, 1)$, $k \!\in\! \{1, 2, \ldots, \nbNodes\}$, and minimum wrap-around distance $\delta^{(w)}$,  \citet[Thm.~2.3]{Moitra2015} showed that
\begin{equation}
	\kappa(\vandermondeMat{\nbSamples \times \nbNodes}) \leq \sqrt{\frac{\nbSamples-1+1/\delta^{(w)}}{\nbSamples-1-1/\delta^{(w)}}},
	\label{eq: upper bound condition number for nodes on the unit circle Moitra} 
\end{equation}
for $\nbSamples > 1 + 1/\delta^{(w)}$. 
This result is obtained from the upper bound on $\sigma_\mathrm{max}^2(\vandermondeMat{\nbSamples \times \nbNodes})$ reported by Selberg in \citep{Selberg1991} and a new lower bound on $\sigma_\mathrm{min}^2(\vandermondeMat{\nbSamples \times \nbNodes})$ derived by Moitra in \citep{Moitra2014}. 

Moitra's main insight was to recognize that replacing the extremal majorant of $\chi_E$ in Selberg's proof of the large sieve inequality by the extremal minorant of $\chi_E$ readily leads to a lower bound on $\sigma_\mathrm{min}^2(\vandermondeMat{\nbSamples \times \nbNodes})$. 
We note that the condition $\nbSamples > 1+1/\delta^{(w)}$ for \eqref{eq: upper bound condition number for nodes on the unit circle Moitra} to hold excludes the case of square Vandermonde matrices, that is, $\nbSamples = \nbNodes$, because $\nbSamples > 1 + 1/\delta^{(w)} \geq \nbNodes + 1$ as a consequence of $\delta^{(w)} \leq 1/\nbNodes$.  

We proceed to explaining in detail how \eqref{eq: upper bound condition number for nodes on the unit circle Moitra} is obtained and to this end start by briefly reviewing Selberg's proof of the large sieve inequality. Selberg starts by considering the extremal majorant 
\begin{equation*}
	\forall z \in \C, \qquad C_E(z) \triangleq \frac{1}{2}\big(B((\nbSamples-1)\delta^{(w)}-z) + B(z)\big)
\end{equation*}
of the characteristic function $\chi_E$ of the interval $E = [0, (\nbSamples-1)\delta^{(w)}]$, 
where $B$ stands for Beurling's extremal majorant of the signum function given by~\citep{Beurling1989} 
\begin{equation}
	\forall z \in \C, \quad \!B(z) \triangleq \!\left(\!\frac{\sin(\pi z)}{\pi}\!\right)^2\!\left(\sum_{n = 0}^\infty \frac{1}{(z-n)^2} - \sum_{n = -\infty}^{-1} \frac{1}{(z-n)^2} + \frac{2}{z}\right)\!.
	\label{eq: beurling function} 
\end{equation}
An important property of $C_E$ is
\begin{equation}
	\int_{-\infty}^\infty \left(C_E(u) - \chi_E(u)\right)\mathrm{d}u = 1.
	\label{eq: equality integral extremal majorant characteristic function} 
\end{equation}
Letting $\vandermondeMat{\nbSamples \times \nbNodes}$ be the Vandermonde matrix with nodes $\node_k = e^{2\pi i\xi_k}$, $\xi_k \in [0, 1)$, for $k \in \{1, 2, \ldots, \nbNodes\}$, Selberg first notes that
\begin{equation}
	\sum_{k = 1}^\nbNodes \abs{S_{\boldsymbol{y}, \nbSamples}(\xi_k)}^2 = \norm{(\vandermondeMat{\nbSamples \times \nbNodes})^H\boldsymbol{y}}_2^2,
	\label{eq: selberg equation link vandermonde} 
\end{equation}
for all $\boldsymbol{y} \triangleq \{y_n\}_{n = 0}^{\nbSamples-1} \in \C^\nbSamples$. This implies that the large sieve inequality holds with every $\Delta(\nbSamples, \delta^{(w)})$ satisfying
\begin{equation}
	\Delta(\nbSamples, \delta^{(w)}) \geq \sigma_\mathrm{max}^2((\vandermondeMat{\nbSamples \times \nbNodes})^H) = \sigma_\mathrm{max}^2(\vandermondeMat{\nbSamples \times \nbNodes}). \label{eq: proof selberg 0} 
\end{equation}
Conversely, every $\Delta(\nbSamples, \delta^{(w)})$ such that \eqref{eq: def large sieve inequality} holds for all $\boldsymbol{y} \triangleq \{y_n\}_{n = 0}^{\nbSamples-1} \in \C^\nbSamples$ must satisfy \eqref{eq: proof selberg 0}. 
Selberg goes on to derive an upper bound on $\sigma_\mathrm{max}^2(\vandermondeMat{\nbSamples \times \nbNodes})$ as follows. Let $\boldsymbol{x} \triangleq \{x_k\}_{k = 1}^\nbNodes$, $\psi_{\boldsymbol{x}}(u) \triangleq \sum_{k = 1}^\nbNodes x_k e^{2\pi i\xi_k u}$, for all $u \in \R$, and note that
\begin{align}
	\norm{\vandermondeMat{\nbSamples \times \nbNodes}\boldsymbol{x}}_2^2 &= \sum_{n = 0}^{\nbSamples-1} \abs{\psi_{\boldsymbol{x}}(n)}^2 \notag \\
		&= \sum_{n = -\infty}^\infty \chi_E(\delta^{(w)} n)\abs{\psi_{\boldsymbol{x}}(n)}^2 \notag \\
		&\leq \sum_{n = -\infty}^\infty C_E(\delta^{(w)} n)\abs{\psi_{\boldsymbol{x}}(n)}^2 \label {eq: proof selberg 0.5} \\ 
		&= \sum_{n = -\infty}^\infty C_E(\delta^{(w)} n) \sum_{k, \ell = 1}^\nbNodes x_k\overline{x_\ell}e^{2\pi i(\xi_k - \xi_\ell)n}. \label{eq: proof selberg 1} 
\end{align}
$C_E$ is integrable over $\R$, thanks to $C_E \geq 0$ and \eqref{eq: equality integral extremal majorant characteristic function}, and therefore the Fourier transform $\widehat{C}_E$ of its restriction to $\R$ is continuous. Moreover, as $C_E$ is an entire function of exponential type at most $2\pi$, $\widehat{C}_E$ is supported on $[-1, 1]$.
The Poisson summation formula then yields
\begin{align}
	\sum_{n = -\infty}^\infty C_E(\delta^{(w)} n)e^{2\pi i(\xi_k - \xi_\ell)n} &= (\delta^{(w)})^{-1}\sum_{n = -\infty}^\infty \widehat{C}_E((\delta^{(w)})^{-1}(n-(\xi_k - \xi_\ell))) \notag \\
		&= \begin{cases} (\delta^{(w)})^{-1}\widehat{C}_E(0), & k = \ell \\ 0, & \text{otherwise,} \end{cases} \label{eq: application Poisson} 
\end{align}
where \eqref{eq: application Poisson} follows from $\abs{n - (\xi_k - \xi_\ell)} \geq \delta^{(w)}$, for all $n \in \Z$, and all $k, \ell \in \{1, 2, \ldots, \nbNodes\}$ such that $k \neq \ell$, and the fact that $\widehat{C}_E$ is a continuous function supported on $[-1, 1]$ (which implies $\widehat{C}_E(-1) = \widehat{C}_E(1) = 0$).
Note that the conditions for the application of the Poisson summation formula are met as $C_E$ is integrable over $\R$, which, combined with the fact that $C_E$ is an entire function of exponential type at most $2\pi$, implies that $C_E'$ is integrable over $\R$ \cite[Pt.~2, Sec.~3., Prob.~7]{Young1990} and $C_E(u) \rightarrow 0$ as $\abs{u} \rightarrow \infty$. From \eqref{eq: equality integral extremal majorant characteristic function} we therefore get
\begin{equation}
	\widehat{C}_E(0) = \int_{-\infty}^\infty C_E(u)\mathrm{d}u =  1 + \int_{-\infty}^\infty \chi_E(u)\mathrm{d}u = 1 + (\nbSamples-1)\delta^{(w)}.
	\label{eq: proof selberg 2} 
\end{equation}
Combining \eqref{eq: proof selberg 1}, \eqref{eq: application Poisson}, and \eqref{eq: proof selberg 2} thus yields 
\begin{equation}
	\norm{\vandermondeMat{\nbSamples \times \nbNodes}\boldsymbol{x}}_2^2 \leq \left(\nbSamples - 1 + 1/\delta^{(w)}\right)\!\norm{\boldsymbol{x}}_2^2.
	\label{eq: proof selberg result upper bound} 
\end{equation}
As \eqref{eq: proof selberg result upper bound} holds for all $\boldsymbol{x} \in \C^\nbNodes$, we can conclude that $\sigma_\mathrm{max}^2(\vandermondeMat{\nbSamples \times \nbNodes}) \leq \nbSamples - 1 + 1/\delta^{(w)}$, which, thanks to \eqref{eq: proof selberg 0}, yields the large sieve inequality with $\Delta(\nbSamples, \delta^{(w)}) = \nbSamples-1+1/\delta^{(w)}$.
\citet{Bombieri1969} showed that the large sieve inequality with $\Delta(\nbSamples, \delta^{(w)}) = \nbSamples -1 + 1/\delta^{(w)}$ is tight by constructing an explicit example saturating \eqref{eq: def large sieve inequality} with $\Delta(\nbSamples, \delta^{(w)}) = \nbSamples-1+1/\delta^{(w)}$.

We are now ready to review Moitra's lower bound on $\sigma_\mathrm{min}^2(\vandermondeMat{\nbSamples \times \nbNodes})$. Specifically, Moitra recognized that Selberg's idea for upper-bounding $\sigma_\mathrm{max}^2(\vandermondeMat{\nbSamples \times \nbNodes})$ can also be applied to lower-bound $\sigma_{\mathrm{min}}^2(\vandermondeMat{\nbSamples \times \nbNodes})$, simply by working with the extremal minorant of $\chi_E$, constructed by Selberg in \cite{Selberg1991}, instead of the extremal majorant. 
An extremal minorant of a function $\psi \colon \R \rightarrow \R$ is an entire function $m_\psi \colon \C \rightarrow \C$ of exponential type at most $2\pi$ which minorizes $\psi$ along the real axis, i.e., $m_\psi(u) \leq \psi(u)$, for all $u \in \R$, and at the same time minimizes the integral $\int_{-\infty}^\infty (\psi(u) - m_\psi(u))\mathrm{d}u$. The extremal minorant of $\chi_E$ constructed by Selberg is
\begin{equation*}
	\forall z \in \C, \qquad c_E(z) \triangleq -\frac{1}{2}\big(B(z-(\nbSamples-1)\delta^{(w)}) + B(-z)\big),
\end{equation*}
where $B$ was defined in \eqref{eq: beurling function}.
By construction, $c_E$ satisfies
\begin{equation}
	\int_{-\infty}^\infty (\chi_E(u) - c_E(u))\mathrm{d}u = 1.
	\label{eq: equality integral extremal minorant characteristic function} 
\end{equation}
Moitra showed that $\sigma_\mathrm{min}^2(\vandermondeMat{\nbSamples \times \nbNodes}) \geq \nbSamples-1-1/\delta^{(w)}$ by replacing $\leq$ in \eqref{eq: proof selberg 0.5} and $C_E$ in \eqref{eq: proof selberg 0.5}-\eqref{eq: proof selberg 1} by $\geq$ and $c_E$, respectively, and employing arguments similar to those in \eqref{eq: application Poisson} and \eqref{eq: proof selberg 2} with $c_E$ in place of $C_E$. 
The final result \eqref{eq: upper bound condition number for nodes on the unit circle Moitra} then follows by using this lower bound in conjunction with the Selberg upper bound on $\sigma_{\mathrm{max}}^2(\vandermondeMat{\nbSamples \times \nbNodes})$ in \eqref{eq: definition spectral condition number}.

\subsection{Relation to other bounds in the literature}

We now put the Selberg--Moitra bound into perspective with respect to other bounds (for nodes on the unit circle) available in the literature. 
Both the Selberg--Moitra bound \eqref{eq: upper bound condition number for nodes on the unit circle Moitra} as well as the Liao-Fannjiang bound \eqref{eq: upper bound condition number liao} depend neither on the maximum wrap-around distance $\Delta\!^{(w)}$, as Ferreira's bound \eqref{eq: Ferreira's bound} does, nor do they exhibit a dependence on $\nbNodes$ as is the case for Baz\'an's bound \eqref{eq: upper bound condition number for nodes on the unit circle Bazan}.
While \eqref{eq: Ferreira's bound}, \eqref{eq: upper bound condition number liao}, and \eqref{eq: upper bound condition number for nodes on the unit circle Moitra} depend on the minimum wrap-around distance $\delta^{(w)}$, Baz\'an's bound \eqref{eq: upper bound condition number for nodes on the unit circle Bazan} is in terms of the minimum distance $\sigma$ between the nodes $\node_k$.
However, as $\sigma = 2\sin(\pi\delta^{(w)})$ (which follows from a simple geometric argument) and $2x/\pi \leq \sin(x) \leq x$, for $x \in [0, \pi/2)$, we get $4\delta^{(w)} \leq \sigma \leq 2\pi\delta^{(w)}$, so that the bounds \eqref{eq: Ferreira's bound}, \eqref{eq: upper bound condition number liao}, and \eqref{eq: upper bound condition number for nodes on the unit circle Moitra} can readily be expressed in terms of $\sigma$.

We next analyze Ferreira's bound \eqref{eq: Ferreira's bound}. As $\delta^{(w)} \leq 1/2$ and $2x/\pi \leq \sin(x) \leq x$, for $x \in [0, \pi/2)$, it follows that $\beta$ in \eqref{eq: defn beta Ferreira's bound} satisfies
\begin{equation*}
	1 \leq \frac{1}{\sqrt{3}\delta^{(w)}} \leq \beta \leq \frac{\pi}{2\sqrt{3}\delta^{(w)}} \leq \frac{1}{\delta^{(w)}}.
\end{equation*}
Further, we have 
\begin{equation}
	\frac{17}{18\delta^{(w)}} - \frac{1}{2} \leq [\beta] + \frac{\beta^2}{[\beta]} \leq \frac{7}{3\delta^{(w)}} + \frac{1}{2}, \label{eq: third inequality comparison bounds ferreira}  
\end{equation}
where we used $x-1/2 \leq [x] \leq x+1/2$, the fact that the functions $x \mapsto x^2/(x-1/2) + x + 1/2$ and $x \mapsto x - 1/2 + x^2/(x+1/2)$ are non-decreasing on $[1, \infty)$, and $\delta^{(w)} \leq 1/2$.
Employing \eqref{eq: third inequality comparison bounds ferreira} in Ferreira's bound \eqref{eq: Ferreira's bound}, we get
\begin{equation*}
	\sqrt{\frac{\nbSamples - 3/2 + 17/(18\delta^{(w)})}{\nbSamples + 3/2 - 17/(18\delta^{(w)})}} \leq B(\nbSamples, \beta) \leq \sqrt{\frac{\nbSamples - 1/2 + 7/(3\delta^{(w)})}{\nbSamples + 1/2 - 7/(3\delta^{(w)})}},
\end{equation*}
which shows that Ferreira's bound \eqref{eq: Ferreira's bound} exhibits the same structure as the Selberg--Moitra bound \eqref{eq: upper bound condition number for nodes on the unit circle Moitra}.

\section{Main result}

The main conceptual contribution of the present paper is an extension of the connection between the extremal singular values of Vandermonde matrices and the large sieve principle from the unit circle to the unit disk. 
As a byproduct, we find a new large sieve-type inequality involving polynomials in $z \in \C$ with $\abs{z} \leq 1$ instead of trigonometric polynomials (i.e., polynomials in $e^{2\pi i\xi}$). 
This generalization can not be deduced from the Selberg--Moitra result whose proof relies on Fourier-analytic techniques and the Poisson summation formula and is hence restricted to nodes on the unit circle and to the classical large sieve inequality involving polynomials in variables that take value on the unit circle.  
It turns out, however, that an alternative connection between the extremal singular values of Vandermonde matrices with nodes on the unit circle and the large sieve can be obtained based on the Montgomery--Vaughan proof \cite{Montgomery1973} of the large sieve inequality. The key insight now is that this alternative connection---thanks to being built on generalizations of Hilbert's inequality---can be extended from the unit circle to the unit disk. As a byproduct, the corresponding result---when particularized to the unit circle---slightly improves upon the Selberg--Moitra upper bound.

For pedagogical reasons, we start by explaining our approach for the special case of nodes on the unit circle. The general case of nodes in the unit disk is presented in Section \ref{sec: spectral properties of vandermonde matrices with nodes in the unit disk}.

\subsection{An alternative connection for nodes on the unit circle}
The Montgomery--Vaughan proof of the large sieve inequality with $\Delta(\nbSamples, \delta^{(w)}) = \nbSamples+1/\delta^{(w)}$ is based on a generalization of Hilbert's inequality~\citep{Montgomery1974}, which in its original form states that\footnote{Hilbert actually proved \eqref{eq: Hilbert's inequality} with a factor of $2\pi$ instead of $\pi$. Later \citet{Schur1911} replaced the factor $2\pi$ by the best possible constant $\pi$, but the inequality \eqref{eq: Hilbert's inequality} has come to be referred to as ``Hilbert's inequality''.}
\begin{equation}
	\abs{\sum_{\substack{k, \ell = 1 \\ k \neq \ell}}^K \frac{x_k\overline{x_\ell}}{k - \ell}} \leq \pi \sum_{k = 1}^K \abs{x_k}^2,
	\label{eq: Hilbert's inequality} 
\end{equation}
for arbitrary $\boldsymbol{x} \triangleq \{x_k\}_{k = 1}^K \in \C^K$.
Specifically, Montgomery and Vaughan generalize \eqref{eq: Hilbert's inequality} as follows.

\begin{thm}[{Generalization of Hilbert's inequality,~\cite[Thms.~1 \& 2]{Montgomery1974}}] 
	\label{thm: first generalization of Hilbert's inequality} 
	Let $\nbNodes \in \N\!\setminus\! \{0\}$.
	\begin{enumerate}[a)] 
		\item\label{enumerate item: first generalization of Hilbert's inequality pas sinus} 
		For all $\boldsymbol{u} \triangleq \{u_k\}_{k = 1}^\nbNodes \in \R^\nbNodes$ such that 
		\begin{equation*}
			\delta \triangleq \min_{\substack{1 \leq k,\ell \leq K \\ k \neq \ell}} \abs{u_k - u_\ell} > 0,
		\end{equation*}
		we have 
	\begin{equation}
		\abs{\sum_{\substack{k, \ell = 1 \\ k \neq \ell}}^K \frac{\alpha_k\overline{\alpha_\ell}}{2\pi(u_k - u_\ell)}} \leq \frac{1}{2\delta} \sum_{k = 1}^K \abs{\alpha_k}^2,
		\label{eq: first generalization hilbert's inequality}  
	\end{equation}
	for all $\boldsymbol{\alpha} \triangleq \{\alpha_k\}_{k = 1}^\nbNodes \in \C^\nbNodes$.
	
		\item\label{enumerate item: first generalization of Hilbert's inequality sinus} 
	For all  $\boldsymbol{\xi} \triangleq \{\xi_k\} \in \R^\nbNodes$ such that 
	\begin{equation*}
		\delta^{(w)} \triangleq \min_{\substack{1 \leq k, \ell \leq K \\ k \neq \ell}}\min_{n \in \Z} \abs{\xi_k - \xi_\ell + n} > 0,
	\end{equation*}
	we have
	\begin{equation}
		\abs{\sum_{\substack{k, \ell = 1 \\ k \neq \ell}}^K \frac{a_k\overline{a_\ell}}{\sin(\pi(\xi_k - \xi_\ell))}} \leq \frac{1}{\delta^{(w)}} \sum_{k = 1}^K \abs{a_k}^2,
		\label{eq: inequality sin} 
	\end{equation}
	for all $\boldsymbol{a} \triangleq \{a_k\}_{k = 1}^\nbNodes \in \C^\nbNodes$.
	\end{enumerate}
\end{thm}

Setting $u_k = k$, for $k \in \{1, 2, \ldots, \nbNodes\}$ in \eqref{eq: first generalization hilbert's inequality}, we recover \eqref{eq: Hilbert's inequality} since
\vspace{-0.2cm}
\begin{equation*}
	\delta =\min_{\substack{1 \leq k, \ell \leq \nbNodes \\ k \neq \ell}} \abs{k - \ell} = 1. \\[-0.4cm]
\end{equation*}
Based on Theorem~\ref{thm: first generalization of Hilbert's inequality}, \citet[Thm.~1]{Montgomery1973} proved the large sieve inequality with $\Delta(\nbSamples, \delta^{(w)}) = \nbSamples + 1/\delta^{(w)}$, which, thanks to \eqref{eq: proof selberg 0}, implies $\sigma_\mathrm{max}^2(\vandermondeMat{\nbSamples \times \nbNodes}) \leq \Delta(\nbSamples, \delta^{(w)})$. 
We next adapt the methodology of the proof of \citep[Thm.~1]{Montgomery1973} to derive a lower bound on $\sigma_\mathrm{min}^2(\vandermondeMat{\nbSamples \times \nbNodes})$, and, en route, present the proof of $\sigma_\mathrm{max}^2(\vandermondeMat{\nbSamples \times \nbNodes}) \leq \nbSamples+1/\delta^{(w)}$ provided in \cite{Montgomery1973}. Improving the Montgomery--Vaughan result, by way of Cohen's dilatation trick, to $\sigma_\mathrm{max}^2(\vandermondeMat{\nbSamples \times \nbNodes}) \leq \nbSamples-1+1/\delta^{(w)}$ and combining the result thereof with our new lower bound yields an improvement of the Selberg--Moitra result.

Let $\boldsymbol{x} \triangleq \left(x_1\ x_2\ \ldots\ x_\nbNodes\right)^T \in \C^\nbNodes$. 
For $\node_k = e^{2\pi i\xi_k}$, $k \in \{1, 2, \ldots, \nbNodes\}$, we have
\begin{align}
	\norm{\vandermondeMat{\nbSamples \times \nbNodes}\boldsymbol{x}}_2^2 &= \sum_{n = 0}^{\nbSamples-1}\abs{\sum_{k = 1}^\nbNodes x_k\node_k^n}^2 = \sum_{n = 0}^{\nbSamples-1}\sum_{k, \ell=1}^\nbNodes x_k\overline{x_\ell} e^{2\pi i( \xi_k -  \xi_\ell)n} \notag \\
		&= \sum_{n = 0}^{\nbSamples-1} \left(\sum_{k = 1}^\nbNodes \abs{x_k}^2 + \sum_{\substack{k, \ell = 1 \\ k \neq \ell}}^\nbNodes x_k\overline{x_\ell} e^{2\pi i ( \xi_k -  \xi_\ell)n}\right) \notag \\
		&= \nbSamples\sum_{k = 1}^\nbNodes \abs{x_k}^2 +  \sum_{\substack{k, \ell = 1 \\ k \neq \ell}}^\nbNodes x_k\overline{x_\ell} \left(\sum_{n = 0}^{\nbSamples-1} e^{2\pi i ( \xi_k -  \xi_\ell)n}\right) \notag \\
		&= \nbSamples \norm{\boldsymbol{x}}_2^2 + \sum_{\substack{k, \ell = 1 \\ k \neq \ell}}^\nbNodes x_k\overline{x_\ell} \,\frac{1 - e^{2\pi i( \xi_k -  \xi_\ell)\nbSamples}}{1 - e^{2\pi i( \xi_k -   \xi_\ell)}} \notag \\
		&= \nbSamples \norm{\boldsymbol{x}}_2^2 - \sum_{\substack{k, \ell = 1 \\ k \neq \ell}}^\nbNodes x_k\overline{x_\ell} \,\frac{1 - e^{2\pi i( \xi_k -  \xi_\ell)\nbSamples}}{2i e^{\pi i ( \xi_k -  \xi_\ell)}\sin(\pi( \xi_k -  \xi_\ell))} \notag \\
		&= \nbSamples \norm{\boldsymbol{x}}_2^2 - \underbrace{\sum_{\substack{k, \ell = 1 \\ k \neq \ell}}^\nbNodes \frac{x_k\overline{x_\ell}e^{-\pi i( \xi_k -  \xi_\ell)}}{2i\sin(\pi( \xi_k -  \xi_\ell))}}_{\reversetriangleq \displaystyle X_1} + \underbrace{\sum_{\substack{k, \ell = 1 \\ k \neq \ell}}^\nbNodes \frac{x_k\overline{x_\ell}e^{\pi i( \xi_k -  \xi_\ell)(2\nbSamples-1)}}{2i\sin(\pi( \xi_k -  \xi_\ell))}}_{\reversetriangleq \displaystyle X_2}. \label{eq: equation proof condition number of the Vandermonde matrix with nodes on the unit circle 1} 
\end{align}
As the nodes $\node_1, \node_2, \ldots, \node_\nbNodes$ are, by assumption, pairwise distinct, we have
\begin{equation}
	\delta^{(w)} = \min_{\substack{1 \leq k, \ell \leq \nbNodes \\ k \neq \ell}}\min_{n \in \Z} \abs{ \xi_k -  \xi_\ell + n}> 0,
	\label{eq: minimum distance proof condition number of the Vandermonde matrix on the unit circle} 
\end{equation}
so that application of Theorem~\ref{thm: first generalization of Hilbert's inequality}\ref{enumerate item: first generalization of Hilbert's inequality sinus}) first with $a_k \triangleq x_ke^{-\pi i \xi_k}$, $k \in \{1, 2, \ldots, \nbNodes\}$, 
yields
\begin{equation}
	\abs{X_1} \leq \frac{1}{2\delta^{(w)}}\sum_{k = 1}^\nbNodes \abs{x_ke^{\pi i \xi_k}}^2 = \frac{\norm{\boldsymbol{x}}_2^2}{2\delta^{(w)}} \label{eq: equation proof condition number of the Vandermonde matrix with nodes on the unit circle 2} 
\end{equation}
and then with $a_k \triangleq x_k e^{\pi i \xi_k(2\nbSamples-1)}$, $k \in \{1, 2, \ldots, \nbNodes\}$, results in
\begin{equation}
	\abs{X_2} \leq \frac{1}{2\delta^{(w)}}\sum_{k = 1}^\nbNodes \abs{x_ke^{\pi i \xi_k(2\nbSamples-1)}}^2 = \frac{\norm{\boldsymbol{x}}_2^2}{2\delta^{(w)}}. \label{eq: equation proof condition number of the Vandermonde matrix with nodes on the unit circle 3} 
\end{equation}
Combining \eqref{eq: equation proof condition number of the Vandermonde matrix with nodes on the unit circle 1}, \eqref{eq: equation proof condition number of the Vandermonde matrix with nodes on the unit circle 2}, and \eqref{eq: equation proof condition number of the Vandermonde matrix with nodes on the unit circle 3}, and using the forward and the reverse triangle inequality, we obtain 
\begin{equation}
	(\nbSamples - 1/\delta^{(w)})\norm{\boldsymbol{x}}_2^2 \leq \norm{\vandermondeMat{\nbSamples \times \nbNodes}\boldsymbol{x}}_2^2 \leq (\nbSamples + 1/\delta^{(w)})\norm{\boldsymbol{x}}_2^2,
	\label{eq: equation proof condition number of the Vandermonde matrix with nodes on the unit circle 4} 
\end{equation}
for all $\boldsymbol{x} \in \C^\nbNodes$.
The lower and upper bounds in \eqref{eq: equation proof condition number of the Vandermonde matrix with nodes on the unit circle 4} therefore yield
\begin{align}
	\sigma_\mathrm{min}^2(\vandermondeMat{\nbSamples \times \nbNodes}) &\geq \nbSamples-1/\delta^{(w)} \label{eq: lower bound sigma min Vandermonde matrix first} \\ 
	\sigma_\mathrm{max}^2(\vandermondeMat{\nbSamples \times \nbNodes}) &\leq \nbSamples+1/\delta^{(w)}. \label{eq: upper bound sigma max Vandermonde matrix first} 
\end{align}
The upper bound $\nbSamples + 1/\delta^{(w)}$ in \eqref{eq: upper bound sigma max Vandermonde matrix first} can be refined to $\nbSamples - 1 + 1/\delta^{(w)}$ through Cohen's dilatation trick, explained for the general case of nodes in the unit disk in Section \ref{sec: spectral properties of vandermonde matrices with nodes in the unit disk} (proof of Theorem~\ref{thm: main result general upper bound condition number vandermonde matrix unit disk}).  In summary, we get
\begin{equation}
	\kappa\!\left(\vandermondeMat{\nbSamples \times \nbNodes}\right) = \frac{\sigma_\mathrm{max}(\vandermondeMat{\nbSamples \times \nbNodes})}{\sigma_\mathrm{min}(\vandermondeMat{\nbSamples \times \nbNodes})} \leq \sqrt{\frac{\nbSamples-1+1/\delta^{(w)}}{\nbSamples-1/\delta^{(w)}}},
	\label{eq: slight improvement condition number} 
\end{equation}
for $\nbSamples > 1/\delta^{(w)}$,
which constitutes a slight improvement over the Selberg--Moitra bound \eqref{eq: upper bound condition number for nodes on the unit circle Moitra}.

\subsection{Extremal singular values of Vandermonde matrices with nodes in the unit disk}
\label{sec: spectral properties of vandermonde matrices with nodes in the unit disk} 

We are now ready to proceed to our main result, namely a lower bound on $\sigma_\mathrm{min}^2(\vandermondeMat{\nbSamples \times \nbNodes})$ and an upper bound on $\sigma_\mathrm{max}^2(\vandermondeMat{\nbSamples \times \nbNodes})$ for nodes in the unit disk.
Extending the connection between the extremal singular values of Vandermonde matrices and the Montgomery--Vaughan proof of the large sieve inequality from the unit circle to the unit disk requires a further generalization of Hilbert's inequality as follows.

\begin{thm}[Further generalization of Hilbert's inequality, {\cite[Eq.~5.11]{Graham1981}, \cite{Montgomery1998}}]
	\label{thm: second generalization of hilbert's inequality} 

	Let $\nbNodes \in \N \setminus \{0\}.$ For all $\boldsymbol{\rho} \triangleq \{\rho_k\}_{k = 1}^\nbNodes \in \C^\nbNodes$ with $\rho_k \triangleq \lambda_k + 2\pi i u_k$, where $\lambda_k > 0$ and $u_k \in \R$ is such that 
	\begin{equation*}
		\delta_k \triangleq \min_{\substack{1 \leq \ell \leq \nbNodes \\ \ell \neq k}} \abs{u_k - u_\ell} > 0,
	\end{equation*}
	for all $k \in \{1, 2, \ldots, \nbNodes\}$, we have\footnote{Note that the center term in \eqref{eq: generalization hilbert's inequality case lambda constant} is real-valued as 
	\begin{equation*}
		\overline{\sum_{k, \ell = 1}^M \frac{\alpha_k\overline{\alpha_\ell}}{\rho_k + \overline{\rho_\ell}}} = \sum_{k, \ell = 1}^M \frac{\overline{\alpha_k}\alpha_\ell}{\overline{\rho_k} + \rho_\ell} = \sum_{k, \ell = 1}^M \frac{\alpha_k\overline{\alpha_\ell}}{\rho_k + \overline{\rho_\ell}}.
	\end{equation*}}
	\begin{equation}
		\abs{\sum_{\substack{k, \ell = 1 \\ k \neq \ell}}^K \frac{\alpha_k\overline{\alpha_\ell}}{\rho_k + \overline{\rho_\ell}}} \leq \frac{42}{\pi}\sum_{k = 1}^\nbNodes \frac{\abs{\alpha_k}^2}{\delta_k},
		\label{eq: second generalization hilbert's inequality} 
	\end{equation}
	for all $\boldsymbol{\alpha} \triangleq \{\alpha_k\}_{k = 1}^\nbNodes \in \C^\nbNodes$. 
	Moreover, in the case $\lambda_1 = \lambda_2 = \ldots = \lambda_\nbNodes = \lambda > 0$, \eqref{eq: second generalization hilbert's inequality} can be refined to 
	\begin{equation}
		\frac{1}{\delta(e^{2\lambda/\delta} - 1)}\sum_{k = 1}^\nbNodes \abs{\alpha_k}^2 \leq \sum_{k, \ell = 1}^\nbNodes \frac{\alpha_k\overline{\alpha_\ell}}{\rho_k + \overline{\rho_\ell}} \leq \frac{e^{2\lambda/\delta}}{\delta(e^{2\lambda/\delta} - 1)}\sum_{k = 1}^\nbNodes \abs{\alpha_k}^2,
		\label{eq: generalization hilbert's inequality case lambda constant} 
	\end{equation}
	for all $\boldsymbol{\alpha} \triangleq \{\alpha_k\}_{k = 1}^\nbNodes \in \C^\nbNodes$, where $\delta \triangleq \displaystyle \min_{1 \leq k \leq \nbNodes} \delta_k$.
\end{thm}

As pointed out in \citep{Montgomery1998} the inequalities in \eqref{eq: generalization hilbert's inequality case lambda constant} are best possible while \eqref{eq: second generalization hilbert's inequality} is not. This will be seen to have important ramifications for the range of validity of our main bounds in Theorem \ref{thm: main result general upper bound condition number vandermonde matrix unit disk} and Corollary~\ref{cor: main result general upper bound condition number vandermonde matrix unit disk}. We furthermore observe that \eqref{eq: first generalization hilbert's inequality} can be recovered from \eqref{eq: generalization hilbert's inequality case lambda constant} by subtracting $\sum_{k = 1}^\nbNodes \abs{\alpha_k}^2/(2\lambda)$ in \eqref{eq: generalization hilbert's inequality case lambda constant} and letting $\lambda \rightarrow 0$. Indeed, the lower bound on $\displaystyle\sum_{\substack{k, \ell = 1 \\ k \neq \ell}}^K \frac{\alpha_k\overline{\alpha_\ell}}{2\pi i(u_k - u_\ell)}$ resulting from \eqref{eq: first generalization hilbert's inequality} can be obtained from \eqref{eq: generalization hilbert's inequality case lambda constant} by noting that
\begin{equation}
	\lim_{\lambda \rightarrow 0} \left(\frac{1}{\delta(e^{2\lambda/\delta} - 1)} - \frac{1}{2\lambda}\right) = \lim_{\lambda \rightarrow 0} \frac{1}{\delta}\left(\frac{1}{e^{2\lambda/\delta}-1} - \frac{1}{2\lambda/\delta}\right) = -\frac{1}{2\delta}
	\label{eq: first limit generalization} 
\end{equation}
and the upper bound is a consequence of
\begin{equation}
	\lim_{\lambda \rightarrow 0} \left(\frac{e^{2\lambda/\delta}}{\delta(e^{2\lambda/\delta} - 1)} - \frac{1}{2\lambda}\right) = \lim_{\lambda \rightarrow 0} \left(\frac{1}{\delta} + \frac{1}{\delta(e^{2\lambda/\delta} - 1)} - \frac{1}{2\lambda}\right) = \frac{1}{2\delta},
	\label{eq: second limit generalization} 
\end{equation}
where \eqref{eq: first limit generalization} follows from l'H\^{o}pital's rule applied twice.

Theorem \ref{thm: second generalization of hilbert's inequality} generalizes Theorem \ref{thm: first generalization of Hilbert's inequality}\ref{enumerate item: first generalization of Hilbert's inequality pas sinus}) from $i\R$ to the complex plane, i.e., the $2\pi iu_k \in i\R$ in \eqref{eq: first generalization hilbert's inequality} are replaced by the $\rho_k = \lambda_k + 2\pi i u_k \in \C$ in \eqref{eq: second generalization hilbert's inequality}. We will also need a corresponding generalization of Theorem~\ref{thm: first generalization of Hilbert's inequality}\ref{enumerate item: first generalization of Hilbert's inequality sinus}). This generalization is formalized in Theorem \ref{thm: result to be used for derivation of main result} and builds on the following result.

\begin{prop}
	\label{prop: equivalence hilbert's inequality sh} 
	Let $A \colon (0,\infty)^3 \rightarrow (0,\infty)$ and $B \colon (0,\infty)^3 \rightarrow (0,\infty)$ be functions satisfying 
	\begin{align*}
		\varepsilon A(\varepsilon x, \varepsilon y, \varepsilon z) &= A(x, y, z) \\
		\varepsilon B(\varepsilon x, \varepsilon y, \varepsilon z) &= B(x, y, z), 
	\end{align*}
	for all $x > 0$, $y > 0$, $z > 0$, and $\varepsilon > 0$.
	The following statements are equivalent:
	\begin{enumerate}[i)]
		\item\label{item 1 prop hilbert sh} 
		For all $M \in \N \!\setminus\! \{0\}$ and $\boldsymbol{\rho} \triangleq \{\rho_k\}_{k = 1}^M \in \C^M$ with $\rho_k \triangleq \lambda_k + 2\pi iu_k$, where $\lambda_k > 0$ and $u_k \in \R$ is such that 
	\begin{equation*}
		\delta_k \triangleq \min_{\substack{1 \leq \ell \leq M \\ \ell \neq k}} \abs{u_k - u_\ell} > 0,
	\end{equation*}
	for all $k \in \{1, 2, \ldots, M\}$,  we have
	\begin{equation}
		\sum_{k = 1}^M A(\lambda_k, \delta_k, \delta) \abs{\alpha_k}^2 \leq \sum_{k, \ell = 1}^M \frac{\alpha_k\overline{\alpha_\ell}}{\rho_k + \overline{\rho_\ell}} \leq \sum_{k = 1}^M B(\lambda_k, \delta_k, \delta)\abs{\alpha_k}^2,
		\label{eq: general hilbert fraction} 
	\end{equation}
	for all $\boldsymbol{\alpha} \triangleq \{\alpha_k\}_{k = 1}^M \in \C^M$, 
	where $\delta \triangleq \displaystyle \min_{1 \leq k \leq M} \delta_k$.
	
		\item\label{item 2 prop hilbert sh} 
		For all $M \in \N \!\setminus\! \{0\}$ and $\boldsymbol{r} \triangleq \{r_k\}_{k = 1}^M \in \C^M$ with $r_k \triangleq d_k + 2\pi i\xi_k$, where $d_k > 0$ and $\xi_k \in \R$ is such that 
	\begin{equation*}
		\delta_k^{(w)} \triangleq \min_{\substack{1 \leq \ell \leq M \\ \ell \neq k}}\min_{n \in \Z} \abs{\xi_k - \xi_\ell + n} > 0,
	\end{equation*}
	for all $k \in \{1, 2, \ldots, M\}$,  we have
	\begin{equation}
		\!\sum_{k = 1}^M 2A(d_k, \delta_k^{(w)}, \delta^{(w)}) \abs{a_k}^2 \leq \!\sum_{k, \ell = 1}^M \frac{a_k\overline{a_\ell}}{\sh((r_k + \overline{r_\ell})/2)} \!\leq \sum_{k = 1}^M 2B(d_k, \delta_k^{(w)}, \delta^{(w)})\abs{a_k}^2\!,
		\label{eq: general hilbert sh} 
	\end{equation}
	for all $\boldsymbol{a} \triangleq \{a_k\}_{k = 1}^M \in \C^M$, 
	where $\delta^{(w)} \triangleq \displaystyle \min_{1 \leq k \leq M} \delta_k^{(w)}$.
	\end{enumerate}
\end{prop}

\begin{proof}
	See \ref{app: proof of proposition first appendix}.
\end{proof}

Theorem~\ref{thm: second generalization of hilbert's inequality} provides a generalization of Theorem~\ref{thm: first generalization of Hilbert's inequality}\ref{enumerate item: first generalization of Hilbert's inequality pas sinus}). 
Combining Proposition~\ref{prop: equivalence hilbert's inequality sh} with Theorem~\ref{thm: second generalization of hilbert's inequality}, we get the following generalization of Theorem~\ref{thm: first generalization of Hilbert's inequality}\ref{enumerate item: first generalization of Hilbert's inequality sinus}).

\begin{thm}
	\label{thm: result to be used for derivation of main result} 
	Let $\nbNodes \in \N \!\setminus\! \{0\}$. For all $\boldsymbol{r} \triangleq \{r_k\}_{k = 1}^\nbNodes$ with $r_k \triangleq d_k + 2\pi i \xi_k$, where $d_k > 0$ and $\xi_k \in \R$ is such that
	\begin{equation*}
		\delta_k^{(w)} \triangleq \min_{\substack{1 \leq \ell \leq \nbNodes \\ \ell \neq k}}\min_{n \in \Z} \abs{\xi_k - \xi_\ell + n} > 0,
	\end{equation*}
	for all $k \in \{1, 2, \ldots, \nbNodes\}$, we have 
	\begin{equation}
		\sum_{k = 1}^\nbNodes \!\left(\frac{1}{d_k} - \frac{84}{\pi\delta_k^{(w)}}\right)\!\abs{a_k}^2 \leq \sum_{k, \ell = 1}^\nbNodes \frac{a_k\overline{a_\ell}}{\sh((r_k + \overline{r_\ell})/2)} \leq \sum_{k = 1}^\nbNodes \!\left(\frac{1}{d_k} + \frac{84}{\pi \delta_k^{(w)}}\right)\!\abs{a_k}^2, \label{eq: second hilbert's inequality sh general} 
	\end{equation}
	for all $\boldsymbol{a} = \{a_k\}_{k = 1}^\nbNodes \in \C^\nbNodes$. 
	Moreover, in the case $d_1 = d_2 = \ldots = d_\nbNodes = d > 0$, \eqref{eq: second hilbert's inequality sh general} can be refined to
	\begin{equation}
		\frac{2}{\delta^{(w)}(e^{2d/\delta^{(w)}} - 1)}\sum_{k = 1}^\nbNodes \abs{a_k}^2 \leq \sum_{k, \ell = 1}^\nbNodes \frac{a_k\overline{a_\ell}}{\sh((r_k + \overline{r_\ell})/2)} \leq \frac{2e^{2d/\delta^{(w)}}}{\delta^{(w)}(e^{2d/\delta^{(w)}} - 1)}\sum_{k = 1}^\nbNodes \abs{a_k}^2,
		\label{eq: second hilbert's inequality sh particular case} 
	\end{equation}
	for all $\boldsymbol{a} = \{a_k\}_{k = 1}^\nbNodes \in \C^\nbNodes$, 
	where $\delta \triangleq \displaystyle \min_{1 \leq \ell \leq \nbNodes} \delta_k$.
\end{thm}

\begin{proof}
	See \ref{app: proof of theorem 2 appendix 2}.
\end{proof}

We note that \eqref{eq: inequality sin} can be recovered from \eqref{eq: second hilbert's inequality sh particular case} by subtracting the diagonal terms $\sum_{k = 1}^\nbNodes \abs{a_k}^2\!/\!\sh(d)$ in \eqref{eq: second hilbert's inequality sh particular case}, letting $d \rightarrow 0$, and noting that $\sh(i\pi\xi) = i\sin(\pi \xi)$, for $\xi \in \R$. Indeed, we have
\begin{align*}
	&\lim_{d \rightarrow 0} \!\left(\frac{2}{\delta^{(w)}(e^{2d/\delta^{(w)}} - 1)} - \frac{1}{\sh(d)}\right) \\
		&\hspace{2.5cm}= \lim_{d \rightarrow 0} 2\!\left[\frac{1}{\delta^{(w)}(e^{2d/\delta^{(w)}} - 1)} - \frac{1}{2d} + \frac{1}{2d}\!\left(\!1 - \frac{d}{\sh(d)}\!\right)\!\right]\! = -\frac{1}{\delta^{(w)}}
\end{align*}
and 
\begin{align*}
	&\lim_{d \rightarrow 0} \!\left(\frac{2e^{2d/\delta^{(w)}}}{\delta^{(w)}(e^{2d/\delta^{(w)}} - 1)} - \frac{1}{\sh(d)}\right)\! \\
		&\hspace{2.5cm}= \lim_{\lambda \rightarrow 0} 2\!\left[\frac{e^{2d/\delta^{(w)}}}{\delta^{(w)}(e^{2d/\delta^{(w)}} - 1)} - \frac{1}{2d} + \frac{1}{2d}\!\left(\!1 - \frac{d}{\sh(d)}\right)\!\right]\! = \frac{1}{\delta^{(w)}},
\end{align*}
where we used \eqref{eq: first limit generalization}, \eqref{eq: second limit generalization}, and $\lim_{z \rightarrow 0} \sh(z)/z = 1$. 

We next show that the constant in the lower bound in \eqref{eq: second hilbert's inequality sh particular case} can be improved through a slight modification of a result by \citet{Graham1981}. This improvement is relevant as it leads to improved bounds on $\sigma_\mathrm{min}^2(\vandermondeMat{\nbSamples \times \nbNodes})$ and $\kappa(\vandermondeMat{\nbSamples \times \nbNodes})$ and to a more general condition for these bounds to be valid.

\begin{cor}
	\label{cor: result to be used for derivation of main result} 
	Let $\nbNodes \in \N \!\setminus\! \{0\}$. For all $d > 0$ and $\boldsymbol{\xi} \triangleq \{\xi_k\}_{k = 1}^\nbNodes \in \R^\nbNodes$ such that
	\begin{equation*}
		\delta^{(w)} \triangleq \min_{\substack{1 \leq k,\ell \leq \nbNodes \\ k \neq \ell}}\min_{n \in \Z} \abs{\xi_k - \xi_\ell + n} > 0,
	\end{equation*}
	for all $k \in \{1, 2, \ldots, \nbNodes\}$, we have 
	\begin{equation}
		\frac{2e^d}{\delta^{(w)}(e^{2d/\delta^{(w)}} - 1)}\!\sum_{k = 1}^\nbNodes \!\abs{a_k}^2\! \leq \!\sum_{k, \ell = 1}^\nbNodes \frac{a_k\overline{a_\ell}}{\sh(d + \pi i(\xi_k - \xi_\ell))}\! \leq \frac{2e^{2d/\delta^{(w)}}}{\delta^{(w)}(e^{2d/\delta^{(w)}} - 1)}\!\sum_{k = 1}^\nbNodes \!\abs{a_k}^2\!,
		\label{eq: second hilbert's inequality sh particular case 2} 
	\end{equation}
	for all $\boldsymbol{a} = \{a_k\}_{k = 1}^\nbNodes \in \C^\nbNodes$. 
\end{cor}

\begin{proof}
	Based on an extremal minorant and an extremal majorant of the function
	\begin{equation*}
		\forall t \in \R, \qquad f(t) = \begin{cases} e^{-2\delta t}, & t \geq 0\\
					0, & t < 0, \end{cases}
	\end{equation*} 
	\citet{Graham1981} showed\footnote{The inequalities provided in \citet{Graham1981} are actually given by
	\begin{equation*}
		\!\frac{e^d}{\delta^{(w)}(e^{2d/\delta^{(w)}} - 1)} \!\sum_{k = 1}^\nbNodes \!\abs{a_k}^2\! \leq \sum_{k, \ell = 1}^\nbNodes \frac{a_k\overline{a_\ell}}{\sh(d + \pi i(\xi_\ell - \xi_k))} \leq \frac{e^de^{2d/\delta^{(w)}}}{\delta^{(w)}(e^{2d/\delta^{(w)}} - 1)}\!\sum_{k = 1}^\nbNodes \!\abs{a_k}^2\!.
	\end{equation*}
	We believe, however, that there is a mathematical typo in \citep{Graham1981} and that a factor of $2$ is missing in the lower and the upper bounds.
	} that for $d > 0$, and $\boldsymbol{\xi} \triangleq \{\xi_k\}_{k = 1}^\nbNodes \in \R^\nbNodes$ such that
	\begin{equation*}
		\delta^{(w)} \triangleq \min_{\substack{1 \leq k, \ell \leq \nbNodes \\ k \neq \ell}}\min_{n \in \Z} \abs{\xi_k - \xi_\ell + n} > 0,
	\end{equation*}
	we have
	\begin{equation}
		\!\frac{2e^d}{\delta^{(w)}(e^{2d/\delta^{(w)}} - 1)} \!\sum_{k = 1}^\nbNodes \!\abs{a_k}^2\! \leq \sum_{k, \ell = 1}^\nbNodes \frac{a_k\overline{a_\ell}}{\sh(d + \pi i(\xi_\ell - \xi_k))} \leq \frac{2e^de^{2d/\delta^{(w)}}}{\delta^{(w)}(e^{2d/\delta^{(w)}} - 1)}\!\sum_{k = 1}^\nbNodes \!\abs{a_k}^2\!,
		\label{eq: new bound graham} 
	\end{equation}
	for all $\boldsymbol{a} = \{a_k\}_{k = 1}^\nbNodes \in \C^\nbNodes$. 
	Since $d > 0$, the lower and upper bounds in \eqref{eq: new bound graham} are larger than those in \eqref{eq: second hilbert's inequality sh particular case}. Combining the improved lower bound in \eqref{eq: new bound graham} with the upper bound in \eqref{eq: second hilbert's inequality sh particular case} yields the desired result.
\end{proof}

We are now ready to establish our new bounds on $\sigma_\mathrm{min}^2(\vandermondeMat{\nbSamples \times \nbNodes})$ and $\sigma_\mathrm{max}^2(\vandermondeMat{\nbSamples \times \nbNodes})$ for nodes $\node_1, \node_2, \ldots, \node_\nbNodes$ in the unit disk.

\begin{thm}[Lower bound on $\sigma_\mathrm{min}^2(\vandermonde{\nbSamples \times \nbNodes}{\node_1, \node_2, \ldots, \node_\nbNodes})$ and upper bound on $\sigma_\mathrm{max}^2(\vandermonde{\nbSamples \times \nbNodes}{\node_1, \node_2, \ldots, \node_\nbNodes})$ for nodes in the unit disk]
	\label{thm: main result general upper bound condition number vandermonde matrix unit disk} 
	Let $\boldsymbol{z} \triangleq \{\node_k\}_{k = 1}^\nbNodes \in \C^\nbNodes$ with $\node_k \triangleq \abs{\node_k}\!e^{2\pi i\xi_k}$ be such that $0 < \abs{\node_k} \leq 1$, $\xi_k \in [0, 1)$, and
	\begin{equation*}
		\delta_k^{(w)} \triangleq \min_{\substack{1 \leq \ell \leq \nbNodes \\ \ell \neq k}} \min_{n \in \Z} \abs{\xi_k - \xi_\ell + n} > 0,
	\end{equation*}
	for all $k \in \{1, 2, \ldots, \nbNodes\}$. The extremal singular values of the \mbox{\!Vandermonde} matrix $\vandermondeMat{\nbSamples \times \nbNodes}$ with nodes $\node_1, \node_2, \ldots, \node_\nbNodes$ satisfy
	\begin{align}
		\sigma_\mathrm{min}^2(\vandermonde{\nbSamples \times \nbNodes}{\node_1, \node_2, \ldots, \node_\nbNodes}) &\geq \mathscr{L}(\nbSamples, \abs{\boldsymbol{z}}\!, \boldsymbol{\delta}^{(w)}) \label{eq: lower bound sigma min general} \\ 
		\sigma_\mathrm{max}^2(\vandermonde{\nbSamples \times \nbNodes}{\node_1, \node_2, \ldots, \node_\nbNodes}) &\leq \min\!\Big\{\mathscr{U}(\nbSamples, \abs{\boldsymbol{z}}\!, \boldsymbol{\delta}^{(w)}), \mathscr{U}(\nbSamples-1, \abs{\boldsymbol{z}}\!, \boldsymbol{\delta}^{(w)})\Big\}, \label{eq: upper bound sigma max general} 
	\end{align}
	where $\abs{\boldsymbol{z}} \triangleq \{\abs{z_k}\}_{k = 1}^\nbNodes$, $\boldsymbol{\delta}^{(w)} \triangleq \{\delta_k^{(w)}\}_{k = 1}^\nbNodes$, and
	\begin{align}
		\mathscr{L}(\nbSamples, \abs{\boldsymbol{z}}\!, \boldsymbol{\delta}^{(w)}) &\triangleq \min_{1 \leq k \leq \nbNodes} \left\{\frac{1}{\abs{\node_k}}\!\left[\varphi_\nbSamples(\abs{\node_k}) - \frac{42}{\pi\delta_k^{(w)}}\!\left(1 + \abs{\node_k}^{2\nbSamples}\right)\right]\right\} \label{eq: definition lower bound L} \\ 
		\mathscr{U}(\nbSamples, \abs{\boldsymbol{z}}\!, \boldsymbol{\delta}^{(w)}) &\triangleq \max_{1 \leq k \leq \nbNodes} \left\{\frac{1}{\abs{\node_k}}\!\left[\varphi_\nbSamples({\abs{\node_k}}) + \frac{42}{\pi\delta_k^{(w)}}\!\left(1 + \abs{\node_k}^{2\nbSamples}\right)\right]\right\}, \label{eq: definition upper bound U} 
	\end{align}
	with
	\begin{equation}
		\forall k \in \{1, 2, \ldots, \nbNodes\}, \quad \varphi_\nbSamples(\abs{\node_k}) \triangleq \begin{cases} \displaystyle\frac{\abs{\node_k}^{2\nbSamples}-1}{2\ln\abs{\node_k}}, & \abs{\node_k} < 1 \\ \nbSamples, & \abs{\node_k} = 1.\end{cases}
		\label{eq: definition phi function} 
	\end{equation}
	Moreover, if $\abs{\node_1} = \abs{\node_2} = \ldots = \abs{\node_\nbNodes} = A < 1$, \eqref{eq: lower bound sigma min general} and \eqref{eq: upper bound sigma max general} can be refined to
	\begin{align}
		\sigma_\mathrm{min}^2(\vandermonde{\nbSamples \times \nbNodes}{\node_1, \node_2, \ldots, \node_\nbNodes}) &\geq \frac{1 - A^{2(\nbSamples + 1/2 - 1/\delta^{(w)})}}{\delta^{(w)}(A^{-2/\delta^{(w)}} - 1)A^2} \label{eq: lower bound sigma min particular case} \\ 
		\sigma_\mathrm{max}^2(\vandermonde{\nbSamples \times \nbNodes}{\node_1, \node_2, \ldots, \node_\nbNodes}) &\leq \frac{A^{-2/\delta^{(w)}}\!\left(1 - A^{2(\nbSamples-1+1/\delta^{(w)})}\right)}{\delta^{(w)}(A^{-2/\delta^{(w)}} - 1)}, \label{eq: upper bound sigma max particular case} 
	\end{align}
	where $\delta^{(w)} \triangleq \displaystyle\min_{1 \leq k \leq \nbNodes} \delta_k^{(w)}$. 
\end{thm}

\begin{proof}
	See \ref{app: proof of theorem 5 appendix 4}.
\end{proof}

The following upper bound on the condition number $\kappa(\vandermondeMat{\nbSamples \times \nbNodes})$ is an immediate consequence of Theorem \ref{thm: main result general upper bound condition number vandermonde matrix unit disk}.

\begin{cor}[Upper bound on $\kappa(\vandermonde{\nbSamples \times \nbNodes}{\node_1, \node_2, \ldots, \node_\nbNodes})$ for nodes in the unit disk]
	\label{cor: main result general upper bound condition number vandermonde matrix unit disk} 
	Let $\boldsymbol{z} \triangleq \{\node_k\}_{k = 1}^\nbNodes \in \C^\nbNodes$ with $\node_k \triangleq \abs{\node_k}\!e^{2\pi i\xi_k}$ be such that $0 < \abs{\node_k} \leq 1$, $\xi_k \in [0, 1)$, and
	\begin{equation*}
		\delta_k^{(w)} \triangleq \min_{\substack{1 \leq \ell \leq \nbNodes \\ \ell \neq k}} \min_{n \in \Z} \abs{\xi_k - \xi_\ell + n} > 0,
	\end{equation*}
	for all $k \in \{1, 2, \ldots, \nbNodes\}$.
	The spectral condition number satisfies
	\begin{equation}
		\!\kappa(\vandermonde{\nbSamples \times \nbNodes}{\node_1, \node_2, \ldots, \node_\nbNodes}) \leq \sqrt{\frac{\min\!\Big\{\mathscr{U}(\nbSamples, \abs{\boldsymbol{z}}\!, \boldsymbol{\delta}^{(w)}), \mathscr{U}(\nbSamples-1, \abs{\boldsymbol{z}}\!, \boldsymbol{\delta}^{(w)})\Big\}}{\mathscr{L}(\nbSamples, \abs{\boldsymbol{z}}\!, \boldsymbol{\delta}^{(w)})}},
		\label{eq: main result general} 
	\end{equation}
	if for all $k \in \{1, 2, \ldots, \nbNodes\}$,
	\begin{equation}
		\delta_k^{(w)} > \frac{42}{\pi}\!\left(\frac{1 + \abs{\node_k}^{2\nbSamples}}{\varphi_\nbSamples(\abs{\node_k})}\right),
		\label{eq: condition main result general} 
	\end{equation}
		where $\mathscr{L}(\nbSamples, \abs{\boldsymbol{z}}\!, \boldsymbol{\delta}^{(w)})$, $\mathscr{U}(\nbSamples, \abs{\boldsymbol{z}}\!, \boldsymbol{\delta}^{(w)})$, and $\varphi_\nbSamples(\abs{\node_k})$ are defined in \eqref{eq: definition lower bound L}, \eqref{eq: definition upper bound U}, and \eqref{eq: definition phi function}, respectively, $\abs{\boldsymbol{z}} \triangleq \{\abs{z_k}\}_{k = 1}^\nbNodes$, and $\boldsymbol{\delta}^{(w)} \triangleq \{\delta_k^{(w)}\}_{k = 1}^\nbNodes$. Moreover, if $\abs{\node_1} = \abs{\node_2} = \ldots = \abs{\node_\nbNodes} = A < 1$, the spectral condition number satisfies
	\begin{equation}
		\kappa(\vandermonde{\nbSamples \times \nbNodes}{\node_1, \node_2, \ldots, \node_\nbNodes}) \leq A^{-1/\delta^{(w)}}\sqrt{\frac{A^2\!\left(1 - A^{2(\nbSamples-1+1/\delta^{(w)})}\right)}{1 - A^{2(\nbSamples + 1/2 - 1/\delta^{(w)})}}}
		\label{eq: main result particular case} 
	\end{equation}
	if $\nbSamples > 1/\delta^{(w)} - 1/2$. 
\end{cor}

\begin{proof}
	See \ref{app: proof of corollary main result appendix 3}.
\end{proof}

First, we note that the upper bounds in \eqref{eq: upper bound sigma max general} and \eqref{eq: upper bound sigma max particular case} lead to a generalization of the large sieve inequality from the unit circle to the unit disk in the following sense. For $\boldsymbol{z} \triangleq \{\node_k\}_{k = 1}^\nbNodes \in \C^\nbNodes$ such that $\node_k \triangleq \abs{\node_k}\!e^{2\pi i\xi_k}$, $0 < \abs{\node_k} \leq 1$, $\xi_k \in [0, 1)$, and
	\begin{equation*}
		\delta_k^{(w)} \triangleq \min_{\substack{1 \leq \ell \leq \nbNodes \\ \ell \neq k}} \min_{n \in \Z} \abs{\xi_k - \xi_\ell + n} > 0,
	\end{equation*}
	and $\boldsymbol{y} \triangleq \{y_n\}_{n = 0}^{\nbSamples-1} \in \C^\nbSamples$, we have
	\begin{equation*}
		\sum_{k = 1}^\nbNodes \abs{S_{\boldsymbol{y}, \nbSamples}(\node_k)}^2 \leq \Delta(\nbSamples, \abs{\boldsymbol{z}}\!, \boldsymbol{\delta}^{(w)}) \sum_{n = 0}^{\nbSamples-1} \abs{y_n}^2,
	\end{equation*}
	where the trigonometric polynomial \eqref{eq: definition S in the large sieve} in \eqref{eq: def large sieve inequality} is replaced by the polynomial 
	\begin{equation*}
		\forall z \in \C, \qquad S_{\boldsymbol{y}, \nbSamples}(z) \triangleq \sum_{n = 0}^{\nbSamples-1} y_n \overline{z}^n,
	\end{equation*}
	the sieve factor $\Delta(\nbSamples, \abs{\boldsymbol{z}}\!, \boldsymbol{\delta}^{(w)})$ is given by
	\begin{equation*}
		\Delta(\nbSamples, \abs{\boldsymbol{z}}\!, \boldsymbol{\delta}^{(w)}) \!= \!\begin{cases} 
			\displaystyle \frac{A^{-2/\delta^{(w)}}\!\left(1 - A^{2(\nbSamples-1+1/\delta^{(w)})}\right)}{\delta^{(w)}(A^{-2/\delta^{(w)}} - 1)}, & \!\!\!\!\abs{\node_1}\! = \ldots = \!\abs{\node_\nbNodes}\! = A,\\[0.5cm]
			\min\!\Big\{\!\mathscr{U}(\nbSamples, \abs{\boldsymbol{z}}\!, \boldsymbol{\delta}^{(w)}), \mathscr{U}(\nbSamples-1, \abs{\boldsymbol{z}}\!, \boldsymbol{\delta}^{(w)})\!\Big\}, & \!\!\!\!\text{else,}\end{cases}
	\end{equation*}
 	$\abs{\boldsymbol{z}} \triangleq \{\abs{z_k}\}_{k = 1}^\nbNodes$, $\boldsymbol{\delta}^{(w)} \triangleq \{\delta_k^{(w)}\}_{k = 1}^\nbNodes$, and $\mathscr{U}(\nbSamples, \abs{\boldsymbol{z}}\!, \boldsymbol{\delta}^{(w)})$ is as defined in \eqref{eq: definition upper bound U}.

We furthermore note that the bound \eqref {eq: slight improvement condition number}, valid for nodes on the unit circle, can be recovered by letting $A \rightarrow 1$ in \eqref{eq: main result particular case}. Moreover, \eqref{eq: lower bound sigma min Vandermonde matrix first} and \eqref{eq: slight improvement condition number} can be improved to
\begin{equation}
	\sigma_\mathrm{min}^2(\vandermonde{\nbSamples \times \nbNodes}{\node_1, \node_2, \ldots, \node_\nbNodes}) \geq \nbSamples + 1/2 - 1/\delta^{(w)} \label{eq: improvement sigma min unit circle with A} 
\end{equation}
and
\begin{equation}
	\kappa(\vandermonde{\nbSamples \times \nbNodes}{\node_1, \node_2, \ldots, \node_\nbNodes}) \leq \sqrt{\frac{\nbSamples-1+1/\delta^{(w)}}{\nbSamples+1/2-1/\delta^{(w)}}}, \label{eq: condition to hold precision} 
\end{equation}
respectively, by letting $A \rightarrow 1$ in \eqref{eq: lower bound sigma min particular case} and \eqref{eq: main result particular case}, respectively, which leads to the announced  improvement of the Selberg--Moitra bound. Indeed, $\sigma_\mathrm{min}^2(\vandermondeMat{\nbSamples \times \nbNodes})$ and $\kappa(\vandermondeMat{\nbSamples \times \nbNodes})$ are continuous functions of $A$ for $\nbSamples > 1/\delta^{(w)} -1/2$.  
We can therefore establish \eqref{eq: improvement sigma min unit circle with A} and \eqref{eq: condition to hold precision} by taking the limits
\begin{align}
	\frac{1 - A^{2(\nbSamples+1/2 - 1/\delta^{(w)})}}{\delta^{(w)}(A^{-2/\delta^{(w)}} - 1)A^2} &= \underbrace{\frac{e^{2(\nbSamples+1/2-1/\delta^{(w)})\ln(A)} - 1}{2(\nbSamples+1/2-1/\delta^{(w)})\ln(A)}}_{\xrightarrow[A \rightarrow 1]{} 1} \cdot \underbrace{\frac{-(2/\delta^{(w)})\ln(A)}{e^{-(2/\delta^{(w)})\ln(A)} - 1}}_{\xrightarrow[A \rightarrow 1]{} 1} \notag \\
	&\hspace{2cm}\cdot \underbrace{\frac{\nbSamples+1/2-1/\delta^{(w)}}{A^2}}_{\xrightarrow[A \rightarrow 1]{} \nbSamples+1/2 - \frac{1}{\delta^{(w)}}} \notag \\
	&\xrightarrow[A \rightarrow 1]{} \nbSamples\!+\!\frac{1}{2}\!-\!\frac{1}{\delta^{(w)}} \label{eq: limit computation 1} 
\end{align}
and 
\begin{align}
	\frac{A^{-2/\delta^{(w)}}\!\left(1 - A^{2(\nbSamples-1+1/\delta^{(w)})}\right)}{\delta^{(w)}(A^{-2/\delta^{(w)}} - 1)A} &= \underbrace{e^{-(2\delta^{(w)})\ln(A)}}_{\xrightarrow[A \rightarrow 1]{} 1}\cdot\underbrace{\frac{1 - A^{2(\nbSamples-1+1/\delta^{(w)})}}{\delta^{(w)}(A^{-2/\delta^{(w)}} - 1)A}}_{\xrightarrow[A \rightarrow 1]{} \nbSamples-1+1/\delta^{(w)}} \notag \\
		&\xrightarrow[A \rightarrow 1]{} \nbSamples-1+\frac{1}{\delta^{(w)}}, \label{eq: limit computation 2} 
\end{align}
respectively. 
We note that \eqref{eq: condition to hold precision} holds under the condition $\nbSamples > 1/\delta^{(w)} - 1/2$. This improvement is interesting as the condition for validity of the bound on $\kappa(\vandermondeMat{\nbSamples \times \nbNodes})$ no longer excludes the case of square Vandermonde matrices, as is the case for the original Selberg--Moitra bound and our first improvement thereof provided in \eqref{eq: slight improvement condition number}. To see this, simply note that thanks to $\nbSamples > 1/\delta^{(w)}-1/2 \geq \nbNodes -1/2$, the square case $\nbSamples = \nbNodes$ is now allowed as here $\nbSamples = \nbNodes > \nbNodes - 1/2$. Owing to $\delta^{(w)} \leq 1/\nbNodes$ this comes, however, at the cost of the nodes being almost equally spaced. 

We next investigate the qualitative dependence of our bounds on the quantities $\nbSamples$, $\delta_k^{(w)}$, and $\abs{\node_k}$. To this end, we first show that $\varphi_\nbSamples(\abs{z_k})$ is non-decreasing in $\abs{\node_k}$, for fixed $\nbSamples$, and non-increasing in $\nbSamples$, for fixed $\abs{\node_k}$. While the latter follows by inspection, to see the former, we write $\varphi_\nbSamples(\abs{\node_k}) = \nbSamples f(\abs{\node_k}^{2\nbSamples})$ and note that
\begin{equation*}
	f(x) \triangleq \begin{cases} \displaystyle \frac{x-1}{\ln(x)}, & x \in (0, 1) \\
		1, & x = 1 \end{cases}
\end{equation*}
is non-decreasing. Consequently, the lower bound \eqref{eq: lower bound sigma min general} increases and the upper bound \eqref{eq: upper bound sigma max general} decreases as the nodes $\node_k = \abs{\node_k}\!e^{2\pi i\xi_k}$ move closer to the unit circle. 
 Furthermore, $\mathscr{L}(\nbSamples, \abs{\boldsymbol{z}}\!, \boldsymbol{\delta}^{(w)})$ and $\mathscr{U}(\nbSamples, \abs{\boldsymbol{z}}\!, \boldsymbol{\delta}^{(w)})$ are increasing and decreasing in $\delta^{(w)}$, respectively. This allows us to conclude that the upper bound on $\kappa(\vandermondeMat{\nbSamples \times \nbNodes})$ in \eqref{eq: main result general} decreases as the nodes $\node_k = \abs{\node_k}\!e^{2\pi i\xi_k}$ get closer to the unit circle and/or the node frequencies $\xi_k$ are more separated. 
Indeed, we have
\begin{equation*}
	\frac{1+\abs{\node_k}^{2\nbSamples}}{\varphi_\nbSamples(\abs{\node_k})} = -2\ln\abs{\node_k} \!\left(\frac{1 + \abs{\node_k}^{2\nbSamples}}{1 - \abs{\node_k}^{2\nbSamples}}\right) \reversetriangleq h(\abs{\node_k})
\end{equation*}
and note that the function $h \colon x \mapsto -2\ln(x)(1+x^{2\nbSamples})/(1-x^{2\nbSamples})$ is non-increasing. 
The condition \eqref{eq: condition main result general} therefore requires that the wrap-around distance $\delta_k^{(w)}$ increase as $\abs{\node_k}$ gets smaller. 
Specifically, \eqref{eq: condition main result general} is violated if there exists a node $\node_k$ with small modulus $\abs{\node_k}$ together with another node $\node_\ell$ (of arbitrary modulus $\abs{\node_\ell} \leq 1$) so that the wrap-around distance between $\xi_k$ and $\xi_\ell$, i.e., $\min_{n \in \Z} \abs{\xi_k - \xi_\ell + n}$, is small.  This shows that a large minimum distance
\begin{equation*}
	\sigma_k \triangleq \min_{\substack{1 \leq \ell \leq \nbNodes \\  \ell \neq k}} \abs{\node_k - \node_\ell}
\end{equation*}
alone is not enough to guarantee \eqref{eq: condition main result general}. 
Moreover, condition \eqref{eq: condition main result general} excludes the case of nodes placed on a ray emanating from the origin, as here the wrap-around distance equals zero. 
Finally, we emphasize that owing to the large constant $42/\pi$ in \eqref{eq: condition main result general}, which stems from the Montgomery--Vaaler result \eqref{eq: second generalization hilbert's inequality} not being best possible, \eqref{eq: condition main result general} is quite restrictive as it will be satisfied only for nodes $\node_k$ that are very close to the unit circle. 
For $\abs{\node_1} = \abs{\node_2} = \ldots = \abs{\node_\nbNodes} = A$, we not only get a much larger range of validity for our upper bound \eqref{eq: main result particular case} on $\kappa(\vandermonde{\nbSamples \times \nbNodes}{\node_1, \node_2, \ldots, \node_\nbNodes})$ than for the general upper bound \eqref{eq: main result general}, but we also obtain sharper bounds on $\sigma_\mathrm{min}^2(\vandermonde{\nbSamples \times \nbNodes}{\node_1, \node_2, \ldots, \node_\nbNodes})$ and $\sigma_\mathrm{max}^2(\vandermonde{\nbSamples \times \nbNodes}{\node_1, \node_2, \ldots, \node_\nbNodes})$ as the corresponding results are based on \eqref{eq: generalization hilbert's inequality case lambda constant}, which, as pointed out in \cite{Graham1981}, is best possible. One would hope that the constant $42/\pi$ in \eqref{eq: second generalization hilbert's inequality} could be improved to be closer to the corresponding constant $1/2$ in \eqref{eq: first generalization hilbert's inequality} or that $42/\pi$ could be turned into a smaller constant which would possibly depend on $\min_{1 \leq k \leq \nbNodes} \abs{\node_k}$ and/or $\max_{1 \leq k \leq \nbNodes} \abs{\node_k}$ as in the Graham--Vaaler result \eqref{eq: generalization hilbert's inequality case lambda constant}.

\subsection{Comparison to Baz\'an's bound}
\label{sec: discussion bazan's results} 

We finally compare our bounds \eqref {eq: main result general} and \eqref {eq: main result particular case} to Baz\'an's bound \citep{Bazan2000} on $\kappa(\vandermonde{\nbSamples \times \nbNodes}{\node_1, \node_2, \ldots, \node_\nbNodes})$ and start by reviewing Baz\'an's bound.  
It is shown in~\cite[Thm.~6]{Bazan2000} that the spectral condition number of $\vandermondeMat{\nbSamples \times \nbNodes}$, $\nbSamples \geq \nbNodes$, with nodes $\node_k$ in the unit disk satisfies
\begin{equation}
	\frac{\sigma_\mathrm{max}(\mathbf{G}_{\nbSamples})}{A_\mathrm{max}} \leq \kappa(\vandermondeMat{\nbSamples \times \nbNodes}) \leq \frac{1}{2}\left(\eta + \sqrt{\eta^2 - 4}\right)\!.
	\label{eq: bazan quality upper bound 0} 
\end{equation}
Here, $A_\mathrm{max} \triangleq \displaystyle\max_{1 \leq k \leq \nbNodes} \abs{\node_k}$ and $\mathbf{G}_{\nbSamples} \in \C^{\nbNodes \times \nbNodes}$ is a matrix constructed as follows. Let $\widehat{\boldsymbol{f}}\!_\nbSamples \in \C^\nbSamples$ be the minimum $\ell^2$-norm solution of the linear system of equations $(\vandermondeMat{\nbSamples \times \nbNodes})^T\boldsymbol{f} = \boldsymbol{z}_\nbSamples$, where $\boldsymbol{\node}_\nbSamples \triangleq \left(\node_1^\nbSamples\ \node_2^\nbSamples\ \ldots\ \node_\nbNodes^\nbSamples\right)^T \in \C^\nbNodes$. 
Set 
\begin{equation*}
	\mathbf{G}_\nbSamples \triangleq \mathbf{W}_{\nbSamples \times \nbNodes}^H\boldsymbol{\Gamma}_\nbSamples\mathbf{W}_{\nbSamples \times \nbNodes},
\end{equation*}
where
\begin{equation*}
	\mathbf{W}_{\nbSamples \times \nbNodes} \triangleq \overline{\vandermondeMat{\nbSamples \times \nbNodes}\left(\left(\vandermondeMat{\nbSamples \times \nbNodes}\right)^H\vandermondeMat{\nbSamples \times \nbNodes}\right)^{-1/2}} \in \C^{\nbSamples \times \nbNodes},
\end{equation*}
$\boldsymbol{\Gamma}_\nbSamples \triangleq \left(\mathbf{e}_2\ \mathbf{e}_3\ \ldots \ \mathbf{e}_\nbSamples \ \widehat{\boldsymbol{f}}\!_\nbSamples\right) \in \C^{\nbSamples \times \nbSamples}$, and $\boldsymbol{e}_n \in \C^\nbSamples$, for $n \in \{2, 3, \ldots, \nbSamples\}$, is the $n$th unit vector whose elements are all zero apart from the $n$th entry which equals $1$. 
The quantity $\eta$ is given by
\begin{equation}
	\eta \triangleq \nbNodes\left(1 + \frac{D_\nbSamples^2}{(\nbNodes-1)\sigma^2}\right)^{\frac{\nbNodes-1}{2}}\!\!\left(\frac{\psi_\nbSamples(A_\mathrm{max})}{\psi_\nbSamples(A_\mathrm{min})}\right)^{1/2}\!\! - \nbNodes + 2,
	\label{eq: bazan quality upper bound} 
\end{equation}
where $D_\nbSamples^2 \triangleq \norm{\mathbf{G}_\nbSamples}^2 - \big(\abs{\node_1}^2 + \ldots + \abs{\node_\nbNodes}^2\big)$ is the so-called departure of $\mathbf{G}_\nbSamples$ from normality, $\sigma$ is the minimum distance between the nodes $\node_1, \node_2, \ldots, \node_\nbNodes$ as defined in \eqref{eq: minimum distance complex number}, $\psi_\nbSamples(x) \triangleq \sum_{n = 0}^{\nbSamples-1} x^{2n}$, and $A_\mathrm{min} \triangleq \displaystyle\min_{1 \leq k \leq \nbNodes} \abs{\node_k}$. 

Owing to the complicated and, in particular, implicit nature of Baz\'an's bound, it appears difficult to draw crisp conclusions therefrom on the behavior of $\kappa(\vandermondeMat{\nbSamples \times \nbNodes})$ as a function of the nodes $\node_1, \node_2, \ldots, \node_\nbNodes$ and $\nbSamples$.   
It is, however, possible to extract a statement of asymptotic (in $\nbSamples$, for fixed $\nbNodes$) nature from \eqref{eq: bazan quality upper bound 0}. Specifically, it is stated in \citep[Lem.~7]{Bazan2000} that
\begin{align*}
	(\nbNodes-1) + \frac{\prod_{k = 1}^\nbNodes \!\abs{\node_k}^2\!}{1 + \|\widehat{\boldsymbol{f}}\!_\nbSamples\big\|_2^2} - \sum_{k = 1}^\nbNodes \!\abs{\node_k}^2\! \leq D_\nbSamples^2 \leq (\nbNodes-1) + \big\|\widehat{\boldsymbol{f}}\!_\nbSamples\big\|_2^2 + \prod_{k = 1}^\nbNodes \!\abs{\node_k}^2\! - \sum_{k = 1}^\nbNodes \!\abs{z_k}^2\!,
	\label{eq: upper bound DN} 
\end{align*}
for $\nbSamples \geq \nbNodes$. Since 
$\lim_{\nbSamples \rightarrow \infty} \norm{\widehat{\boldsymbol{f}}\!_\nbSamples}_2^2 = 0$ \citep[Thm.~2]{Bazan2000}, we get
\begin{equation}
	\lim_{\nbSamples \rightarrow \infty} D_\nbSamples^2 = (\nbNodes-1) + \prod_{k = 1}^\nbNodes \abs{\node_k}^2 - \sum_{k = 1}^\nbNodes \abs{\node_k}^2.
	\label{eq: limit of DN as N tends to infinity} 
\end{equation}
Now, since for fixed $u_\ell \in [0, 1]$, $\ell \in \{1, 2, \ldots, \nbNodes\} \setminus \{k\}$, the function $u_k \mapsto \prod_{\ell = 1}^\nbNodes u_\ell - \sum_{\ell = 1}^\nbNodes u_\ell$ is non-increasing, the limit in \eqref{eq: limit of DN as N tends to infinity} increases when the nodes $\node_1, \node_2, \ldots, \node_\nbNodes$ move closer to the unit circle. Moreover, \eqref{eq: limit of DN as N tends to infinity} equals zero when $\abs{\node_1} = \abs{\node_2} = \ldots = \abs{\node_\nbNodes} = 1$. Based on \eqref{eq: limit of DN as N tends to infinity} Baz\'an obtained the large--$\nbSamples$ asymptotes of the lower and upper bounds in \eqref{eq: bazan quality upper bound 0}. Specifically, it is shown in \cite[Lem.~8, Cor.~9]{Bazan2000} that the limit $\kappa_a \triangleq \lim_{\nbSamples \rightarrow \infty} \kappa(\vandermondeMat{\nbSamples \times \nbNodes})$ exists, and for $A_\mathrm{min} \leq A_\mathrm{max} < 1$ satisfies
\begin{equation*}
	\frac{1}{A_\mathrm{max}} \leq \kappa_a \leq \frac{1}{2}\left(\eta_a + \sqrt{\eta_a^2 - 4}\right),
\end{equation*}
where
\begin{align*}
	\eta_a \triangleq \nbNodes\left(1 + \frac{1}{\sigma^2} + \frac{\prod_{k = 1}^\nbNodes \abs{\node_k}^2 - \sum_{k = 1}^\nbNodes \abs{\node_k}^2}{(\nbNodes-1)\sigma^2}\right)^{\frac{\nbNodes-1}{2}}\!\!\left(\frac{1 - A_\mathrm{min}^2}{1 - A_\mathrm{max}^2}\right)^{1/2}\! - \nbNodes + 2.
\end{align*}
In addition, it is proven in \cite[Cor.~10]{Bazan2000} that for $\abs{\node_1} = \abs{\node_2} = \ldots = \abs{\node_\nbNodes} = A < 1$ and $1 - A^2 \leq \sigma^2$, one has
\begin{equation}
	\kappa_a \leq \nbNodes 2^{\frac{\nbNodes-1}{2}} - \nbNodes+2.
	\label{eq: Bazan's asymptitic equal A} 
\end{equation}
For $A_\mathrm{min} = A_\mathrm{max} = 1$, it follows from \eqref{eq: limit of DN as N tends to infinity} that $\lim_{\nbSamples \rightarrow \infty} D_\nbSamples^2 = 0$, and hence, by \eqref{eq: bazan quality upper bound 0}, that $\kappa_a \leq 1$, which, together with $\kappa_a \geq 1$, implies $\kappa_a = 1$. 

While these results provide insight into the asymptotic behavior of $D_\nbSamples^2$ as $\nbSamples \rightarrow \infty$, an analysis of the speed of convergence of $D_\nbSamples^2$ to the right-hand side (RHS) of \eqref{eq: limit of DN as N tends to infinity} does not seem to be available in the literature. It therefore appears difficult to draw conclusions about the finite--$\nbSamples$ behavior of $D_\nbSamples^2$. 
In fact, $D_\nbSamples^2$ seems as difficult to characterize, in the finite--$\nbSamples$ regime, as the condition number itself. 
Moreover, the numerical evaluation of $D_\nbSamples^2$ requires the computation of $\widehat{\boldsymbol{f}}_\nbSamples$, which in turn requires solving the linear system of equations $(\vandermondeMat{\nbSamples \times \nbNodes})^T\boldsymbol{f} = \boldsymbol{z}_\nbSamples$. When $\kappa(\vandermondeMat{\nbSamples \times \nbNodes})$ is large, the computation of $D_\nbSamples^2$ and therefore the numerical evaluation of Baz\'an's upper bound can become numerically unstable.

Finally, we compare Baz\'an's upper bound with our results. We start by noting that, owing to the large constant $42/\pi$ in condition \eqref{eq: condition main result general} needed for our bound \eqref{eq: main result general} to hold, Baz\'an's bound is valid for more general node configurations. 
In particular, since our bound \eqref{eq: main result general} holds only for nodes that are very close to the unit circle, a comparison to Baz\'an's bound in the general case is not particularly meaningful. For the special case $\abs{\node_1} = \abs{\node_2} = \ldots = \abs{\node_\nbNodes} = A$, however, our bound \eqref{eq: main result particular case} is based on the Graham--Vaaler result \eqref{eq: generalization hilbert's inequality case lambda constant} for which the (non-optimal) constant $42/\pi$ does not appear.  
Specifically, the asymptote (in $\nbSamples$, with $\nbNodes$ fixed) of our bound \eqref{eq: main result particular case} satisfies
\begin{equation}
	\kappa_a \leq A^{1/2-1/\delta^{(w)}}.
	\label{eq: our bound asymptotic equal A} 
\end{equation}
A general comparison of \eqref{eq: Bazan's asymptitic equal A} and \eqref{eq: our bound asymptotic equal A} is difficult as the two bounds do not depend on the same quantities. We can, however, make specific exemplary statements. For example, for $A \geq 0.8$ and $\delta^{(w)} = 1/\nbNodes$, \eqref{eq: our bound asymptotic equal A} implies $\kappa_a \leq 1.25^{\nbNodes-1/2}$, which improves upon \eqref{eq: Bazan's asymptitic equal A} for $\nbNodes \geq 1$. On the other hand, for $A \geq 1/2$ and equally spaced nodes, i.e., $\delta^{(w)} = 1/\nbNodes$, \eqref{eq: our bound asymptotic equal A} becomes $\kappa_a \leq 2^{\nbNodes - 1/2}$, so that Baz\'an's bound \eqref{eq: Bazan's asymptitic equal A}, for $\nbNodes \geq 4$, is better in that case.
Detailed numerical comparisons between our bound \eqref{eq: main result particular case} and Baz\'an's bound for $\abs{\node_1} = \abs{\node_2} = \ldots = \abs{\node_\nbNodes} = A$ are provided in the next section. 

We finally note that our upper bound on $\kappa(\vandermondeMat{\nbSamples \times \nbNodes})$ is obtained by combining a lower bound on $\sigma_\mathrm{min}^2(\vandermondeMat{\nbSamples \times \nbNodes})$ and an upper bound on $\sigma_\mathrm{max}^2(\vandermondeMat{\nbSamples \times \nbNodes})$ in \eqref{eq: definition spectral condition number}. 
Baz\'an, on the other hand, directly provides a condition number upper bound and does not report individual bounds on $\sigma_\mathrm{min}^2(\vandermondeMat{\nbSamples \times \nbNodes})$ and $\sigma_\mathrm{max}^2(\vandermondeMat{\nbSamples \times \nbNodes})$.

%% file: numericalResults.tex

\section{Numerical results}
\label{sec: numerical results}

We consider the case $\abs{\node_1} = \abs{\node_2} = \ldots = \abs{\node_\nbNodes} = A$ and compare our bound \eqref{eq: main result particular case} to Baz\'an's bound by averaging over $500$ randomly selected node configurations. Specifically, for each $d \in (0, 0.5]$, we construct an $\nbSamples \times \nbNodes$ Vandermonde matrix with nodes $\node_k \triangleq Ae^{2\pi i\xi_k}$, for $k \in \{1, 2, \ldots, \nbNodes\}$, where $\nbSamples = 100$, $A$ varies from $0.1$ to $1$, $\nbNodes$ is chosen uniformly at random in the set $\left\{2, 3, \ldots, \lfloor 1/d\rfloor\right\}$, $\xi_k = k/\nbNodes+r_k$, and $r_k$ is chosen uniformly at random in the interval $\left[0, 1/\nbNodes-d\right]$. The minimum wrap-around distance $\delta^{(w)}$ between the $\xi_k$ is therefore guaranteed to satisfy $\delta^{(w)} \geq d$.  
The results are depicted in Figures~\ref{fig: comparison Bazan's bound} and \ref{fig: comparison Bazan's bound 1}. We observe that for $d \leq 0.1$, our bound is much tighter than Baz\'an's bound. For $d = 0.18$ (Figure \ref{fig: comparison Bazan's bound 1}d), Baz\'an's bound is tighter than our bound for small values of $A$.  
For $d \geq 0.2$ (not depicted), Baz\'an's bound is slightly tighter than our bound, and this for all values of $A$. However, as $\delta^{(w)} \leq 1/\nbNodes$ and our construction guarantees that $\delta^{(w)} \geq d$, $d \geq 0.2$ implies $\nbNodes \leq 5$, which is small compared to $\nbSamples = 100$; the assumption $d \geq 0.2$ leading to Baz\'an's bound being tighter than our bound \eqref{eq: main result particular case} is therefore quite restrictive. We finally note that the curve corresponding to Baz\'an's bound in the case $d = 0.05$ (Figure~\ref{fig: comparison Bazan's bound 1}a) is wiggly for small values of $A$. This is because numerical evaluation of Baz\'an's bound involves solving the linear system $(\vandermondeMat{\nbSamples \times \nbNodes})^T\boldsymbol{f} = \boldsymbol{z}_\nbSamples$ and $\vandermondeMat{\nbSamples \times \nbNodes}$ is ill-conditioned in these cases, leading to numerical instability.

\begin{figure}
	\ifthenelse{\equal{\versionColor}{true}}{
		\ifthenelse{\equal{\pdfFigs}{true}}{
			\includegraphics[width = 0.94\textwidth]{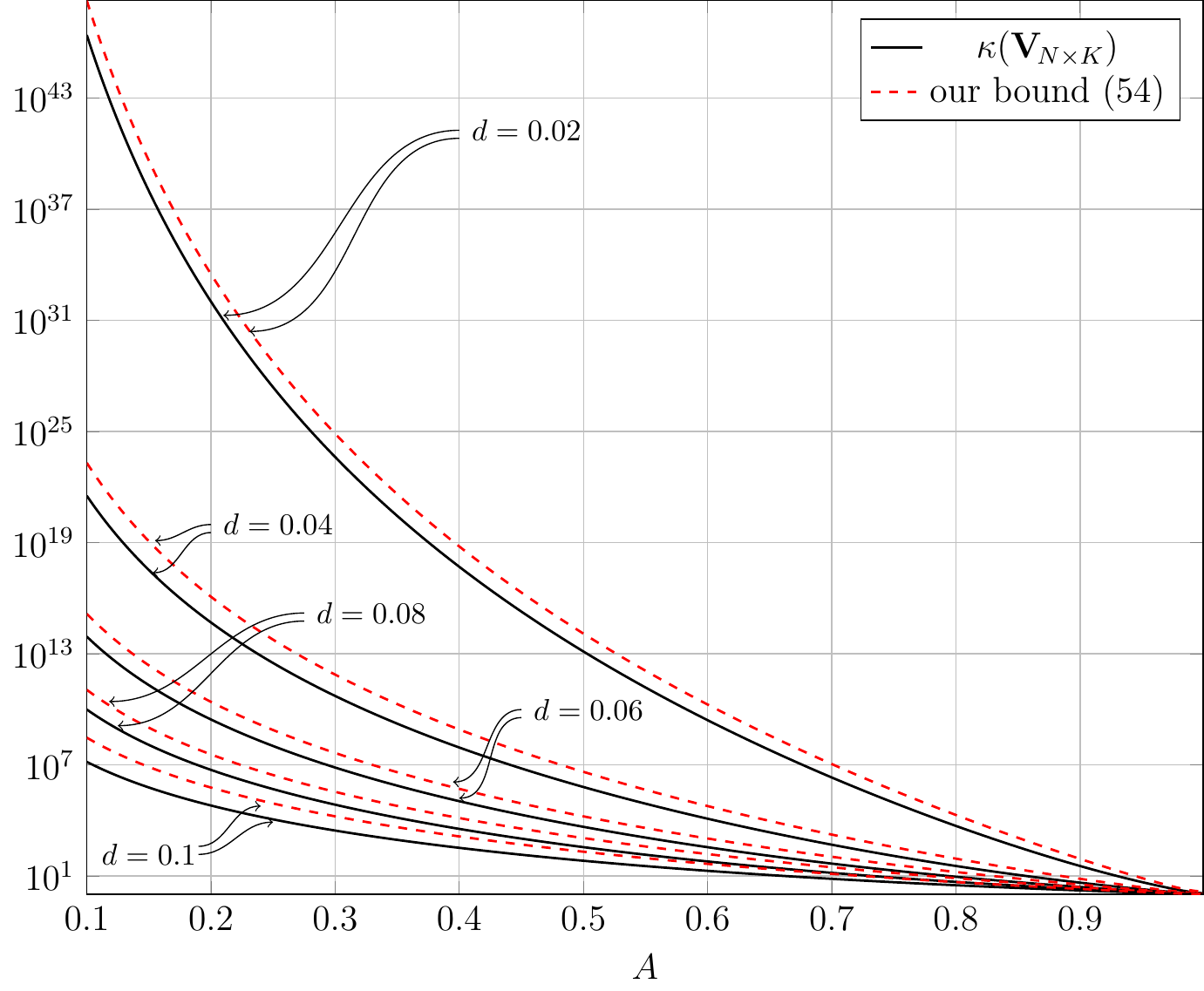}
		}{
			\ifthenelse{\equal{\generateFigs}{yes}}{\tikzsetnextfilename{numericalResults1}}{}
			\begin{tikzpicture}
				\begin{semilogyaxis}
			                    [no markers,
			                    width = \textwidth, height = 12cm,
			                    xlabel = {$A$},
			                    grid = both,
			                    enlarge x limits = false, enlarge y limits = false,
						]       
					   \addplot[thick, black] file {data/conditionNumber_N100_100delta2nbTrials500.dat};       			
			                    \addplot[thick, red, dashed] file {data/ourBound_N100_100delta2nbTrials500.dat};
			
			                    \addplot[thick, black] file {data/conditionNumber_N100_100delta4nbTrials500.dat};
			                    \addplot[thick, red, dashed] file {data/ourBound_N100_100delta4nbTrials500.dat};
			                    
			                    \addplot[thick, black] file {data/conditionNumber_N100_100delta6nbTrials500.dat};
			                    \addplot[thick, red, dashed] file {data/ourBound_N100_100delta6nbTrials500.dat};
			                    
			                    \addplot[thick, black] file {data/conditionNumber_N100_100delta8nbTrials500.dat};
			                     \addplot[thick, red, dashed] file {data/ourBound_N100_100delta8nbTrials500.dat};
			                    
			                    \addplot[thick, black] file {data/conditionNumber_N100_100delta10nbTrials500.dat};
			                    \addplot[thick, red, dashed] file {data/ourBound_N100_100delta10nbTrials500.dat};
			                    
			                     \addlegendentry{$\kappa(\vandermondeMat{\nbSamples \times \nbNodes})$}
			                     \addlegendentry{our bound \eqref{eq: main result particular case}}

					   \node[anchor=west] (curve2_1) at (300,95) {{\footnotesize $d = 0.02$}};
			                    \draw (curve2_1) edge[out=180,in=0,->] (110,72);
			                    \node[anchor=west] (curve2_2) at (300,94) {};
			                    \draw (curve2_2) edge[out=180,in=0,->] (131,70);
			                    
			                    \node[anchor=west] (curve4_1) at (100,46) {{\footnotesize $d = 0.04$}};
			                    \draw (curve4_1) edge[out=180,in=0,->] (55,44);
			                    \node[anchor=west] (curve4_2) at (100,45) {};
			                    \draw (curve4_2) edge[out=180,in=0,->] (53,40);
			                    
			                    \node[anchor=west] (curve6_1) at (350,23) {{\footnotesize $d = 0.06$}};
			                    \draw (curve6_1) edge[out=180,in=0,->] (295,14);
			                    \node[anchor=west] (curve6_2) at (350,22) {};
			                    \draw (curve6_2) edge[out=180,in=0,->] (300,12);
						
			                    \node[anchor=west] (curve8_1) at (175,35) {{\footnotesize $d = 0.08$}};
			                    \draw (curve8_1) edge[out=180,in=0,->] (18,24);
			                    \node[anchor=west] (curve8_2) at (175,34) {};
			                    \draw (curve8_2) edge[out=180,in=0,->] (25,21);
			                    
			                    \node[anchor=east] (curve10_1) at (90,5) {{\footnotesize $d = 0.1$\!\!}};
			                    \draw (curve10_1) edge[out=0,in=180,->] (150,9);
			                    \node[anchor=east] (curve10_2) at (90,6) {};
			                    \draw (curve10_2) edge[out=0,in=180,->] (140,11);
				\end{semilogyaxis}
			\end{tikzpicture}
		}
	}{
		\ifthenelse{\equal{\pdfFigs}{true}}{
			\includegraphics[width = 0.95\textwidth]{numericalResults1BW.pdf}
		}{
			\ifthenelse{\equal{\generateFigs}{yes}}{\tikzsetnextfilename{numericalResults1BW}}{}
			\begin{tikzpicture}
				\begin{semilogyaxis}
			                    [no markers,
			                    width = \textwidth, height = 12cm,
			                    xlabel = {$A$},
			                    grid = both,
			                    enlarge x limits = false, enlarge y limits = false,
						]       
					   \addplot[thick, black] file {data/conditionNumber_N100_100delta2nbTrials500.dat};       			
			                    \addplot[thick, gray, dashed] file {data/ourBound_N100_100delta2nbTrials500.dat};
			
			                    \addplot[thick, black] file {data/conditionNumber_N100_100delta4nbTrials500.dat};
			                    \addplot[thick, gray, dashed] file {data/ourBound_N100_100delta4nbTrials500.dat};
			                    
			                    \addplot[thick, black] file {data/conditionNumber_N100_100delta6nbTrials500.dat};
			                    \addplot[thick, gray, dashed] file {data/ourBound_N100_100delta6nbTrials500.dat};
			                    
			                    \addplot[thick, black] file {data/conditionNumber_N100_100delta8nbTrials500.dat};
			                     \addplot[thick, gray, dashed] file {data/ourBound_N100_100delta8nbTrials500.dat};
			                    
			                    \addplot[thick, black] file {data/conditionNumber_N100_100delta10nbTrials500.dat};
			                    \addplot[thick, gray, dashed] file {data/ourBound_N100_100delta10nbTrials500.dat};
			                    
			                     \addlegendentry{$\kappa(\vandermondeMat{\nbSamples \times \nbNodes})$}
			                     \addlegendentry{our bound \eqref{eq: main result particular case}}

					   \node[anchor=west] (curve2_1) at (300,95) {{\footnotesize $d = 0.02$}};
			                    \draw (curve2_1) edge[out=180,in=0,->] (110,72);
			                    \node[anchor=west] (curve2_2) at (300,94) {};
			                    \draw (curve2_2) edge[out=180,in=0,->] (131,70);
			                    
			                    \node[anchor=west] (curve4_1) at (100,46) {{\footnotesize $d = 0.04$}};
			                    \draw (curve4_1) edge[out=180,in=0,->] (55,44);
			                    \node[anchor=west] (curve4_2) at (100,45) {};
			                    \draw (curve4_2) edge[out=180,in=0,->] (53,40);
			                    
			                    \node[anchor=west] (curve6_1) at (350,23) {{\footnotesize $d = 0.06$}};
			                    \draw (curve6_1) edge[out=180,in=0,->] (295,14);
			                    \node[anchor=west] (curve6_2) at (350,22) {};
			                    \draw (curve6_2) edge[out=180,in=0,->] (300,12);
						
			                    \node[anchor=west] (curve8_1) at (175,35) {{\footnotesize $d = 0.08$}};
			                    \draw (curve8_1) edge[out=180,in=0,->] (18,24);
			                    \node[anchor=west] (curve8_2) at (175,34) {};
			                    \draw (curve8_2) edge[out=180,in=0,->] (25,21);
			                    
			                    \node[anchor=east] (curve10_1) at (90,5) {{\footnotesize $d = 0.1$\!\!}};
			                    \draw (curve10_1) edge[out=0,in=180,->] (150,9);
			                    \node[anchor=east] (curve10_2) at (90,6) {};
			                    \draw (curve10_2) edge[out=0,in=180,->] (140,11);
				\end{semilogyaxis}
			\end{tikzpicture}
		}	
	}
	\caption{Comparison of our upper bound \eqref{eq: main result particular case} to the true (average) condition number.}
	\label{fig: comparison Bazan's bound} 
\end{figure}

\begin{figure}
	\ifthenelse{\equal{\versionColor}{true}}{
		\subfigure[$d = 0.05$]{
			\ifthenelse{\equal{\pdfFigs}{true}}{
				\includegraphics[width = 0.95\textwidth, height=0.3\textwidth]{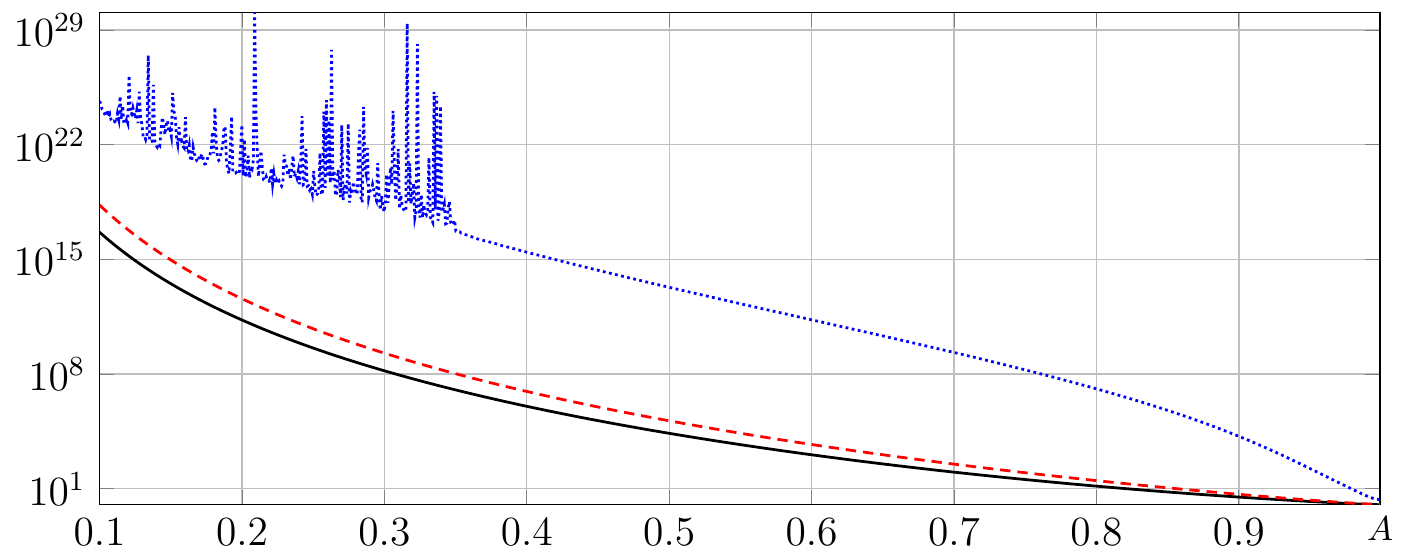}
			}{
				\ifthenelse{\equal{\generateFigs}{yes}}{\tikzsetnextfilename{numericalResults2-2}}{}
				\begin{tikzpicture}
					\begin{semilogyaxis}
				                    [no markers,
				                    width = \textwidth, height = 0.3\textheight,
				                    xlabel = {\footnotesize $A$},
				                    grid = both,
				                    enlarge x limits = false, enlarge y limits = false, 
				             	   every axis x label/.style={ 
				                    	at={(current axis.right of origin)},
				                    	anchor=west,above=-5mm},
						   ]              			
				                    \addplot[thick, black] file {data/conditionNumber_N100_100delta5nbTrials500.dat};
				                    \addplot[thick, red, densely dashed] file {data/ourBound_N100_100delta5nbTrials500.dat};
				                    \addplot[thick, blue, densely dotted] file {data/bazan_N100_100delta5nbTrials500.dat};
					\end{semilogyaxis}
				\end{tikzpicture}
			}
		}
		\subfigure[$d = 0.10$]{
			\ifthenelse{\equal{\pdfFigs}{true}}{
				\includegraphics[width = 0.95\textwidth]{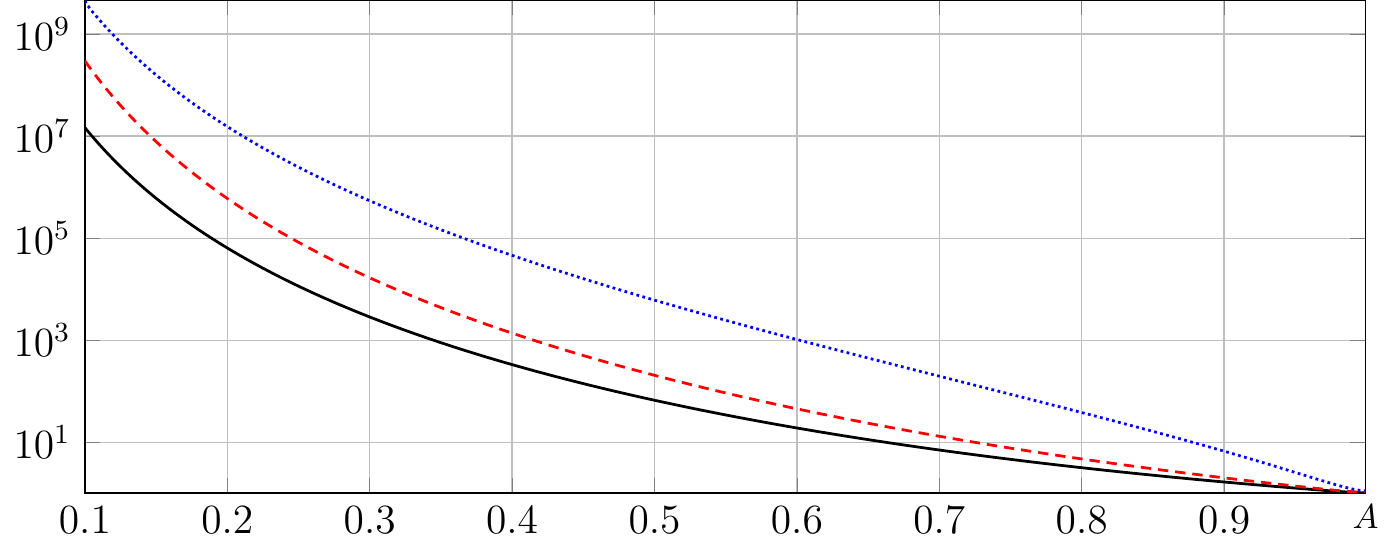}
			}{
				\ifthenelse{\equal{\generateFigs}{yes}}{\tikzsetnextfilename{numericalResults2-3}}{}
				\begin{tikzpicture}
					\begin{semilogyaxis}
				                    [no markers,
				                    width = \textwidth, height = 0.3\textheight,
				                    xlabel = {\footnotesize $A$},
				                    grid = both,
				                    enlarge x limits = false, enlarge y limits = false, 
				     	            every axis x label/.style={ 
				                    	at={(current axis.right of origin)},
				                    	anchor=west,above=-5mm},
						   ]              			
				                    \addplot[thick, black] file {data/conditionNumber_N100_100delta10nbTrials500.dat};
				                    \addplot[thick, red, densely dashed] file {data/ourBound_N100_100delta10nbTrials500.dat};
				                    \addplot[thick, blue, densely dotted] file {data/bazan_N100_100delta10nbTrials500.dat};
					\end{semilogyaxis}
				\end{tikzpicture}
			}
		}
		\subfigure[$d = 0.18$]{
			\ifthenelse{\equal{\pdfFigs}{true}}{
				\includegraphics[width = 0.95\textwidth]{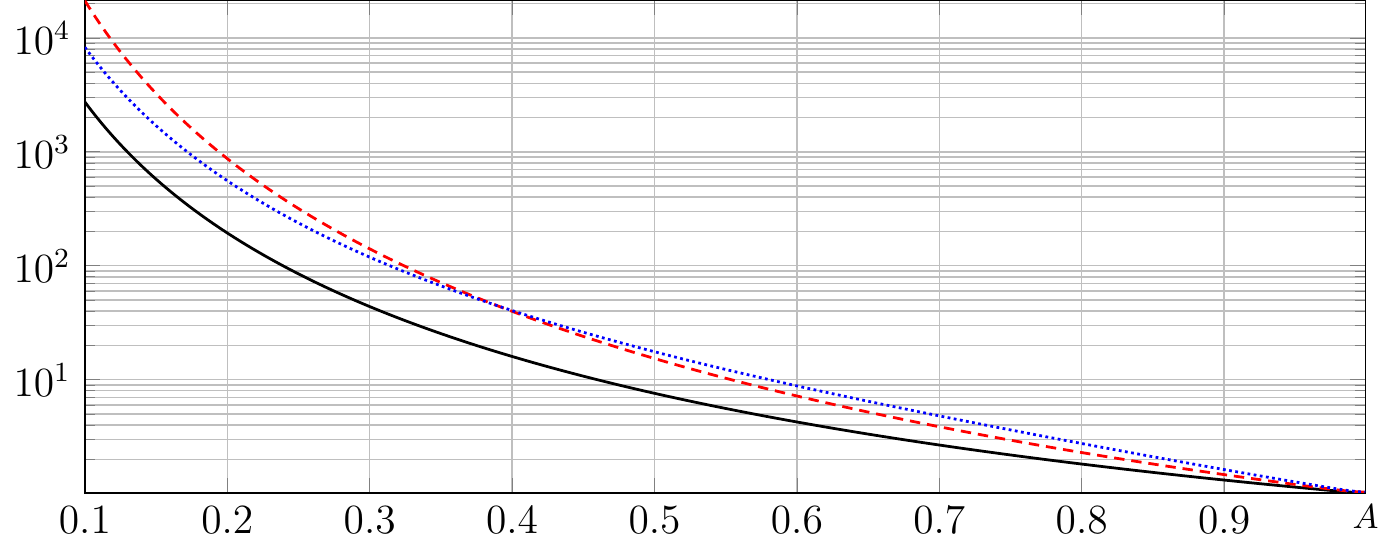}
			}{
				\ifthenelse{\equal{\generateFigs}{yes}}{\tikzsetnextfilename{numericalResults2-4}}{}
				\begin{tikzpicture}
					\begin{semilogyaxis}
				                    [no markers,
				                    width = \textwidth, height = 0.3\textheight,
				                    xlabel = {\footnotesize $A$},
				                    grid = both,
				                    enlarge x limits = false, enlarge y limits = false, 
				                    every axis x label/.style={ 
				                    	at={(current axis.right of origin)},
				                    	anchor=west,above=-5mm},
						   ]              			
				                    \addplot[thick, black] file {data/conditionNumber_N100_100delta18nbTrials500.dat};
				                    \addplot[thick, red, densely dashed] file {data/ourBound_N100_100delta18nbTrials500.dat};
				                    \addplot[thick, blue, densely dotted] file {data/bazan_N100_100delta18nbTrials500.dat};
					\end{semilogyaxis}
				\end{tikzpicture}
			}
		}
	}{
		\subfigure[$d = 0.05$]{
			\ifthenelse{\equal{\pdfFigs}{true}}{
				\includegraphics[width = 0.95\textwidth]{numericalResults2-2BW.pdf}
			}{
				\ifthenelse{\equal{\generateFigs}{yes}}{\tikzsetnextfilename{numericalResults2-2BW}}{}
				\begin{tikzpicture}
					\begin{semilogyaxis}
				                    [no markers,
				                    width = \textwidth, height = 0.3\textheight,
				                    xlabel = {\footnotesize $A$},
				                    grid = both,
				                    enlarge x limits = false, enlarge y limits = false, 
				             	   every axis x label/.style={ 
				                    	at={(current axis.right of origin)},
				                    	anchor=west,above=-5mm},
						   ]              			
				                    \addplot[thick, black] file {data/conditionNumber_N100_100delta5nbTrials500.dat};
				                    \addplot[thick, gray, densely dashed] file {data/ourBound_N100_100delta5nbTrials500.dat};
				                    \addplot[thick, gray, densely dotted] file {data/bazan_N100_100delta5nbTrials500.dat};
					\end{semilogyaxis}
				\end{tikzpicture}
			}
		}
		\subfigure[$d = 0.10$]{
			\ifthenelse{\equal{\pdfFigs}{true}}{
				\includegraphics[width = 0.95\textwidth]{numericalResults2-3BW.pdf}
			}{
				\ifthenelse{\equal{\generateFigs}{yes}}{\tikzsetnextfilename{numericalResults2-3BW}}{}
				\begin{tikzpicture}
					\begin{semilogyaxis}
				                    [no markers,
				                    width = \textwidth, height = 0.3\textheight,
				                    xlabel = {\footnotesize $A$},
				                    grid = both,
				                    enlarge x limits = false, enlarge y limits = false, 
				     	            every axis x label/.style={ 
				                    	at={(current axis.right of origin)},
				                    	anchor=west,above=-5mm},
						   ]              			
				                    \addplot[thick, black] file {data/conditionNumber_N100_100delta10nbTrials500.dat};
				                    \addplot[thick, gray, densely dashed] file {data/ourBound_N100_100delta10nbTrials500.dat};
				                    \addplot[thick, gray, densely dotted] file {data/bazan_N100_100delta10nbTrials500.dat};
					\end{semilogyaxis}
				\end{tikzpicture}
			}
		}
		\subfigure[$d = 0.18$]{
			\ifthenelse{\equal{\pdfFigs}{true}}{
				\includegraphics[width = 0.95\textwidth]{numericalResults2-4BW.pdf}
			}{
				\ifthenelse{\equal{\generateFigs}{yes}}{\tikzsetnextfilename{numericalResults2-4BW}}{}
				\begin{tikzpicture}
					\begin{semilogyaxis}
				                    [no markers,
				                    width = \textwidth, height = 0.3\textheight,
				                    xlabel = {\footnotesize $A$},
				                    grid = both,
				                    enlarge x limits = false, enlarge y limits = false, 
				                    every axis x label/.style={ 
				                    	at={(current axis.right of origin)},
				                    	anchor=west,above=-5mm},
						   ]              			
				                    \addplot[thick, black] file {data/conditionNumber_N100_100delta18nbTrials500.dat};
				                    \addplot[thick, gray, densely dashed] file {data/ourBound_N100_100delta18nbTrials500.dat};
				                    \addplot[thick, gray, densely dotted] file {data/bazan_N100_100delta18nbTrials500.dat};
					\end{semilogyaxis}
				\end{tikzpicture}
			}
		}	
	}
	\caption{Comparison of our upper bound \eqref{eq: main result particular case} (dashed) and Baz\'an's upper bound \cite[Thm.~6]{Bazan2000} (dotted) with the condition number $\kappa(\vandermondeMat{\nbSamples \times \nbNodes})$ (solid).}
	\label{fig: comparison Bazan's bound 1} 
\end{figure}

%% file: appendix1.tex

\section{Proof of Proposition \ref{prop: equivalence hilbert's inequality sh}}
\label{app: proof of proposition first appendix} 

We start by showing \ref{item 2 prop hilbert sh}) $\Rightarrow$ \ref{item 1 prop hilbert sh}), which is accomplished using a scaling argument that generalizes the scaling argument in the proof of \citep[Thm.~2]{Montgomery1974}.   
	To this end, let $M \in \N \!\setminus\! \{0\}$, $\boldsymbol{\alpha} \triangleq \{\alpha_k\}_{k = 1}^M \in \C^M$, consider $\boldsymbol{\rho} \triangleq \{\rho_k\}_{k = 1}^M \in \C^M$ with $\rho_k \triangleq \lambda_k + 2\pi iu_k$, $\lambda_k > 0$, and $u_k \in \R$ such that 
	\begin{equation}
		\delta_k \triangleq \min_{\substack{1 \leq \ell \leq M \\ \ell \neq k}} \abs{u_k - u_\ell} > 0,
		\label{eq: delta proof proposition} 
	\end{equation}
	for all $k \in \{1, 2, \ldots, M\}$, and take $\varepsilon$ satisfying
	\begin{equation}
		0 < \varepsilon < \frac{1}{\displaystyle 2\max_{1 \leq k, \ell \leq M} \abs{u_k - u_\ell}}. 
		\label{eq: varepsilon def} 
	\end{equation}
	We apply \ref{item 2 prop hilbert sh}) with $d_k \triangleq \varepsilon \lambda_k$, $\xi_k \triangleq \varepsilon u_k$, $r_k \triangleq \varepsilon\rho_k = d_k + 2\pi i \xi_k$, and $a_k \triangleq \alpha_k$, for all $k \in \{1, 2, \ldots, M\}$.
	As $\min_{n \in \Z} \abs{\varepsilon (u_k - u_\ell) + n}$ is the distance between $\varepsilon(u_k - u_\ell)$ and its nearest integer and, by \eqref{eq: varepsilon def}, $0 \leq \varepsilon \!\abs{u_k - u_\ell} < 1/2$, we have $\min_{n \in \Z} \abs{\varepsilon (u_k - u_\ell) + n} = \varepsilon\!\abs{u_k - u_\ell}$, and hence
	\begin{equation}
		\delta_k^{(w)} \triangleq \min_{\substack{1 \leq \ell \leq M \\ \ell \neq k}} \min_{n \in \Z} \abs{\varepsilon (u_k - u_\ell) + n} = \varepsilon\min_{\substack{1 \leq \ell \leq M \\ \ell \neq k}} \abs{u_k - u_\ell} = \varepsilon\delta_k,
		\label{eq: equality delta k proposition proof} 
	\end{equation}
	for $k \in \{1, 2, \ldots, M\}$, and
	\begin{equation*}
		\delta^{(w)} = \min_{1 \leq k \leq M} \delta_k^{(w)} =  \min_{1 \leq k \leq M} \varepsilon\delta_k = \varepsilon\delta.  
	\end{equation*}
	It follows from \eqref{eq: delta proof proposition}, \eqref{eq: varepsilon def}, and \eqref{eq: equality delta k proposition proof} that $\delta_k^{(w)} > 0$, for $k \in \{1, 2, \ldots, \nbNodes\}$. Moreover, as $\lambda_k > 0$, by assumption, it follows from $d_k = \varepsilon\lambda_k$ and \eqref{eq: varepsilon def} that $d_k > 0$, for $k \in \{1, 2, \ldots, M\}$.
	The conditions for application of \ref{item 2 prop hilbert sh}) are therefore met, and \eqref{eq: general hilbert sh} yields
	\begin{equation*}
		\sum_{k = 1}^M 2\varepsilon A(\varepsilon\lambda_k, \varepsilon\delta_k, \varepsilon\delta) \abs{a_k}^2 \leq \sum_{k, \ell = 1}^M \frac{\varepsilon a_k\overline{a_\ell}}{\sh(\varepsilon(\rho_k + \overline{\rho_\ell})/2)} \leq \sum_{k = 1}^M 2\varepsilon B(\varepsilon\lambda_k, \varepsilon\delta_k, \varepsilon\delta)\abs{a_k}^2.
	\end{equation*}
	As $\varepsilon A(\varepsilon\lambda_k, \varepsilon\delta_k, \varepsilon\delta) = A(\lambda_k, \delta_k, \delta)$ and $\varepsilon B(\varepsilon \lambda_k, \varepsilon\delta_k, \varepsilon\delta) = B(\lambda_k, \delta_k, \delta)$, for all $k \in \{1, 2, \ldots, M\}$, both by assumption, we have 
	\begin{equation}
		\sum_{k = 1}^M A(\lambda_k, \delta_k, \delta) \abs{a_k}^2 \leq \sum_{k, \ell = 1}^M \frac{\varepsilon a_k\overline{a_\ell}}{2\sh(\varepsilon(\rho_k + \overline{\rho_\ell})/2)} \leq \sum_{k = 1}^M B(\lambda_k, \delta_k, \delta)\abs{a_k}^2.
		\label{eq: last step before conclusion equivalence hilbert} 
	\end{equation}
	Taking the limit $\varepsilon \rightarrow 0$ in \eqref{eq: last step before conclusion equivalence hilbert} and noting that $\lim_{z \rightarrow 0} \sh(z)/z = 1$ and hence 
	\begin{equation*}
		\frac{1}{\rho_k + \overline{\rho_\ell}} = \lim_{\varepsilon \rightarrow 0} \frac{\varepsilon}{2\sh(\varepsilon(\rho_k + \overline{\rho_\ell})/2)},
	\end{equation*}
	we get \eqref{eq: general hilbert fraction}.

	The proof of \ref{item 1 prop hilbert sh}) $\Rightarrow$ \ref{item 2 prop hilbert sh}) is accomplished by generalizing the proof of Theorem~\ref{thm: first generalization of Hilbert's inequality}\ref{enumerate item: first generalization of Hilbert's inequality sinus}) in \citep[Prop.~LS1.3]{Jameson2006}. Let $\nbNodes \in \N \!\setminus\! \{0\}$, $\boldsymbol{a} \triangleq \{a_k\}_{k = 1}^\nbNodes \in \C^\nbNodes$, and consider $\boldsymbol{r} \triangleq \{r_k\}_{k = 1}^\nbNodes \in \C^\nbNodes$ with $r_k \triangleq d_k + 2\pi i\xi_k$, where $d_k > 0$ and $\xi_k \in \R$ is such that 
	\begin{equation}
		\delta_k^{(w)} \triangleq \min_{\substack{1 \leq \ell \leq \nbNodes \\ \ell \neq k}}\min_{n \in \Z} \abs{\xi_k - \xi_\ell + n} > 0,
		\label{eq: second assumption delta k proof proposition} 
	\end{equation}
	for all $k \in \{1, 2, \ldots, \nbNodes\}$. For $k \in \{1, 2, \ldots, \nbNodes\}$ and $m \in \{1, 2, \ldots, N\}$, define $\lambda_{k, m} \triangleq d_k$, $u_{k, m} \triangleq \xi_k + m$, 
	$\rho_{k, m} \triangleq r_k + 2\pi im = d_k + 2\pi i(\xi_k + m)$, and $\alpha_{k, m} \triangleq (-1)^m a_k$, and apply \ref{item 1 prop hilbert sh}) with $M = \nbNodes\nbSamples$, $\boldsymbol{\rho} \triangleq \vectorize(\mathbf{P}) \in \C^M$, and $\boldsymbol{\alpha} \triangleq \vectorize(\mathbf{A}) \in \C^M$, where 
	\begin{equation*} 
		\mathbf{P} \triangleq \{\rho_{k, m}\}_{\substack{1 \leq k \leq \nbNodes \\ 1 \leq m \leq N}} \in \C^{\nbNodes \times N} \qquad \text{and }\qquad \mathbf{A} \triangleq \{\alpha_{k, m}\}_{\substack{1 \leq k \leq \nbNodes \\ 1 \leq m \leq N}} \in \C^{\nbNodes \times N}.
	\end{equation*}
	The conditions for application of \ref{item 1 prop hilbert sh}) are met, as for all $k \in \{1, 2, \ldots, K\}$ and all $m \in \{1, 2, \ldots, N\}$, 
	\begin{align*}
		\delta_{k, m} &\triangleq \min_{\substack{1 \leq \ell \leq K \\ 1 \leq n \leq N \\  (\ell, n) \neq (k, m)}} \abs{u_{k, m} - u_{\ell, n}} =\min_{\substack{1 \leq \ell \leq K \\ 1 \leq n \leq N \\  (\ell, n) \neq (k, m)}} \abs{(\xi_k + m) - (\xi_\ell + n)} \\
			&= \min_{\substack{1 \leq \ell \leq K \\ \ell \neq k}}\min_{n \in \Z} \abs{\xi_k - \xi_\ell + n} = \delta_k^{(w)} > 0,
	\end{align*}
	where the inequality follows by \eqref{eq: second assumption delta k proof proposition}. 
	We therefore get
	\begin{align}
		N\!\sum_{k = 1}^\nbNodes \!A(d_k, \delta_k^{(w)}\!, \delta^{(w)})\!\abs{a_k}^2 \leq \sum_{k, \ell = 1}^\nbNodes\sum_{m, n = 1}^N &\!\frac{(-1)^{m+n}a_k\overline{a_\ell}}{r_k + \overline{r_\ell} + 2\pi i(m-n)} \notag \\
			&\qquad\leq N\!\sum_{k = 1}^\nbNodes \!B(d_k, \delta_k^{(w)}\!, \delta^{(w)})\!\abs{a_k}^2,
		\label{eq: inequality to write in another form} 
	\end{align}
	where $\delta^{(w)} \triangleq \displaystyle\min_{1 \leq k \leq \nbNodes} \delta_k^{(w)}$.
	Next, we find an alternative expression for the center term in \eqref{eq: inequality to write in another form}. To this end, let $z \in \C \!\setminus\! 2\pi i\Z$, and note that
	\begin{align}
		\sum_{m, n = 1}^N \frac{(-1)^{m+n}}{z + 2\pi i(m-n)} &= \frac{N}{z} + \sum_{q = 1}^{N-1}\sum_{m = 1}^{N-q} \frac{(-1)^q}{z -2\pi iq} + \sum_{q = 1}^{N-1}\sum_{m = 1}^{N-q} \frac{(-1)^q}{z + 2\pi iq} \notag \\
			&= \frac{N}{z} + \sum_{q = 1}^{N-1} (N - q)\frac{(-1)^q}{z -2\pi iq} + \sum_{q = 1}^{N-1} (N - q)\frac{(-1)^q}{z + 2\pi iq} \notag \\
			&= \frac{N}{z} + \sum_{q = -N}^{-1} (N - \abs{q}) \frac{(-1)^q}{z + 2\pi iq} + \sum_{q = 1}^{N} (N - \abs{q})\frac{(-1)^q}{z +2\pi iq} \notag \\
			&= \sum_{q = -N}^{N} (N - \abs{q}) \frac{(-1)^q}{z + 2\pi iq}. \label{eq: correction sum diagonal term} 
	\end{align}
	Setting $z = r_k + \overline{r_\ell}$, $k, \ell \in \{1, 2, \ldots, \nbNodes\}$, in \eqref{eq: correction sum diagonal term} to recover the center term in \eqref{eq: inequality to write in another form} yields
	\begin{equation*}
		\sum_{k, \ell = 1}^\nbNodes\sum_{m, n = 1}^N \!\frac{(-1)^{m+n}a_k\overline{a_\ell}}{r_k + \overline{r_\ell} + 2\pi i(m-n)} = \sum_{k, \ell = 1}^\nbNodes\sum_{q = -N}^{N} (N - \abs{q}) \frac{(-1)^qa_k\overline{a_\ell}}{r_k + \overline{r_\ell} + 2\pi iq}.
	\end{equation*}
	We next establish that 
	\begin{equation}
		\lim_{N \rightarrow \infty} \sum_{q = -N}^N \left(1 - \frac{\abs{q}}{N}\right)\!\! \frac{(-1)^q}{\rho + 2\pi iq} = \frac{1}{2\sh(\rho/2)},
		\label{eq: equality 1 over sh} 
	\end{equation}
	for $\rho \in \C \!\setminus\! 2\pi i\Z$. 
	To this end, take $\rho \in \C \!\setminus\! 2\pi i \Z$ and let $\varphi \colon \R \rightarrow \C$ be the $1$-periodic function defined by $\varphi(t) \triangleq e^{-\rho t}$, for $t \in [-1/2, 1/2)$. The $q$-th Fourier series coefficient of $\varphi$ is 
	\begin{equation*}
		\int_{-1/2}^{1/2} \varphi(t) e^{-2\pi iqt} \mathrm{d}t = \int_{-1/2}^{1/2} e^{-(\rho + 2\pi iq)t} \mathrm{d}t = \frac{2\sh(\rho/2)(-1)^q}{\rho + 2\pi iq}.
	\end{equation*}
	As $\varphi$ is continuous on $\left(-1/2, 1/2\right)$, according to Fej\'er's theorem~\cite[Thm.~III.3.4]{Zygmund2002}, the sequence $\{\sigma_N\}_{N \in \N}$ of functions $\sigma_N \colon \R \rightarrow \C$ defined as
	\begin{equation*}
		\forall t \in \R, \qquad \sigma_N(t) \triangleq 2\sh(\rho/2)\sum_{q = -N}^N \left(1 - \frac{\abs{q}}{N}\right)\!\!\frac{(-1)^q}{\rho + 2\pi iq}e^{2\pi iqt},
	\end{equation*}
	converges pointwise to $\varphi$ on $\left(-1/2, 1/2\right)$, that is,
	\begin{equation}
		 e^{-\rho t} = 2\sh(\rho/2)\lim_{N \rightarrow \infty} \sum_{n = -N}^N \!\left(1 - \frac{\abs{q}}{N}\right) \!\!\frac{(-1)^q}{\rho + 2\pi iq}e^{2\pi iqt},
		\label{eq: fourier series proof} 
	\end{equation}
	for $t \in \left(-1/2, 1/2\right)$. 
	Evaluating \eqref{eq: fourier series proof} at $t = 0$, we obtain \eqref{eq: equality 1 over sh} as desired and hence 
	\begin{align}
		\lim_{N \rightarrow \infty} \frac{1}{N}\sum_{k, \ell = 1}^\nbNodes\sum_{m, n = 1}^N \!\frac{(-1)^{m+n}a_k\overline{a_\ell}}{r_k + \overline{r_\ell} + 2\pi i(m-n)} &= \!\sum_{k, \ell = 1}^\nbNodes\lim_{N \rightarrow \infty}\sum_{q = -N}^{N} \!\!\left(1 - \frac{\abs{q}}{N}\right) \!\!\frac{(-1)^qa_k\overline{a_\ell}}{r_k + \overline{r_\ell} + 2\pi iq} \notag \\
			&= \sum_{k, \ell = 1}^\nbNodes \frac{a_k\overline{a_\ell}}{2\sh((r_k + \overline{r_\ell})/2)}. \label{eq: conclusion proposition sh hilbert equivalence} 
	\end{align}
	Dividing \eqref{eq: inequality to write in another form} by $\nbSamples$ and letting $\nbSamples \rightarrow \infty$, it follows from \eqref{eq: conclusion proposition sh hilbert equivalence} that
	\begin{equation*}
		\sum_{k = 1}^\nbNodes 2A(d_k, \delta_k^{(w)}, \delta^{(w)}) \abs{a_k}^2 \leq \sum_{k, \ell = 1}^\nbNodes \frac{a_k\overline{a_\ell}}{\sh((r_k + \overline{r_\ell})/2)} \leq \sum_{k = 1}^\nbNodes 2B(d_k, \delta_k^{(w)}, \delta^{(w)})\abs{a_k}^2,
	\end{equation*}
	which completes the proof. 

%% file: appendix2.tex

\section{Proof of Theorem \ref{thm: result to be used for derivation of main result}}
\label{app: proof of theorem 2 appendix 2} 

	 According to Theorem~\ref{thm: second generalization of hilbert's inequality}, for all $\nbNodes \in \N \setminus\! \{0\}$, $\boldsymbol{\alpha} \triangleq \{\alpha_k\}_{k = 1}^\nbNodes \in \C^\nbNodes$, and $\boldsymbol{\rho} \triangleq \{\rho_k\}_{k = 1}^\nbNodes \in \C^\nbNodes$ with $\rho_k \triangleq \lambda_k + 2\pi i u_k$, where $\lambda_k > 0$ and $u_k \in \R$ is such that
	\begin{equation*}
		\delta_k \triangleq \min_{\substack{1 \leq \ell \leq \nbNodes \\ \ell \neq k}} \abs{u_k - u_\ell} > 0,
	\end{equation*}
	for all $k \in \{1, 2, \ldots, \nbNodes\}$, \eqref{eq: second generalization hilbert's inequality} holds so that adding $\sum_{k = 1}^\nbNodes \abs{a_k}^2\!/(2\lambda_k)$ to \eqref{eq: second generalization hilbert's inequality} yields
	\begin{equation*}
		\sum_{k = 1}^\nbNodes \!\left(\frac{1}{2\lambda_k} - \frac{42}{\pi \delta_k}\right)\!\abs{\alpha_k}^2 \leq \sum_{k, \ell = 1}^\nbNodes \frac{\alpha_k\overline{\alpha_\ell}}{\rho_k + \overline{\rho_\ell}} \leq \sum_{k = 1}^\nbNodes \!\left(\frac{1}{2\lambda_k} + \frac{42}{\pi\delta_k}\right)\!\abs{\alpha_k}^2.
	\end{equation*}
	Application of Proposition \ref{prop: equivalence hilbert's inequality sh}, with
	\begin{align*}
		A(x, y, z) &= 1/x - 42/(\pi y) \\
		B(x, y, z) &= 1/x + 42/(\pi y),
	\end{align*}
	for $x > 0$, $y > 0$, and $z > 0$, to \eqref{eq: second generalization hilbert's inequality} now yields \eqref{eq: second hilbert's inequality sh general}. 
	In the case $d_1 = d_2 = \ldots = d_\nbNodes = d$, \eqref{eq: second generalization hilbert's inequality} is refined to \eqref{eq: generalization hilbert's inequality case lambda constant}. 
	Application of Proposition \ref{prop: equivalence hilbert's inequality sh}, with 
	\begin{align*}
		A(x, y, z) &= 1/(z(e^{2x/z} - 1)) \\
		B(x, y, z) &= e^{2x/z}/(z(e^{2x/z} - 1)), 
	\end{align*}
	for $x > 0$, $y > 0$, and $z > 0$, to \eqref{eq: generalization hilbert's inequality case lambda constant} then yields \eqref{eq: second hilbert's inequality sh particular case}.

%% file: appendix3.tex

\section{Proof of Theorem~\ref{thm: main result general upper bound condition number vandermonde matrix unit disk}}
\label{app: proof of theorem 5 appendix 4}

The proof strategy is as follows. \!We first establish a lower bound on 
$\sigma_\mathrm{min}^2(\vandermonde{\nbSamples \times \nbNodes}{\node_1, \node_2, \ldots, \node_\nbNodes})$ and an upper bound on $\sigma_\mathrm{max}^2(\vandermonde{\nbSamples \times \nbNodes}{\node_1, \node_2, \ldots, \node_\nbNodes})$, both valid for nodes $\node_1, \node_2, \ldots, \node_\nbNodes$ strictly inside the unit circle, i.e., $\abs{\node_k} < 1$, $k \in \{1, 2, \ldots, \nbNodes\}$, and then use a limiting argument to extend these bounds, stated in Lemma \ref{lem: bound strictly inside the unit circle},  to the case where one or more of the nodes lie on the unit circle. Finally, we refine the resulting upper bound on $\sigma_\mathrm{max}^2(\vandermonde{\nbSamples \times \nbNodes}{\node_1, \node_2, \ldots, \node_\nbNodes})$ through Cohen's dilatation trick.

\begin{lem}[Lower bound on $\sigma_\mathrm{min}^2(\vandermonde{\nbSamples \times \nbNodes}{\node_1, \node_2, \ldots, \node_\nbNodes})$ and upper bound on $\sigma_\mathrm{max}^2(\vandermonde{\nbSamples \times \nbNodes}{\node_1, \node_2, \ldots, \node_\nbNodes})$ for nodes strictly inside the unit circle]
	\label{lem: bound strictly inside the unit circle} 
	Let $\boldsymbol{\node} \triangleq \{\node_k\}_{k = 1}^\nbNodes \in \C^\nbNodes$ with $\node_k \triangleq \abs{\node_k}\!e^{2\pi i\xi_k}$ be such that $0 < \abs{\node_k} < 1$, $\xi_k \in [0, 1)$, and
	\begin{equation}
		\delta_k^{(w)} \triangleq \min_{\substack{1 \leq \ell \leq \nbNodes \\ \ell \neq k}} \min_{n \in \Z} \abs{\xi_k - \xi_\ell + n} > 0,
		\label{eq: condition delta strictly positive bis} 
	\end{equation}
	for all $k \in \{1, 2, \ldots, \nbNodes\}$. The extremal singular values of the Vandermonde matrix $\vandermondeMat{\nbSamples \times \nbNodes}$ with nodes $\node_1, \node_2, \ldots, \node_\nbNodes$ satisfy
	\begin{align}
		\sigma_\mathrm{min}^2(\vandermonde{\nbSamples \times \nbNodes}{\node_1, \node_2, \ldots, \node_\nbNodes}) &\geq \min_{1 \leq k \leq \nbNodes} \!\left\{\!\frac{1}{2\!\abs{\node_k}}\!\left[\frac{1}{-\ln\!\abs{\node_k}}\!\left(1 - \abs{\node_k}^{2\nbSamples}\right) - \frac{84}{\pi\delta_k^{(w)}}\!\left(1 + \abs{\node_k}^{2\nbSamples}\right)\right]\!\right\} \label{eq: lower bound sigma min general bis} \\ 
		\sigma_\mathrm{max}^2(\vandermonde{\nbSamples \times \nbNodes}{\node_1, \node_2, \ldots, \node_\nbNodes}) &\leq \max_{1 \leq k \leq \nbNodes} \!\left\{\!\frac{1}{2\!\abs{\node_k}}\!\left[\frac{1}{-\ln\!\abs{\node_k}}\!\left(1 - \abs{\node_k}^{2\nbSamples}\right) + \frac{84}{\pi\delta_k^{(w)}}\!\left(1 + \abs{\node_k}^{2\nbSamples}\right)\right]\!\right\}\!. \label{eq: upper bound sigma max general bis} 
	\end{align}
	Moreover, if $\abs{\node_1} = \abs{\node_2} = \ldots = \abs{\node_\nbNodes} = A$, \eqref{eq: lower bound sigma min general bis} and \eqref{eq: upper bound sigma max general bis} can be refined to
	\begin{equation}
		\sigma_\mathrm{min}^2(\vandermonde{\nbSamples \times \nbNodes}{\node_1, \node_2, \ldots, \node_\nbNodes}) \geq \frac{1 - A^{2(\nbSamples + 1/2 - 1/\delta^{(w)})}}{\delta^{(w)}(A^{-2/\delta^{(w)}} - 1)A^2} \label{eq: lower bound sigma min particular case bis} 
	\end{equation}
	and
	\begin{equation}
		\sigma_\mathrm{max}^2(\vandermonde{\nbSamples \times \nbNodes}{\node_1, \node_2, \ldots, \node_\nbNodes}) \leq \frac{A^{-2/\delta^{(w)}}\!\left(1 - A^{2(\nbSamples - 1/2 +1/\delta^{(w)})}\right)}{\delta^{(w)}(A^{-2/\delta^{(w)}} - 1)A}, \label{eq: upper bound sigma max particular case bis} 
	\end{equation}
	respectively, with $\delta^{(w)} \triangleq \displaystyle\min_{1 \leq k \leq \nbNodes} \delta_k^{(w)}$. 
\end{lem}

\begin{proof}
	Let $\boldsymbol{x} \triangleq \left(x_1\ x_2\ \ldots \ x_\nbNodes\right)^T \in \C^\nbNodes$. We have
	\begin{align}
		\norm{\vandermonde{\nbSamples \times \nbNodes}{\node_1, \node_2, \ldots, \node_\nbNodes}\boldsymbol{x}}_2^2  &= \sum_{n = 0}^{\nbSamples-1} \abs{\sum_{k = 1}^\nbNodes x_k\node_k^n}^2 = \sum_{n = 0}^{\nbSamples-1} \sum_{k, \ell = 1}^\nbNodes \overline{x_k}x_\ell (\overline{\node_k}\node_\ell)^n \notag \\
		&= \sum_{k, \ell = 1}^\nbNodes \overline{x_k}x_\ell\frac{1 - (\overline{\node_k}\node_\ell)^\nbSamples}{1 - \overline{\node_k}\node_\ell} \label{eq: first equality proof main result} \\ 
		&= \sum_{k, \ell = 1}^\nbNodes \overline{x_k}x_\ell(\overline{\node}_k\node_\ell)^{-1/2} \frac{1 - (\overline{\node_k}\node_\ell)^\nbSamples}{(\overline{\node_k}\node_\ell)^{-1/2} - (\overline{\node_k}\node_\ell)^{1/2}} \notag \\
		&= \sum_{k, \ell = 1}^\nbNodes \overline{x_k}x_\ell (\overline{\node_k}\node_\ell)^{-1/2} \frac{1 - (\overline{\node_k}\node_\ell)^\nbSamples}{e^{(r_k+\overline{r_\ell})/2} - e^{-(r_k+\overline{r_\ell})/2}} \label{eq: rewritten with definition of rk} \\ 
		&= \sum_{k, \ell = 1}^\nbNodes \overline{x_k}x_\ell (\overline{\node_k}\node_\ell)^{-1/2} \frac{1 - (\overline{\node_k}\node_\ell)^\nbSamples}{2\sh((r_k + \overline{r_\ell})/2)} \notag \\
		&= \underbrace{\sum_{k, \ell = 1}^\nbNodes \frac{\overline{x_k}x_\ell (\overline{\node_k}\node_\ell)^{-1/2}}{2\sh((r_k + \overline{r_\ell})/2)}}_{\displaystyle\reversetriangleq X_1} -  \underbrace{\sum_{k, \ell = 1}^\nbNodes \frac{\overline{x_k}x_\ell (\overline{\node_k}\node_\ell)^{\nbSamples-1/2}}{2\sh((r_k + \overline{r_\ell})/2)}}_{\displaystyle \reversetriangleq X_2} \label{eq: last equation first block proof main result}, 
	\end{align}
	where in \eqref{eq: rewritten with definition of rk} we set $r_k \triangleq d_k + 2\pi i\xi_k$ with $d_k \triangleq -\ln\!\abs{\node_k}$, for $k \in \{1, 2, \ldots, \nbNodes\}$. To get \eqref{eq: first equality proof main result}, we used $\abs{\node_k} < 1$, for $k \in \{1, 2, \ldots, \nbNodes\}$, which is by assumption and ensures that $\overline{\node_k}\node_\ell \neq 1$, for all $k , \ell \in \{1, 2, \ldots, \nbNodes\}$. We proceed to derive lower and upper bounds on the terms $X_1$ and $X_2$ in \eqref{eq: last equation first block proof main result}. To this end, we first note that, by assumption, $0 < \abs{\node_k} < 1$, and hence $d_k > 0$, for all $k \in \{1, 2, \ldots, \nbNodes\}$. We can therefore apply \eqref{eq: second hilbert's inequality sh general} in Theorem \ref{thm: result to be used for derivation of main result} first with $a_k \triangleq \overline{x_k}(\overline{\node_k})^{-1/2}$, $k \in \{1, 2, \ldots, \nbNodes\}$, to get
	\begin{equation}
		\sum_{k = 1}^\nbNodes \!\left(\frac{1}{-\ln\!\abs{\node_k}} - \frac{84}{\pi\delta_k^{(w)}}\right)\!\frac{\abs{x_k}^2}{2\!\abs{\node_k}} \leq X_1 \leq \sum_{k = 1}^\nbNodes \!\left(\frac{1}{-\ln\!\abs{\node_k}} + \frac{84}{\pi\delta_k^{(w)}}\right)\!\frac{\abs{x_k}^2}{2\!\abs{\node_k}},
		\label{eq: first time used generalization hilbert sh} 
	\end{equation}
	and then with $a_k \triangleq \overline{x_k}(\overline{\node_k})^{\nbSamples-1/2}$, $k \in \{1, 2, \ldots, \nbNodes\}$, to conclude that
	\begin{equation}
		\!\!\sum_{k = 1}^\nbNodes \!\left(\frac{1}{-\ln\!\abs{\node_k}} - \frac{84}{\pi\delta_k^{(w)}}\right)\!\frac{\abs{x_k}^2\abs{\node_k}^{2\nbSamples}}{2\!\abs{\node_k}}\! \leq X_2 \leq \sum_{k = 1}^\nbNodes \!\left(\frac{1}{-\ln\!\abs{\node_k}} + \frac{84}{\pi\delta_k^{(w)}}\right)\!\frac{\abs{x_k}^2\abs{\node_k}^{2\nbSamples}}{2\!\abs{\node_k}}.\!
		\label{eq: second time used generalization hilbert sh} 
	\end{equation}
	With the left-hand side (LHS) of \eqref{eq: first time used generalization hilbert sh} and the RHS of \eqref{eq: second time used generalization hilbert sh}, we get
	\begin{align}
		\norm{\vandermonde{\nbSamples \times \nbNodes}{\node_1, \node_2, \ldots, \node_\nbNodes}\boldsymbol{x}}_2^2 &\geq \sum_{k = 1}^\nbNodes \left[\frac{1}{-\ln\!\abs{\node_k}}\!\left(1 - \abs{\node_k}^{2\nbSamples}\right) - \frac{84}{\pi\delta_k^{(w)}}\!\left(1 + \abs{\node_k}^{2\nbSamples}\right)\right]\!\frac{\abs{x_k}^2}{2\!\abs{\node_k}} \label{eq: first inequality to use vandermonde sigma min} \\ 
			&\geq \min_{1 \leq k \leq \nbNodes} \left\{\frac{1}{2\!\abs{\node_k}}\!\left[\frac{1}{-\ln\!\abs{\node_k}}\!\left(1 - \abs{\node_k}^{2\nbSamples}\right) - \frac{84}{\pi\delta_k^{(w)}}\!\left(1 + \abs{\node_k}^{2\nbSamples}\right)\right]\right\}\!\norm{\boldsymbol{x}}_2^2, \notag
	\end{align}
	which implies \eqref{eq: lower bound sigma min general bis}. 
	Combining the RHS of \eqref{eq: first time used generalization hilbert sh} and the LHS of \eqref{eq: second time used generalization hilbert sh}, we obtain
	\begin{align*}
		\norm{\vandermonde{\nbSamples \times \nbNodes}{\node_1, \node_2, \ldots, \node_\nbNodes}\boldsymbol{x}}_2^2 &\leq \sum_{k = 1}^\nbNodes \left[\frac{1}{-\ln\!\abs{\node_k}}\!\left(1 - \abs{\node_k}^{2\nbSamples}\right) + \frac{84}{\pi\delta_k^{(w)}}\!\left(1 + \abs{\node_k}^{2\nbSamples}\right)\right]\!\frac{\abs{x_k}^2}{2\!\abs{\node_k}} \\
			&\leq \max_{1 \leq k \leq \nbNodes} \left\{\frac{1}{2\!\abs{\node_k}}\!\left[\frac{1}{-\ln\!\abs{\node_k}}\!\left(1 - \abs{\node_k}^{2\nbSamples}\right) + \frac{84}{\pi\delta_k^{(w)}}\!\left(1 + \abs{\node_k}^{2\nbSamples}\right)\right]\right\}\!\norm{\boldsymbol{x}}_2^2,
	\end{align*}
	which proves \eqref{eq: upper bound sigma max general bis}.  
	For the refinements \eqref{eq: lower bound sigma min particular case bis} and \eqref{eq: upper bound sigma max particular case bis}, we derive specialized lower and upper bounds on the terms $X_1$ and $X_2$ in \eqref{eq: last equation first block proof main result}. To this end, we first note that in the case $\abs{\node_1} = \abs{\node_2} = \ldots = \abs{\node_\nbNodes} = A$, we have $d_1 = d_2 = \ldots = d_\nbNodes = -\ln(A)$ and hence $r_k = -\ln(A) + 2\pi i\xi_k$, for $k \in \{1, 2, \ldots, \nbNodes\}$, in $X_1$ and $X_2$. The proof of \eqref{eq: lower bound sigma min particular case bis} and \eqref{eq: upper bound sigma max particular case bis} is effected by employing \eqref{eq: second hilbert's inequality sh particular case 2} in Corollary \ref{cor: result to be used for derivation of main result} first with $a_k \triangleq \overline{x_k}(\overline{\node_k})^{-1/2}$, $k \in \{1, 2, \ldots, \nbNodes\}$, to get
	\begin{equation}
		\frac{1}{A\delta^{(w)}(A^{-2/\delta^{(w)}} - 1)}\sum_{k = 1}^\nbNodes \frac{\abs{x_k}^2}{A} \leq X_1 \leq \frac{A^{-2/\delta^{(w)}}}{\delta^{(w)}(A^{-2/\delta^{(w)}} - 1)}\sum_{k = 1}^\nbNodes \frac{\abs{x_k}^2}{A},
		\label{eq: first time used generalization hilbert sh bis} 
	\end{equation}
	and then with $a_k \triangleq \overline{x_k}(\overline{\node_k})^{\nbSamples-1/2}$, $k \in \{1, 2, \ldots, \nbNodes\}$, to conclude that
	\begin{equation}
		\frac{1}{A\delta^{(w)}(A^{-2/\delta^{(w)}} - 1)}\sum_{k = 1}^\nbNodes \frac{\abs{x_k}^2A^{2\nbSamples}}{A} \leq X_2 \leq \frac{A^{-2/\delta^{(w)}}}{\delta^{(w)}(A^{-2/\delta^{(w)}} - 1)}\sum_{k = 1}^\nbNodes \frac{\abs{x_k}^2A^{2\nbSamples}}{A}.
		\label{eq: second time used generalization hilbert sh bis} 
	\end{equation}
	With the LHS of \eqref{eq: first time used generalization hilbert sh bis} and the RHS of \eqref{eq: second time used generalization hilbert sh bis}, we have
	\begin{equation*}
		\norm{\vandermonde{\nbSamples \times \nbNodes}{\node_1, \node_2, \ldots, \node_\nbNodes}\boldsymbol{x}}_2^2 \geq \frac{1 - A^{2(\nbSamples + 1/2 - 1/\delta^{(w)})}}{\delta^{(w)}(A^{-2/\delta^{(w)}} - 1)A^2}\norm{\boldsymbol{x}}_2^2.
	\end{equation*}
	Finally, combining the RHS of \eqref{eq: first time used generalization hilbert sh bis} and the LHS of \eqref{eq: second time used generalization hilbert sh bis}, we obtain
	\begin{equation*}
		\norm{\vandermonde{\nbSamples \times \nbNodes}{\node_1, \node_2, \ldots, \node_\nbNodes}\boldsymbol{x}}_2^2 \leq \frac{A^{-2/\delta^{(w)}}\!\left(1 - A^{2(\nbSamples - 1/2 + 1/\delta^{(w)})}\right)}{\delta^{(w)}(A^{-2/\delta^{(w)}} - 1)A}\norm{\boldsymbol{x}}_2^2.
	\end{equation*}
\end{proof}

An immediate consequence of Lemma \ref{lem: bound strictly inside the unit circle} is the following bound on the condition number $\kappa(\vandermondeMat{\nbSamples \times \nbNodes})$.

\begin{lem}[Upper bound on $\kappa(\vandermonde{\nbSamples \times \nbNodes}{\node_1, \node_2, \ldots, \node_\nbNodes})$ for nodes strictly inside the unit circle]
	Let $\boldsymbol{\node} \triangleq \{\node_k\}_{k = 1}^\nbNodes \in \C^\nbNodes$ with $\node_k \triangleq \abs{\node_k}\!e^{2\pi i\xi_k}$ be such that $0 < \abs{\node_k} < 1$, $\xi_k \in [0, 1)$, and
	\begin{equation}
		\delta_k^{(w)} \triangleq \min_{\substack{1 \leq \ell \leq \nbNodes \\ \ell \neq k}} \min_{n \in \Z} \abs{\xi_k - \xi_\ell + n} > 0,
		\label{eq: condition delta strictly positive bis} 
	\end{equation}
	for all $k \in \{1, 2, \ldots, \nbNodes\}$. 
	The spectral condition number of $\vandermondeMat{\nbSamples \times \nbNodes}$ satisfies
	\begin{equation}
		\!\kappa(\vandermonde{\nbSamples \times \nbNodes}{\node_1, \node_2, \ldots, \node_\nbNodes}) \leq \!\left(\frac{\displaystyle\max_{1 \leq k \leq \nbNodes} \!\left\{\frac{1}{\abs{\node_k}}\!\left[\frac{1}{-\ln\!\abs{\node_k}}\!\left(1 - \abs{\node_k}^{2\nbSamples}\right) + \frac{84}{\pi\delta_k^{(w)}}\!\left(1 + \abs{\node_k}^{2\nbSamples}\right)\right]\right\}}{\displaystyle\min_{1 \leq k \leq \nbNodes} \!\left\{\frac{1}{\abs{\node_k}}\!\left[\frac{1}{-\ln\!\abs{\node_k}}\!\left(1 - \abs{\node_k}^{2\nbSamples}\right) - \frac{84}{\pi\delta_k^{(w)}}\!\left(1 + \abs{\node_k}^{2\nbSamples}\right)\right]\right\}}\right)^{1/2}
		\label{eq: main result general bis} 
	\end{equation}
	if for all $k \in \{1, 2, \ldots, \nbNodes\}$,
	\begin{equation}
		\delta_k^{(w)} > -\ln\!\abs{\node_k}\frac{84}{\pi}\!\left(\frac{1 + \abs{\node_k}^{2\nbSamples}}{1 - \abs{\node_k}^{2\nbSamples}}\right). 
		\label{eq: condition main result general bis} 
	\end{equation}
	Moreover, if $\abs{\node_1} = \abs{\node_2} = \ldots = \abs{\node_\nbNodes} = A$, we have
	\begin{equation}
		\kappa(\vandermonde{\nbSamples \times \nbNodes}{\node_1, \node_2, \ldots, \node_\nbNodes}) \leq A^{-1/\delta^{(w)}}\sqrt{\frac{A\!\left(1 - A^{2(\nbSamples - 1/2 + 1/\delta^{(w)})}\right)}{1 - A^{2(\nbSamples + 1/2 - 1/\delta^{(w)})}}}
		\label{eq: main result particular case bis} 
	\end{equation}
	under the condition $\nbSamples > 1/\delta^{(w)}-1/2$.
\end{lem}

\begin{proof}
	Using \eqref{eq: lower bound sigma min general bis} and \eqref{eq: upper bound sigma max general bis} in \eqref{eq: definition spectral condition number} yields \eqref{eq: main result general bis}. Condition \eqref{eq: condition main result general bis} ensures that the lower bound in \eqref{eq: lower bound sigma min general bis} is strictly positive, which enables division in \eqref{eq: main result general bis}. The refinement \eqref{eq: main result particular case bis} is obtained by using \eqref{eq: lower bound sigma min particular case bis} and \eqref{eq: upper bound sigma max particular case bis} in \eqref{eq: definition spectral condition number}. The condition $\nbSamples > 1/\delta^{(w)} - 1/2$ ensures that the lower bound in \eqref{eq: lower bound sigma min particular case bis} is strictly positive, which, again, enables division in \eqref{eq: main result particular case bis}.
\end{proof}

We are now ready to prove Theorem \ref{thm: main result general upper bound condition number vandermonde matrix unit disk} proper. This will be accomplished by first showing that  \eqref{eq: lower bound sigma min general bis}, \eqref{eq: upper bound sigma max general bis}, and \eqref{eq: main result general bis} can be extended (through a limiting argument) to the case where one or more of the nodes satisfy $\abs{\node_k} = 1$, and second by refining the resulting upper bound on $\sigma_\mathrm{max}^2(\vandermonde{\nbSamples \times \nbNodes}{\node_1, \node_2, \ldots, \node_\nbNodes})$ and hence the upper bound on $\kappa(\vandermonde{\nbSamples \times \nbNodes}{\node_1, \node_2, \ldots, \node_\nbNodes})$ via Cohen's dilatation trick. 

The basic idea of the proof is to construct a sequence of Vandermonde matrices parametrized by $M \in \N \!\setminus \! \{0\}$, with nodes strictly inside the unit circle and approaching the unit circle as $M \rightarrow \infty$. 
	
	Specifically, let $M \in \N \setminus \! \{0\}$ and $\vandermonde{\nbSamples \times \nbNodes}{\node_1, \node_2, \ldots, \node_\nbNodes}^{(M)}$ be the Vandermonde matrix with nodes $\node_1^{(M)}, \node_2^{(M)}, \ldots, \node_\nbNodes^{(M)}$ such that
	\begin{equation*}
		\node_k^{(M)} = \begin{cases} \abs{\node_k}\!e^{2\pi i\xi_k}, & \abs{\node_k} < 1\\
			\left(1 - \frac{1}{M}\right)\!e^{2\pi i\xi_k}, & \abs{\node_k} = 1.\end{cases}
	\end{equation*}
	The nodes $\node_1^{(M)}, \node_2^{(M)}, \ldots, \node_\nbNodes^{(M)}$ are all strictly inside the unit circle, that is, $\abs{\node_k^{(M)}} < 1$, for all $k \in \{1, 2, \ldots, \nbNodes\}$, and we can therefore apply results from the proof of Lemma~\ref{lem: bound strictly inside the unit circle} to obtain bounds on the extremal singular values of $\vandermonde{\nbSamples \times \nbNodes}{\node_1, \node_2, \ldots, \node_\nbNodes}^{(M)}$. 
	Specifically, let $\boldsymbol{x} \triangleq \{x_k\}_{k = 1}^\nbNodes \in \C^\nbNodes$ and evaluate the lower bound in \eqref{eq: first inequality to use vandermonde sigma min} for the nodes $\node_1^{(M)}, \node_2^{(M)}, \ldots, \node_\nbNodes^{(M)}$ to get
	\begin{align*}
		&\norm{\vandermonde{\nbSamples \times \nbNodes}{\node_1, \node_2, \ldots, \node_\nbNodes}^{(M)}\boldsymbol{x}}_2^2 \geq \sum_{k = 1}^\nbNodes \left[\frac{1}{-\ln\!\abs{\node_k^{(M)}}}\!\left(1 - \abs{\node_k^{(M)}}^{2\nbSamples}\right) - \frac{84}{\pi\delta_k^{(w)}}\!\left(1 + \abs{\node_k^{(M)}}^{2\nbSamples}\right)\right]\!\frac{\abs{x_k}^2}{2\!\abs{\node_k^{(M)}}} \\
		&= \sum_{\substack{k = 1 \\ \abs{\node_k} = 1}}^\nbNodes \!\left[\frac{1}{-\ln\!\left(1 - \frac{1}{M}\right)}\!\left(1 - \left(\!1 - \frac{1}{M}\right)^{2\nbSamples}\right) - \frac{84}{\pi\delta_k^{(w)}}\!\left(\!1 + \left(1 - \frac{1}{M}\right)^{2\nbSamples}\right)\right]\!\frac{\abs{x_k}^2}{2\!\left(1 - \frac{1}{M}\right)} \\
			&\hspace{2cm} + \sum_{\substack{k = 1 \\ \abs{\node_k} < 1}}^\nbNodes \left[\frac{1}{-\ln\!\abs{\node_k}}\!\left(1 - \abs{\node_k}^{2\nbSamples}\right) - \frac{84}{\pi\delta_k^{(w)}}\!\left(1 + \abs{\node_k}^{2\nbSamples}\right)\right]\frac{\abs{x_k}^2}{2\!\abs{\node_k}}. \\[-0.7cm]
	\end{align*}
	As $\lim_{M \rightarrow \infty} \node_k^{(M)} = \node_k$, for $k \in \{1, 2, \ldots, \nbNodes\}$, it follows that
	\begin{align*}
		\norm{\vandermonde{\nbSamples \times \nbNodes}{\node_1, \node_2, \ldots, \node_\nbNodes}\boldsymbol{x}}_2^2 &= \lim_{M \rightarrow \infty} \norm{\vandermonde{\nbSamples \times \nbNodes}{\node_1, \node_2, \ldots, \node_\nbNodes}^{(M)}\boldsymbol{x}}_2^2 \\
			&\geq \sum_{\substack{k = 1 \\ \abs{\node_k} < 1}}^\nbNodes \!\!\left[\frac{1}{-\ln\!\abs{\node_k}}\!\left(1 - \abs{\node_k}^{2\nbSamples}\right) - \frac{84}{\pi\delta_k^{(w)}}\!\left(1 + \abs{\node_k}^{2\nbSamples}\right)\right]\!\frac{\abs{x_k}^2}{2\!\abs{\node_k}} \\[-0.4cm]
			&\hspace{3cm} + \sum_{\substack{k = 1 \\ \abs{\node_k} = 1}}^\nbNodes \!\!\left(\nbSamples - \frac{42}{\pi\delta_k}\right)\!\abs{x_k}^2\\
			&= \sum_{k = 1}^\nbNodes \!\left[\varphi_\nbSamples(\abs{\node_k}) - \frac{42}{\pi\delta_k^{(w)}}\!\left(1 + \abs{\node_k}^{2\nbSamples}\right)\right]\!\frac{\abs{x_k}^2}{\abs{\node_k}} \\
			&\geq \mathscr{L}(\nbSamples, \abs{\boldsymbol{z}}\!, \boldsymbol{\delta}^{(w)})\!\norm{\boldsymbol{x}}_2^2,
	\end{align*}
	for all $\boldsymbol{x} = \{x_k\}_{k = 1}^\nbNodes \in \C^\nbNodes$. This implies
	\begin{equation}
		\sigma_\mathrm{min}^2(\vandermonde{\nbSamples \times \nbNodes}{\node_1, \node_2, \ldots, \node_\nbNodes}) \geq \mathscr{L}(\nbSamples, \abs{\boldsymbol{z}}\!, \boldsymbol{\delta}^{(w)}),
		\label{eq: lower bound final result proof} 
	\end{equation}
	and thereby establishes \eqref{eq: lower bound sigma min general}. We can show similarly that
	\begin{equation}
		\sigma_\mathrm{max}^2(\vandermonde{\nbSamples \times \nbNodes}{\node_1, \node_2, \ldots, \node_\nbNodes}) \leq \mathscr{U}(\nbSamples, \abs{\boldsymbol{z}}\!, \boldsymbol{\delta}^{(w)}).
		\label{eq: upper bound sigma max before cohen's dilatation trick} 
	\end{equation}
	To get \eqref{eq: upper bound sigma max general}, we refine \eqref{eq: upper bound sigma max before cohen's dilatation trick} using Cohen's dilatation trick as follows.
	Let $\boldsymbol{y} \triangleq \{y_n\}_{n = 0}^{\nbSamples-1}$, and set
	\begin{equation}
		\forall (a, \omega) \in [0, \infty) \times \R, \qquad U_{\boldsymbol{y}, \nbSamples}(a, \xi) \triangleq \sum_{n = 0}^{\nbSamples-1} y_n a^ne^{-2\pi i\xi n}.
		\label{eq: function U cohen's dilatation trick main result cor} 
	\end{equation}
	It follows from \eqref {eq: upper bound sigma max before cohen's dilatation trick} that
	\begin{align}
		 \sum_{k = 1}^\nbNodes \abs{U_{\boldsymbol{y}, \nbSamples}(\abs{z_k}\!, \xi_k)}^2 = \norm{(\vandermondeMat{\nbSamples \times \nbNodes})^H\boldsymbol{y}}_2^2 \leq \mathscr{U}(\nbSamples, \abs{\boldsymbol{z}}, \boldsymbol{\delta}^{(w)})\!\norm{\boldsymbol{y}}_2^2.
		\label{eq: equation proof condition number of the Vandermonde matrix with nodes in the unit disk 8} 
	\end{align}
	Now, we take $R \in \N \!\setminus\! \{0\}$ and apply Cohen's dilatation trick with respect to the variable $\xi$ in \eqref{eq: function U cohen's dilatation trick main result cor}. We start by setting
	\begin{equation*}
		\forall (a, \omega) \in [0, \infty) \times \R, \qquad V_{\boldsymbol{y}, \nbSamples, R}(a, \xi) \triangleq U_{\boldsymbol{y}, \nbSamples}(a, R\xi) = \sum_{n = 0}^{\nbSamples-1} y_n a^ne^{-2\pi iR\xi n}.
	\end{equation*}
	With $\boldsymbol{\gamma} \triangleq \{\gamma_n\}_{n = 0}^{(\nbSamples-1)R}$ defined as
	\begin{equation*}
		\gamma_n \triangleq \begin{cases} y_{n/R}, & \text{if } n \equiv 0\ (\text{mod } R) \\ 0, & \text{otherwise,}\end{cases}
	\end{equation*}
	$V_{\boldsymbol{y}, \nbSamples, R}$ can be written as
	\begin{align}
		\hspace{-0.1cm}\forall (a, \omega) \in [0, \infty) \times \R, \quad\!\! V_{\boldsymbol{y}, \nbSamples, R}(a, \xi) &= \!\!\sum_{n = 0}^{(\nbSamples-1)R} \gamma_n a^{n/R}e^{-2\pi i\xi n} \notag \\
			&= U_{\boldsymbol{\gamma}, (\nbSamples-1)R+1}(a^{1/R}, \xi). 
		\label{eq: equation proof condition number of the Vandermonde matrix with nodes in the unit disk 6} 
	\end{align}
	We then have
	\begin{align}
		&R\sum_{k = 1}^\nbNodes \abs{U_{\boldsymbol{y}, \nbSamples}(\abs{\node_k}\!, \xi_k)}^2 = \sum_{r = 1}^R\sum_{k = 1}^\nbNodes \abs{U_{\boldsymbol{y}, \nbSamples}(\abs{\node_k}\!, \xi_k + r)}^2  \label{eq: equation proof condition number of the Vandermonde matrix with nodes in the unit disk 6.1} \\ 
			&\hspace{1.5cm}= \sum_{r = 1}^R\sum_{k = 1}^\nbNodes \abs{V_{\boldsymbol{y}, \nbSamples, R}\!\left(\abs{\node_k}\!, \frac{\xi_k+ r}{R}\right)}^2 \label{eq: equation proof condition number of the Vandermonde matrix with nodes in the unit disk 6.2} \\ 
			&\hspace{1.5cm}= \sum_{r = 1}^R\sum_{k = 1}^\nbNodes \abs{U_{\boldsymbol{\gamma}, (\nbSamples-1)R+1}\!\left(\abs{\node_k}^{1/R}\!,\frac{\xi_k+ r}{R}\right)}^2, \label{eq: equation proof condition number of the Vandermonde matrix with nodes in the unit disk 7} 
	\end{align}
	where \eqref{eq: equation proof condition number of the Vandermonde matrix with nodes in the unit disk 6.1} holds as $\xi \mapsto U_{\boldsymbol{y}, \nbSamples}(a, \xi)$ is $1$-periodic, \eqref{eq: equation proof condition number of the Vandermonde matrix with nodes in the unit disk 6.2} is by definition of $V_{\boldsymbol{y}, \nbSamples, R}$, and \eqref{eq: equation proof condition number of the Vandermonde matrix with nodes in the unit disk 7} follows from \eqref{eq: equation proof condition number of the Vandermonde matrix with nodes in the unit disk 6}. We have
	\begin{equation*}
		\min_{\substack{1 \leq \ell \leq \nbNodes \\ 1\leq s \leq R \\ (\ell, s) \neq (k, r)}} \min_{n \in \Z} \abs{\frac{\xi_k+r}{R} - \frac{\xi_\ell+s}{R}+ \frac{n}{R}} = \frac{\delta_k^{(w)}}{R} > 0
	\end{equation*}
	and can therefore apply \eqref{eq: equation proof condition number of the Vandermonde matrix with nodes in the unit disk 8} to $U_{\boldsymbol{\gamma}, (\nbSamples-1)R+1}$ with the substitutions
	\begin{equation*}
		\begin{array}{llcl}
			& \nbSamples & \longleftarrow & (\nbSamples-1)R+1 \\[0.1cm]
			& \boldsymbol{y} = \{y_n\}_{n = 0}^{\nbSamples-1} & \longleftarrow & \boldsymbol{\gamma} = \{\gamma_n\}_{n = 0}^{(\nbSamples-1)R} \\[0.1cm]
			& \{(\abs{z_k}, \xi_k)\}_{k = 1}^\nbNodes & \longleftarrow & \displaystyle\left\{\left(\abs{\node_k}^{1/R}, \frac{\xi_k+r}{R}\right)\right\}_{\substack{1 \leq k \leq \nbNodes \\ 1 \leq r \leq R}} \\[0.3cm]
			& \abs{\node_k} & \longleftarrow & \abs{\node_k}^{1/R} \\[0.2cm]
			& \delta_k^{(w)} & \longleftarrow & \delta_k^{(w)}/R,
		\end{array}
	\end{equation*}
	i.e., we apply \eqref{eq: equation proof condition number of the Vandermonde matrix with nodes in the unit disk 8} to the $2$-D sequence $\displaystyle\left\{\left(\abs{\node_k}^{1/R}, \frac{\xi_k+r}{R}\right)\right\}_{\substack{1 \leq k \leq \nbNodes \\ 1 \leq r \leq R}}$ with the corresponding replacements for $\nbSamples$, $\boldsymbol{y}$, and $\delta^{(w)}$.
	This yields
	\begin{equation}
		\sum_{r = 1}^R\sum_{k = 1}^\nbNodes \abs{U_{\boldsymbol{\gamma}, (\nbSamples-1)R+1}\!\left(\abs{\node_k}^{1/R},\frac{\xi_k+r}{R}\right)}^2 \leq \mathscr{U}\!\left(\!(\nbSamples-1)R+1, \abs{\boldsymbol{z}}^{1/R}\!, \frac{\boldsymbol{\delta}^{(w)}}{R}\right)\!\norm{\boldsymbol{\gamma}}_2^2, 
		\label{eq: reference to give for explaining equation} 
	\end{equation}
	where $\abs{\boldsymbol{z}}^{1/R} \triangleq \left\{\abs{z_k}^{1/R}\right\}_{k = 1}^\nbNodes \in \C^\nbNodes$. 
	Since $\norm{\boldsymbol{y}}_2 = \norm{\boldsymbol{\gamma}}_2$, it follows from \eqref{eq: equation proof condition number of the Vandermonde matrix with nodes in the unit disk 7} that
	\begin{align}
		R\sum_{k = 1}^\nbNodes \abs{U_{\boldsymbol{y}, \nbSamples}(\abs{\node_k}, \xi_k)}^2 \leq \mathscr{U}\!\left(\!(\nbSamples-1)R+1, \abs{\boldsymbol{z}}^{1/R}\!, \frac{\boldsymbol{\delta}^{(w)}}{R}\right)\!\norm{\boldsymbol{y}}_2^2.
		\label{eq: inequality to divide final steps} 
	\end{align}
	Thanks to \eqref{eq: definition upper bound U}, we have
	\begin{align*}
		&\frac{1}{R}\mathscr{U}\!\!\left(\!(\nbSamples-1)R+1, \abs{\boldsymbol{z}}^{1/R}\!, \frac{\boldsymbol{\delta}^{(w)}}{R}\right) \\[0.3cm]
			&\hspace{1cm}= \max_{1 \leq k \leq \nbNodes} \left\{\frac{1}{\abs{\node_k}^{1/R}}\!\left[\frac{\varphi_{(\nbSamples-1)R+1}\!\left({\abs{\node_k}}^{1/R}\right)}{R} + \frac{42}{\pi\delta_k^{(w)}}\!\left(1 + \abs{\node_k}^{2(\nbSamples-1) + 2/R}\right)\right]\right\}.
	\end{align*}
	Since
	\begin{equation*}
		\frac{\varphi_{(\nbSamples-1)R+1}\!\left({\abs{\node_k}}^{1/R}\right)}{R} = \begin{cases} \displaystyle \frac{\abs{\node_k}^{2(\nbSamples-1) + 2/R} - 1}{2\ln\!\abs{\node_k}}, & \abs{\node_k} < 1 \\ \displaystyle (\nbSamples-1) + 1/R, & \abs{\node_k} = 1, \end{cases}
	\end{equation*}
	we get
	\begin{equation*}
		\lim_{R \rightarrow \infty} \frac{\varphi_{(\nbSamples-1)R+1}\!\left({\abs{\node_k}}^{1/R}\right)}{R} = \varphi_{\nbSamples-1}\!\left(\abs{\node_k}\right)\!,
	\end{equation*}
	and hence
	\begin{equation*}
		\lim_{R \rightarrow \infty} \frac{1}{R}\mathscr{U}\!\!\left(\!(\nbSamples-1)R+1, \abs{\boldsymbol{z}}^{1/R}\!, \frac{\boldsymbol{\delta}^{(w)}}{R}\right) = \mathscr{U}(\nbSamples-1, \abs{\boldsymbol{z}}, \boldsymbol{\delta}^{(w)}).
	\end{equation*}
	Dividing \eqref{eq: inequality to divide final steps} by $R > 0$ and letting $R \rightarrow \infty$ yields
	\begin{equation}
		\norm{(\vandermondeMat{\nbSamples \times \nbNodes})^H\boldsymbol{y}}_2^2 = {\sum_{k = 1}^\nbNodes \abs{U_{\boldsymbol{y}\!, \nbSamples}(\abs{\node_k}, \xi_k)}^2} \leq \mathscr{U}(\nbSamples-1, \abs{\boldsymbol{z}}, \boldsymbol{\delta}^{(w)})\!\norm{\boldsymbol{y}}_2^2.
		\label{eq: equation to refer to} 
	\end{equation}
	Since \eqref{eq: equation to refer to} holds for all $\boldsymbol{y} \in \C^\nbSamples$, this implies
	\begin{equation}
		\sigma_\mathrm{max}^2(\vandermondeMat{\nbSamples \times \nbNodes}) = \sigma_\mathrm{max}^2((\vandermondeMat{\nbSamples \times \nbNodes})^H) \leq \mathscr{U}(\nbSamples-1, \abs{\boldsymbol{z}}, \boldsymbol{\delta}^{(w)}).
		\label{eq: equation proof main result last 2} 
	\end{equation}
	Neither of the upper bounds in \eqref{eq: upper bound sigma max before cohen's dilatation trick} and \eqref{eq: equation proof main result last 2} is consistently smaller than the other one so that in summary
	\begin{equation}
		\sigma_\mathrm{max}^2(\vandermonde{\nbSamples \times \nbNodes}{\node_1, \node_2, \ldots, \node_\nbNodes}) \leq \min\!\Big\{\mathscr{U}(\nbSamples, \abs{\boldsymbol{z}}, \boldsymbol{\delta}^{(w)}), \mathscr{U}(\nbSamples-1, \abs{\boldsymbol{z}}, \boldsymbol{\delta}^{(w)})\Big\}.
		\label{eq: upper bound final result proof} 
	\end{equation}
	This concludes the proof of \eqref{eq: upper bound sigma max general}.
	
	It remains to establish \eqref{eq: upper bound sigma max particular case}. To this end, we first note that by \eqref{eq: reference to give for explaining equation}
	\begin{align*}
		\sum_{r = 1}^R\sum_{k = 1}^\nbNodes& \abs{U_{\boldsymbol{\gamma}, (\nbSamples-1)R+1}\!\left(\abs{\node_k}^{1/R}\!,\frac{\xi_k+r}{R}\right)}^2 \\
			&\hspace{1.5cm}\leq \frac{RA^{-2/\delta^{(w)}}(1 - A^{2(\nbSamples-1+3/(2R) + 1/\delta^{(w)})})}{\delta^{(w)}(A^{-2/\delta^{(w)}} - 1)A^{1/R}}\!\norm{\boldsymbol{\gamma}}_2^2.
	\end{align*}
	Since $\norm{\boldsymbol{y}}_2 = \norm{\boldsymbol{\gamma}}_2$, it follows from the equality between \eqref{eq: equation proof condition number of the Vandermonde matrix with nodes in the unit disk 6.1} and \eqref{eq: equation proof condition number of the Vandermonde matrix with nodes in the unit disk 7} that
	\begin{align}
		\sum_{k = 1}^\nbNodes \abs{U_{\boldsymbol{y}, \nbSamples}(\abs{\node_k}\!, \xi_k)}^2 \leq \frac{A^{-2/\delta^{(w)}}(1 - A^{2(\nbSamples-1+3/(2R) + 1/\delta^{(w)})})}{\delta^{(w)}(A^{-2/\delta^{(w)}} - 1)A^{1/R}}\!\norm{\boldsymbol{y}}_2^2.
		\label{eq: inequality to divide final steps 2} 
	\end{align}
	Letting $R \rightarrow \infty$ in \eqref{eq: inequality to divide final steps 2} yields
	\begin{equation*}
		\norm{(\vandermondeMat{\nbSamples \times \nbNodes})^H\boldsymbol{y}}_2^2 = {\sum_{k = 1}^\nbNodes \abs{U_{\boldsymbol{y}, \nbSamples}(\abs{\node_k}\!, \xi_k)}^2} \leq \frac{A^{-2/\delta^{(w)}}(1 - A^{2(\nbSamples-1+ 1/\delta^{(w)})})}{\delta^{(w)}(A^{-2/\delta^{(w)}} - 1)}\!\norm{\boldsymbol{y}}_2^2.
	\end{equation*}
	Since this holds for all $\boldsymbol{y} \in \C^\nbSamples$, we get
	\begin{equation}
		\sigma_\mathrm{max}^2(\vandermondeMat{\nbSamples \times \nbNodes}) = \sigma_\mathrm{max}^2((\vandermondeMat{\nbSamples \times \nbNodes})^H) \leq \frac{A^{-2/\delta^{(w)}}(1 - A^{2(\nbSamples-1+ 1/\delta^{(w)})})}{\delta^{(w)}(A^{-2/\delta^{(w)}} - 1)}, 
		\label{eq: very last equation  proof cohen unit disk} 
	\end{equation}
	which is \eqref{eq: upper bound sigma max particular case}. This completes the proof.

%% file: appendix4.tex

\section{Proof of Corollary~\ref{cor: main result general upper bound condition number vandermonde matrix unit disk}}
\label{app: proof of corollary main result appendix 3} 

Using \eqref{eq: lower bound final result proof} and \eqref{eq: upper bound final result proof} in \eqref{eq: definition spectral condition number} yields \eqref{eq: main result general}. Condition \eqref{eq: condition main result general} ensures that $\mathscr{L}(\nbSamples, \abs{\boldsymbol{z}}\!, \boldsymbol{\delta}^{(w)}) > 0$, which enables division in \eqref{eq: main result general}.
	The refinement \eqref{eq: main result particular case} is obtained by employing \eqref{eq: lower bound sigma min particular case} and \eqref{eq: very last equation  proof cohen unit disk} in \eqref{eq: definition spectral condition number}, and the condition $\nbSamples > 1/\delta^{(w)}-1/2$ ensures that the lower bound in \eqref{eq: lower bound sigma min particular case} is positive, which, again, enables division in \eqref{eq: main result particular case}.

%% file: ConditioningVandermonde_arxiv.bbl
\begin{thebibliography}{56}
\providecommand{\natexlab}[1]{#1}
\providecommand{\url}[1]{\texttt{#1}}
\providecommand{\urlprefix}{URL }
\expandafter\ifx\csname urlstyle\endcsname\relax
  \providecommand{\doi}[1]{doi:\discretionary{}{}{}#1}\else
  \providecommand{\doi}[1]{doi:\discretionary{}{}{}\begingroup
  \urlstyle{rm}\url{#1}\endgroup}\fi
\providecommand{\bibinfo}[2]{#2}

\bibitem[{Selberg(1991)}]{Selberg1991}
\bibinfo{author}{A.~Selberg}, \bibinfo{title}{Collected Papers---{V}olume
  {II}}, \bibinfo{publisher}{Springer Verlag}, \bibinfo{address}{Heidelberg},
  \bibinfo{year}{1991}.

\bibitem[{Moitra(2015)}]{Moitra2015}
\bibinfo{author}{A.~Moitra}, \bibinfo{title}{Super-resolution, extremal
  functions and the condition number of {V}andermonde matrices}, in:
  \bibinfo{booktitle}{Proceedings of the 47th Annual ACM Symposium on Theory of
  Computing (STOC)}, \bibinfo{address}{Portland, OR, USA},
  \bibinfo{year}{2015}.

\bibitem[{Baz\'an(2000)}]{Bazan2000}
\bibinfo{author}{F.~S.~V. Baz\'an}, \bibinfo{title}{Conditioning of rectangular
  {V}andermonde matrices with nodes in the unit disk}, \bibinfo{journal}{SIAM
  Journal on Matrix Analysis and Applications}
  \bibinfo{volume}{21}~(\bibinfo{number}{2}) (\bibinfo{year}{2000})
  \bibinfo{pages}{679--693}.

\bibitem[{Bj\"orck and Elfving(1973)}]{Bjorck1973}
\bibinfo{author}{A.~Bj\"orck}, \bibinfo{author}{T.~Elfving},
  \bibinfo{title}{Algorithms for confluent {V}andermonde systems},
  \bibinfo{journal}{Numerische Mathematik}
  \bibinfo{volume}{21}~(\bibinfo{number}{2}) (\bibinfo{year}{1973})
  \bibinfo{pages}{130--137}.

\bibitem[{Heinig and Rost(1995)}]{Heinig1995}
\bibinfo{author}{G.~Heinig}, \bibinfo{author}{K.~Rost},
  \bibinfo{title}{Recursive solution of {C}auchy-{V}andermonde systems of
  equations}, \bibinfo{journal}{Linear Algebra and its Applications}
  \bibinfo{volume}{218} (\bibinfo{year}{1995}) \bibinfo{pages}{59--72}.

\bibitem[{Luther and Rost(2004)}]{Luther2004}
\bibinfo{author}{U.~Luther}, \bibinfo{author}{K.~Rost}, \bibinfo{title}{Matrix
  exponentials and inversion of confluent {V}andermonde matrices},
  \bibinfo{journal}{Electronic Transactions on Numerical Analysis}
  \bibinfo{volume}{18} (\bibinfo{year}{2004}) \bibinfo{pages}{91--100}.

\bibitem[{Pantelous and Karageorgos(2013)}]{Pantelous2013}
\bibinfo{author}{A.~A. Pantelous}, \bibinfo{author}{A.~D. Karageorgos},
  \bibinfo{title}{Generalized inverses of the {V}andermonde matrix:
  {A}pplications in Control Theory}, \bibinfo{journal}{International Journal of
  Control, Automation, and Systems} \bibinfo{volume}{11}~(\bibinfo{number}{5})
  (\bibinfo{year}{2013}) \bibinfo{pages}{1063--1070}.

\bibitem[{Gr\"ochenig(1999)}]{Groechenig1999}
\bibinfo{author}{K.~Gr\"ochenig}, \bibinfo{title}{Irregular sampling,
  {T}oeplitz matrices, and the approximation of entire functions of exponential
  type}, \bibinfo{journal}{Mathematics of Computation}
  \bibinfo{volume}{68}~(\bibinfo{number}{226}) (\bibinfo{year}{1999})
  \bibinfo{pages}{749--765}.

\bibitem[{Vetterli et~al.(2002)Vetterli, Marziliano, and Blu}]{Vetterli2002}
\bibinfo{author}{M.~Vetterli}, \bibinfo{author}{P.~Marziliano},
  \bibinfo{author}{T.~Blu}, \bibinfo{title}{Sampling signals with finite rate
  of innovation}, \bibinfo{journal}{IEEE Transactions on Signal Processing}
  \bibinfo{volume}{50}~(\bibinfo{number}{6}) (\bibinfo{year}{2002})
  \bibinfo{pages}{1417--1428}.

\bibitem[{Feng and Bresler(1996)}]{Feng1996}
\bibinfo{author}{P.~Feng}, \bibinfo{author}{Y.~Bresler},
  \bibinfo{title}{Spectrum-blind minimum-rate sampling and reconstruction of
  multiband signals}, in: \bibinfo{booktitle}{Proceedings of IEEE International
  Conference on Acoustics, Speech, and Signal Processing (ICASSP)},
  vol.~\bibinfo{volume}{3}, \bibinfo{address}{Atlanta, GA, USA},
  \bibinfo{pages}{1688--1691}, \bibinfo{year}{1996}.

\bibitem[{Mishali and Eldar(2009)}]{Mishali2009}
\bibinfo{author}{M.~Mishali}, \bibinfo{author}{Y.~C. Eldar},
  \bibinfo{title}{Blind multiband signal reconstruction: {C}ompressed sensing
  for analog signals}, \bibinfo{journal}{IEEE Transactions on Signal
  Processing} \bibinfo{volume}{57}~(\bibinfo{number}{3}) (\bibinfo{year}{2009})
  \bibinfo{pages}{993--1009}.

\bibitem[{Schmidt(1986)}]{Schmidt1986}
\bibinfo{author}{R.~O. Schmidt}, \bibinfo{title}{Multiple emitter location and
  signal parameter estimation}, \bibinfo{journal}{IEEE Transactions on Antennas
  and Propagation} \bibinfo{volume}{34}~(\bibinfo{number}{3})
  (\bibinfo{year}{1986}) \bibinfo{pages}{276--280}.

\bibitem[{Roy et~al.(1986)Roy, Paulraj, and Kailath}]{Roy1986}
\bibinfo{author}{R.~Roy}, \bibinfo{author}{A.~Paulraj},
  \bibinfo{author}{T.~Kailath}, \bibinfo{title}{{ESPRIT} -- {A} subspace
  rotation approach to estimation of parameters of cisoids in noise},
  \bibinfo{journal}{IEEE Transactions on Acoustics, Speech, and Signal
  Processing} \bibinfo{volume}{34}~(\bibinfo{number}{5}) (\bibinfo{year}{1986})
  \bibinfo{pages}{1340--1342}.

\bibitem[{Hua and Sarkar(1990)}]{Hua1990}
\bibinfo{author}{Y.~Hua}, \bibinfo{author}{T.~K. Sarkar},
  \bibinfo{title}{Matrix pencil method for estimating parameters of
  exponentially damped/undamped sinusoids in noise}, \bibinfo{journal}{IEEE
  Transactions on Acoustics, Speech, and Signal Processing}
  \bibinfo{volume}{38}~(\bibinfo{number}{5}) (\bibinfo{year}{1990})
  \bibinfo{pages}{814--824}.

\bibitem[{Liao and Fannjiang(2016)}]{Liao2014}
\bibinfo{author}{W.~Liao}, \bibinfo{author}{A.~Fannjiang},
  \bibinfo{title}{{MUSIC} for single-snapshot spectral estimation: {S}tability
  and super-resolution}, \bibinfo{journal}{Applied and Computational Harmonic
  Analysis} \bibinfo{volume}{40}~(\bibinfo{number}{1}) (\bibinfo{year}{2016})
  \bibinfo{pages}{33--67}.

\bibitem[{Potts and Tasche(2017)}]{Potts2017}
\bibinfo{author}{D.~Potts}, \bibinfo{author}{M.~Tasche}, \bibinfo{title}{Error
  estimates for the ESPRIT algorithm}, \bibinfo{journal}{Operator Theory:
  Advances and Applications} \bibinfo{volume}{259} (\bibinfo{year}{2017})
  \bibinfo{pages}{621--648}.

\bibitem[{Tang et~al.(2013)Tang, Bhaskar, and Recht}]{Tang2013}
\bibinfo{author}{G.~Tang}, \bibinfo{author}{B.~N. Bhaskar},
  \bibinfo{author}{B.~Recht}, \bibinfo{title}{Near minimax line spectral
  estimation}, \bibinfo{journal}{IEEE Transactions on Information Theory}
  \bibinfo{volume}{61}~(\bibinfo{number}{23}) (\bibinfo{year}{2013})
  \bibinfo{pages}{5987--5999}.

\bibitem[{Potts and Steidl(2003)}]{Potts2006}
\bibinfo{author}{D.~Potts}, \bibinfo{author}{G.~Steidl}, \bibinfo{title}{Fast
  summation at equispaced knots by {NFFT}}, \bibinfo{journal}{SIAM Journal on
  Scientific Computing} \bibinfo{volume}{24}~(\bibinfo{number}{6})
  (\bibinfo{year}{2003}) \bibinfo{pages}{2013--2037}.

\bibitem[{Kunis(2006)}]{Kunis2006}
\bibinfo{author}{S.~Kunis}, \bibinfo{title}{Nonequispaced {FFT}: Generalisation
  and inversion}, Ph.D. thesis, \bibinfo{school}{University of L\"ubeck},
  \bibinfo{year}{2006}.

\bibitem[{Gautschi and Inglese(1988)}]{Gautschi1988}
\bibinfo{author}{W.~Gautschi}, \bibinfo{author}{G.~Inglese},
  \bibinfo{title}{Lower bounds for the condition number of {V}andermonde
  matrices}, \bibinfo{journal}{Numerische Mathematik}
  \bibinfo{volume}{52}~(\bibinfo{number}{3}) (\bibinfo{year}{1988})
  \bibinfo{pages}{241--250}.

\bibitem[{Beckermann(2000)}]{Beckermann2000}
\bibinfo{author}{B.~Beckermann}, \bibinfo{title}{The condition number of real
  {V}andermonde, {K}rylov and positive definite {H}ankel matrices},
  \bibinfo{journal}{Numerische Mathematik}
  \bibinfo{volume}{85}~(\bibinfo{number}{4}) (\bibinfo{year}{2000})
  \bibinfo{pages}{553--577}.

\bibitem[{C\'ordova et~al.(1990)C\'ordova, Gautschi, and
  Ruscheweyh}]{Gautschi1990}
\bibinfo{author}{A.~C\'ordova}, \bibinfo{author}{W.~Gautschi},
  \bibinfo{author}{S.~Ruscheweyh}, \bibinfo{title}{Vandermonde matrices on the
  circle: Spectral properties and conditioning}, \bibinfo{journal}{Numerische
  Mathematik} \bibinfo{volume}{57}~(\bibinfo{number}{1}) (\bibinfo{year}{1990})
  \bibinfo{pages}{577--591}.

\bibitem[{Berman and Feuer(2007)}]{Berman2007}
\bibinfo{author}{L.~Berman}, \bibinfo{author}{A.~Feuer}, \bibinfo{title}{On
  perfect conditioning of {V}andermonde matrices on the unit circle},
  \bibinfo{journal}{Electronic Journal of Linear Algebra} \bibinfo{volume}{16}
  (\bibinfo{year}{2007}) \bibinfo{pages}{157--161}.

\bibitem[{Ferreira(1999)}]{Ferreira99}
\bibinfo{author}{P.~J. S.~G. Ferreira}, \bibinfo{title}{Superresolution, the
  recovery of missing samples, and {V}andermonde matrices on the unit circle},
  in: \bibinfo{booktitle}{Proceedings of the 1999 Workshop on Sampling Theory
  and Applications}, \bibinfo{address}{Loen, Norway},
  \bibinfo{pages}{216--220}, \bibinfo{year}{1999}.

\bibitem[{Moitra(2014)}]{Moitra2014}
\bibinfo{author}{A.~Moitra}, \bibinfo{title}{The threshold for super-resolution
  via extremal functions}, \bibinfo{journal}{submitted,}
  \urlprefix\url{http://arxiv.org/abs/1408.1681}.

\bibitem[{Linnik(1941)}]{Linnik1941}
\bibinfo{author}{Y.~V. Linnik}, \bibinfo{title}{The large sieve},
  \bibinfo{journal}{Doklady Akademii Nauk SSSR (Proceedings of the USSR Academy
  of Sciences)} \bibinfo{volume}{30}~(\bibinfo{number}{4})
  (\bibinfo{year}{1941}) \bibinfo{pages}{291--294}, \bibinfo{note}{(in
  Russian)}.

\bibitem[{R\'enyi(1951)}]{Renyi1951}
\bibinfo{author}{A.~R\'enyi}, \bibinfo{title}{On the large sieve of {Ju V.
  Linnik}}, \bibinfo{journal}{Compositio Mathematica} \bibinfo{volume}{8}
  (\bibinfo{year}{1951}) \bibinfo{pages}{68--75}.

\bibitem[{Roth(1964)}]{Roth1964}
\bibinfo{author}{K.~F. Roth}, \bibinfo{title}{Remark concerning integer
  sequences}, \bibinfo{journal}{Acta Mathematica}
  \bibinfo{volume}{9}~(\bibinfo{number}{3}) (\bibinfo{year}{1964})
  \bibinfo{pages}{257--260}.

\bibitem[{Roth(1965)}]{Roth1965}
\bibinfo{author}{K.~F. Roth}, \bibinfo{title}{On the large sieve of {L}innik
  and {R}\'enyi}, \bibinfo{journal}{Mathematika} \bibinfo{volume}{12}
  (\bibinfo{year}{1965}) \bibinfo{pages}{1--9}.

\bibitem[{Bombieri(1965)}]{Bombieri1965}
\bibinfo{author}{E.~Bombieri}, \bibinfo{title}{On the large sieve},
  \bibinfo{journal}{Mathematika} \bibinfo{volume}{12}~(\bibinfo{number}{2})
  (\bibinfo{year}{1965}) \bibinfo{pages}{201--225}.

\bibitem[{Montgomery(1968)}]{Montgomery1968}
\bibinfo{author}{H.~L. Montgomery}, \bibinfo{title}{A note on the large sieve},
  \bibinfo{journal}{Journal of the London Mathematical Society}
  \bibinfo{volume}{43} (\bibinfo{year}{1968}) \bibinfo{pages}{93--98}.

\bibitem[{Montgomery and Vaughan(1973)}]{Montgomery1973}
\bibinfo{author}{H.~L. Montgomery}, \bibinfo{author}{R.~C. Vaughan},
  \bibinfo{title}{The large sieve}, \bibinfo{journal}{Mathematika}
  \bibinfo{volume}{20}~(\bibinfo{number}{2}) (\bibinfo{year}{1973})
  \bibinfo{pages}{119--134}.

\bibitem[{Montgomery(1978)}]{Montgomery1978}
\bibinfo{author}{H.~L. Montgomery}, \bibinfo{title}{The analytic principle of
  the large sieve}, \bibinfo{journal}{Bulletin of the American Mathematical
  Society} \bibinfo{volume}{84}~(\bibinfo{number}{4}) (\bibinfo{year}{1978})
  \bibinfo{pages}{547--567}.

\bibitem[{Golub and {Van Loan}(1996)}]{Golub1983}
\bibinfo{author}{G.~H. Golub}, \bibinfo{author}{C.~F. {Van Loan}},
  \bibinfo{title}{Matrix computations}, \bibinfo{publisher}{The Johns Hopkins
  University Press}, \bibinfo{address}{Baltimore, MD, USA},
  \bibinfo{year}{1996}.

\bibitem[{Gautschi(1962)}]{Gautschi1962}
\bibinfo{author}{W.~Gautschi}, \bibinfo{title}{On inverses of {V}andermonde and
  confluent {V}andermonde matrices}, \bibinfo{journal}{Numerische Mathematik}
  \bibinfo{volume}{4}~(\bibinfo{number}{1}) (\bibinfo{year}{1962})
  \bibinfo{pages}{117--123}.

\bibitem[{Faure et~al.(2015)Faure, Kritzer, and Pillichshammer}]{Faure2015}
\bibinfo{author}{H.~Faure}, \bibinfo{author}{P.~Kritzer},
  \bibinfo{author}{F.~Pillichshammer}, \bibinfo{title}{From van der {C}orput to
  modern constructions of sequences for quasi-{M}onte {C}arlo rules},
  \bibinfo{journal}{Indagationes Mathematicae}
  \bibinfo{volume}{26}~(\bibinfo{number}{5}) (\bibinfo{year}{2015})
  \bibinfo{pages}{760--822}.

\bibitem[{Nagell(1951)}]{Nagell1951}
\bibinfo{author}{T.~Nagell}, \bibinfo{title}{Introduction to number theory},
  \bibinfo{publisher}{John Wiley \& Sons}, \bibinfo{address}{New York, NY,
  USA}, \bibinfo{year}{1951}.

\bibitem[{Horn and Johnson(1985)}]{Horn1985}
\bibinfo{author}{R.~A. Horn}, \bibinfo{author}{C.~R. Johnson},
  \bibinfo{title}{Matrix analysis}, \bibinfo{publisher}{Cambridge University
  Press}, \bibinfo{address}{New York, NY, USA}, \bibinfo{year}{1985}.

\bibitem[{Negreanu and Zuazua(2006)}]{Negreanu2006}
\bibinfo{author}{M.~Negreanu}, \bibinfo{author}{E.~Zuazua},
  \bibinfo{title}{Discrete {I}ngham inequalities and applications},
  \bibinfo{journal}{SIAM Journal on Numerical Analysis}
  \bibinfo{volume}{44}~(\bibinfo{number}{1}) (\bibinfo{year}{2006})
  \bibinfo{pages}{412--448}.

\bibitem[{Ingham(1936)}]{Ingham1936}
\bibinfo{author}{A.~E. Ingham}, \bibinfo{title}{Some trigonometrical
  inequalities with applications to the theory of series},
  \bibinfo{journal}{Mathematische Zeitschrift}
  \bibinfo{volume}{41}~(\bibinfo{number}{1}) (\bibinfo{year}{1936})
  \bibinfo{pages}{367--379}.

\bibitem[{Bombieri(1974)}]{Bombiri1974}
\bibinfo{author}{E.~Bombieri}, \bibinfo{title}{Le grand crible dans la
  th\'eorie analytique des nombres}, \bibinfo{publisher}{Soci\'et\'e
  Math\'ematique de {France}}, \bibinfo{year}{1974}.

\bibitem[{Montgomery(2001)}]{Montgomery2001}
\bibinfo{author}{H.~L. Montgomery}, \bibinfo{title}{Twentieth century harmonic
  analysis -- {A} celebration}, in: \bibinfo{booktitle}{Harmonic Analysis as
  found in Analytic Number Theory}, vol.~\bibinfo{volume}{33} of
  \emph{\bibinfo{series}{NATO Science Series}}, \bibinfo{publisher}{Kluwer
  Academic Publishers}, \bibinfo{pages}{271--293}, \bibinfo{year}{2001}.

\bibitem[{Davenport and Halberstam(1966)}]{Davenport1966}
\bibinfo{author}{H.~Davenport}, \bibinfo{author}{H.~Halberstam},
  \bibinfo{title}{The values of a trigonometric polynomial at well spaced
  points}, \bibinfo{journal}{Mathematika}
  \bibinfo{volume}{13}~(\bibinfo{number}{1}) (\bibinfo{year}{1966})
  \bibinfo{pages}{91--96}.

\bibitem[{Gallagher(1967)}]{Gallagher1967}
\bibinfo{author}{P.~X. Gallagher}, \bibinfo{title}{The large sieve},
  \bibinfo{journal}{Mathematika} \bibinfo{volume}{14}~(\bibinfo{number}{1})
  (\bibinfo{year}{1967}) \bibinfo{pages}{14--20}.

\bibitem[{Liu(1969)}]{Liu1969}
\bibinfo{author}{M.-C. Liu}, \bibinfo{title}{On a result of {D}avenport and
  {H}alberstam}, \bibinfo{journal}{Journal of Number Theory}
  \bibinfo{volume}{1}~(\bibinfo{number}{4}) (\bibinfo{year}{1969})
  \bibinfo{pages}{385--389}.

\bibitem[{Bombieri and Davenport(1968)}]{Bombieri1968}
\bibinfo{author}{E.~Bombieri}, \bibinfo{author}{H.~Davenport},
  \bibinfo{title}{On the large sieve method}, \bibinfo{journal}{Abhandlungen
  aus Zahlentheorie und Analysis zur Erinnerung an Edmund Landau}
  \bibinfo{volume}{14} (\bibinfo{year}{1968}) \bibinfo{pages}{14--20}.

\bibitem[{Bombieri and Davenport(1969)}]{Bombieri1969}
\bibinfo{author}{E.~Bombieri}, \bibinfo{author}{H.~Davenport},
  \bibinfo{title}{Some inequalities involving trigonometric polynomials},
  \bibinfo{journal}{Annali della Scuola Normale Superiore di Pisa}
  \bibinfo{volume}{23}~(\bibinfo{number}{3}) (\bibinfo{year}{1969})
  \bibinfo{pages}{223--241}.

\bibitem[{Boas(1942)}]{Boas1942}
\bibinfo{author}{R.~P. Boas}, \bibinfo{title}{Entire functions of exponential
  type}, \bibinfo{journal}{Bulletin of the American Mathematical Society}
  \bibinfo{volume}{48}~(\bibinfo{number}{12}) (\bibinfo{year}{1942})
  \bibinfo{pages}{839--849}.

\bibitem[{Beurling(1989)}]{Beurling1989}
\bibinfo{author}{A.~Beurling}, \bibinfo{title}{On functions with a spectral
  gap}, in: \bibinfo{editor}{L.~Carleson}, \bibinfo{editor}{P.~Malliavin},
  \bibinfo{editor}{J.~Neuberger}, \bibinfo{editor}{J.~Werner} (Eds.),
  \bibinfo{booktitle}{The Collected Works of Arne Beurling: Volume 2, Harmonic
  Analysis}, \bibinfo{publisher}{Birkh\"auser}, \bibinfo{address}{Boston, MA,
  USA}, \bibinfo{pages}{370--372}, \bibinfo{year}{1989}.

\bibitem[{Young(1990)}]{Young1990}
\bibinfo{author}{R.~M. Young}, \bibinfo{title}{An introduction to nonharmonic
  {F}ourier series}, in: \bibinfo{booktitle}{Chap.~2: {E}ntire functions of
  exponential type}, \bibinfo{publisher}{Academic Press}, \bibinfo{address}{New
  York, NY, USA}, \bibinfo{year}{1990}.

\bibitem[{Montgomery and Vaughan(1974)}]{Montgomery1974}
\bibinfo{author}{H.~L. Montgomery}, \bibinfo{author}{R.~C. Vaughan},
  \bibinfo{title}{Hilbert's inequality}, \bibinfo{journal}{Journal of the
  London Mathematical Society} \bibinfo{volume}{8}~(\bibinfo{number}{2})
  (\bibinfo{year}{1974}) \bibinfo{pages}{73--82}.

\bibitem[{Schur(1911)}]{Schur1911}
\bibinfo{author}{I.~Schur}, \bibinfo{title}{Bemerkungen zur {T}heorie der
  beschr\"ankten {B}ilinearformen mit unendlich vielen {V}er\"anderlichen},
  \bibinfo{journal}{Journal f\"ur die reine und angewandte Mathematik}
  \bibinfo{volume}{140} (\bibinfo{year}{1911}) \bibinfo{pages}{1--28}.

\bibitem[{Graham and Vaaler(1981)}]{Graham1981}
\bibinfo{author}{S.~W. Graham}, \bibinfo{author}{J.~D. Vaaler},
  \bibinfo{title}{A class of extremal functions for the {F}ourier transform},
  \bibinfo{journal}{Transactions of the American Mathematical Society}
  \bibinfo{volume}{265}~(\bibinfo{number}{1}) (\bibinfo{year}{1981})
  \bibinfo{pages}{283--302}.

\bibitem[{Montgomery and Vaaler(1999)}]{Montgomery1998}
\bibinfo{author}{H.~L. Montgomery}, \bibinfo{author}{J.~D. Vaaler},
  \bibinfo{title}{A further generalization of {H}ilbert's inequality},
  \bibinfo{journal}{Mathematica} \bibinfo{volume}{46}~(\bibinfo{number}{1})
  (\bibinfo{year}{1999}) \bibinfo{pages}{35--39}.

\bibitem[{Jameson(2006)}]{Jameson2006}
\bibinfo{author}{G.~Jameson}, \bibinfo{title}{Notes on the large sieve},
  \bibinfo{type}{Tech. Rep.}, \bibinfo{institution}{Department of Mathematics
  and Statistics, Lancaster University},
  \urlprefix\url{http://www.maths.lancs.ac.uk/~jameson/lsv.pdf},
  \bibinfo{year}{2006}.

\bibitem[{Zygmund(2002)}]{Zygmund2002}
\bibinfo{author}{A.~Zygmund}, \bibinfo{title}{Trigonometric series -- Third
  edition}, \bibinfo{publisher}{Cambridge University Press},
  \bibinfo{address}{London, United Kingdom}, \bibinfo{year}{2002}.

\end{thebibliography}
